\documentclass
[11pt, letterpaper]{article}

\usepackage{macros}
\usepackage{amsthm}
\usepackage{thmtools}
\usepackage{turnstile}
\usepackage{mdframed}
\usepackage{hyperref}
\usepackage{cleveref}
\renewcommand{\edgeco}{\left(-\sqrt{\frac{p}{1-p}}\right)}

\title{Sum-of-Squares Lower Bounds for Independent Set in Ultra-Sparse Random Graphs}

\date{\today}
\begin{document}
% \author{Pravesh K. Kothari\thanks{Princeton University, \texttt{kothari@cs.princeton.edu}. Supported by NSF CAREER Award #2047933, NSF #2211971, an Alfred P. Sloan Fellow-
% ship, and a Google Research Scholar Award.} 
%  \and Aaron Potechin\thanks{The University of Chicago, \texttt{potechin@uchicago.edu}. Supported in part by NSF grant CCF-2008920. }
%  \and
%  Jeff Xu\thanks{Carnegie Mellon University, \texttt{jeffxusichao@cmu.edu}. Supported in part by NSF CAREER Award \#2047933, and work done while supported by a Cylab Presidential Fellowship.}}
\author{
 Pravesh K. Kothari\thanks{Princeton University, \texttt{kothari@cs.princeton.edu}. Supported by  NSF CAREER Award \#2047933, Alfred P. Sloan Fellowship and a Google Research Scholar Award.}
 \and Aaron Potechin\thanks{The University of Chicago, \texttt{potechin@uchicago.edu}. Supported in part by NSF grant CCF-2008920. }
 \and
 Jeff Xu\thanks{Carnegie Mellon University, \texttt{jeffxusichao@cmu.edu}. Supported in part by NSF CAREER Award \#2047933, and a Cylab Presidential Fellowship.}}
\maketitle

\begin{abstract}
%	 Constant degree lower bound is interesting?
We prove that for every $D \in \N$, and large enough constant $d \in \N$, with high probability over the choice of $G \sim G(n,d/n)$, the \Erdos-\Renyi random graph distribution, the canonical degree $2D$ Sum-of-Squares relaxation fails to certify that the largest independent set in $G$ is of size $o(\frac{n}{\sqrt{d} D^4})$. In particular, degree $D$ sum-of-squares strengthening can reduce the integrality gap of the classical \Lovasz theta SDP relaxation by at most a $O(D^4)$ factor. 
  
This is the first lower bound for $>4$-degree Sum-of-Squares (SoS) relaxation for any problems on \emph{ultra sparse} random graphs (i.e. average degree of an absolute constant). Such ultra-sparse graphs were a known barrier for previous methods and explicitly identified as a major open direction (e.g.,~\cite{deshpande2019threshold, kothari2021stressfree}). Indeed, the only other example of an SoS lower bound on ultra-sparse random graphs was a degree-4 lower bound for Max-Cut.

Our main technical result is a new method to obtain spectral norm estimates on graph matrices  (a class of low-degree matrix-valued polynomials in $G(n,d/n)$) that are accurate to within an absolute constant factor. All prior works lose $\poly log n$ factors that trivialize any lower bound on $o(\log n)$-degree random graphs. We combine these new bounds with several upgrades on the machinery for analyzing lower-bound witnesses constructed by pseudo-calibration so that our analysis does not lose any $\omega(1)$-factors that would trivialize our results.  In addition to other SoS lower bounds, we believe that our methods for establishing spectral norm estimates on graph matrices will be useful in the analyses of numerical algorithms on average-case inputs.

% 	 For an Erdos-Renyi random graphs of large bounded average degree $d=O(1)$, it is well-known that spectral techniques can certify its independent set is of size at most $O(\frac{n}{\sqrt{d}})$ whereas the true solution is of size $\frac{2n\log d}{d}$. 
%	  Following the recent surge of progress that sketches out a rough picture for the performance of SoS algorithms on average-case problems, the fine-grained understanding of SoS in the ultra-sparse regime has become a forbidding while central question for investigating information v. computation gap. 
% 	  We prove this gap persists for higher constant degree Sum-of-Squares algorithms, even with SoS degree $\mathsf{dsos}$ growing with $d$. Concretely, SoS of degree-$\dsos$ fails to refute the existence of independent set of size $\Omega( \frac{n}{\sqrt{d} \cdot  \mathsf{dsos}^2} )$ where $\Omega$ hides some absolute constant. 
	 	 
% 	  This is the first result against higher-degree SoS algorithms for problems on sparse graphs in the regime of our interest. Previous results for independent set and coloring were only known at degree-$2$, and more broadly, the prior best hardness for refutation problems on sparse random graphs of bounded average degree was a degree-$4$ lower bound for max-cut.
	  \end{abstract}\clearpage\newpage
\pagenumbering{roman}
\setcounter{tocdepth}{2}
\tableofcontents	
\newpage
\pagenumbering{arabic}
\setcounter{page}{1}
	\section{Introduction}
%\subsection{Success of spectral techniques}
% Given a graph $G =(V,E)$ sampled from Erd\H{o}s-Renyi distribution with average degree $d$, what is the best upper bound we can give efficiently on its max independent size (i.e. independence number)? A closely related question is to ask for the chromatic number: at least how many colors are needed to color the whole graph? These are both average-case variants of two of the canonical NP-hard problems, and in the worst case, the task of remotely approximating them within a polynomial factor is known to be hard unless $\mathsf{P} = \NP$. 

An Erd\H{o}s-Renyi random graph with average degree $d$ has an independence number of  $(1+o_n(1)) \frac{2n\log d}{d}$ and a chromatic number of $(1+o_n(1)) \frac{d}{2 \log d}$ with high probability~\cite{COE15, DM11, DSS16}. There is a long line of work focused on finding the best \emph{efficiently certifiable} bounds on these quantities. When the average-degree of the graphs  $d = \Omega(\poly(\log n))$, the spectral norm of the (centered) adjacency matrix \cite{Feige2005SpectralTA} certifies an upper bound on the independence number of $O(\frac{n}{\sqrt{d}})$. These bounds are off from the ``ground-truth" by a factor of $\frac{1}{\sqrt{d}}$. Upgrading eigenvalue based certificates to those based on the ``basic" semidefinite programming relaxation yields no asymptotic improvement~\cite{CO05}. Investigating whether efficient certificates that improve on the above bounds exist has been a longstanding and major open question.

% \pnote{directly start a discussion of sum-of-squares lower bounds here. A two line intro saying that 1) it's the strongest known SDP hierarchy and has been responsible for several average-case algorithms is enough. Most of 1.1 should be cut I think. }

SDP hierarchies are natural candidates for improving on the certificates based on the basic SDP. The sum-of-squares hierarchy of SDP relaxations~\cite{Lasserre:2000:GOP:588888.589038,parrilo2000structured} is the strongest such hierarchy considered in the literature. Sum-of-squares captures the best known algorithms for many worst-case combinatorial optimization problems \cite{GW94:stoc, AroraRV04} and in the past decade, has been responsible for substantially improved and often optimal algorithms for foundational average-case problems \cite{BRS11, GS11, HSS15, MSS16, RRS17, KSS18, HL18, klivans2020efficient, hopkins2019mean, BK20}. As a result, establishing sum-of-squares lower bounds for basic average-case graph problems has been a major research direction in the past decade and a half. 

%paragraph about pseudocalibration 
Following on the early works in proof complexity, researchers were able to obtain essentially optimal sum-of-squares lower bounds for random constraint satisfaction problems \cite{Grigoriev01, Schoenebeck08, BCK15, KMOW17}. However,  going beyond constraint satisfaction required significantly new ideas. In particular, a sequence of papers in the last decade made progress on lower bounds for the planted clique problems \cite{FK02,DM15, DBLP:conf/stoc/MekaPW15} before it was resolved in~\cite{BHKKMP19} via the new method called \emph{pseudo-calibration}. Since then, there has been considerable progress in obtaining strong lower bounds for foundational average-case problems such as Sparse PCA, densest k-subgraph, and the Sherrington-Kirkpatrick problem~\cite{HKPRSS17, KuniskyBandeira19, MRX20, GJJPR20, kothari2021stressfree,JPRTX, hk21, jones2023sumofsquares}.

% \cite{deshpande2019threshold, BMR19, KuniskyBandeira19, GJJPR20, Kunisky20, potechin2023machinery, kothari2021stressfree, JPRTX, jones2023sumofsquares}

%paragraph about limitations of pseudo-calibration and hardness in the sparse regime

\paragraph{Pseudo-calibration and Spectral Norms of Graph Matrices} Analyzing lower bound witnesses constructed via pseudo-calibration crucially relies on decomposing certain correlated random matrices into a ``Fourier-like" basis of \emph{graph matrices} and understanding bounds on their spectral norms. Each entry of such graph matrices is a polynomial in the underlying edge-indicator variables of the input random graph. There has been considerable progress in the past few years in understanding the spectra of such matrices when the input is a \emph{dense} random graph~\cite{AMP20, cai2020spectrum,JPRTX, GT23}. However, when the input is an ``ultra-sparse" random graph with a constant average degree, known tools turn out to be rather blunt. At a high level, given the rather complicated structure of graph matrices, prior methods only yield a  \emph{coarse} bound on their spectral norms. Such coarse bounds nevertheless suffice for analyzing the lower bound witnesses for problems on dense random graphs (with some polylogarithmic factor losses). On ultra-sparse graphs, however, it turns out to be an absolute non-starter. 

Indeed, the limited progress seen so far in the ultra-sparse regime has come about via rather ad hoc techniques. For the ``basic SDP" (i.e., degree $2$ sum-of-squares), prior works\cite{deshpande2019threshold, BKM19, BMR19} relied on the fact that the spectrum of the underlying matrices (that turns out to have essentially independent entries)  can be completely understood via some powerful tools from random matrix theory. The work of \cite{MRX20} is the only one to go beyond the basic SDP and obtain degree four sum-of-squares lower bounds for the Sherrington-Kirkpatrick problem (essentially Max-Cut with independent Gaussian weights) on ultra sparse random graphs. They accomplished this via a certain ``lifting" technique that circumvented pseudo-calibration. Their technique, however, is unwieldy to even extend to degree 6 sum-of-squares relaxation or problems with ``hard" constraints (such as independent set, the focus of the present work).

\paragraph{A curse of sparsity?} More generally, the ultra-sparse regime has been a challenge for understanding algorithmic problems on graphs for a long time. For example, it took a long sequence of works and some sophisticated tools in random matrix theory (e.g., the Ihara-Bass formula and spectral understanding of the \emph{non-backtracking walk matrix} of random graphs) to resolve the algorithmic threshold for \emph{community detection} on the stochastic block model~\cite{abbe2014exact, abbesandon15, abbe17, Abbe2018ProofOT, BKM19, BMR19}. %a prototypical problem for statistics versus computation trade-off, actually holds in the bounded degree regime \cite{specredemp}. Even before that, this issue was also already manifest when researchers studied spectral techniques based on (centered) adjacency matrix to certify independence number in this regime \cite{CO05}.

One basic issue that makes the ultra-sparse setting challenging is that the spectrum of random matrices that arise in such settings suffers from a large variance that makes the analysis tricky. For example, for $d=\omega(1)$, the  \Erdos-\Renyi random graph is almost regular (deviations in degree are negligible compared to the average). When $d=O(1)$, however, this is no longer true: there are bound to be vertices of super-constant degree $\Omega(\sqrt{\log n})$ on the one end and vertices of degree $0$ (i.e. isolated vertices) on the other. A trickier issue arises from events with small but non-negligible probability that makes the analysis based on the moment method (the workhorse in such analyses)  tricky. For example, it can be shown that for the adjacency matrix $A$ of an \Erdos-\Renyi graph, even after removing rows and columns corresponding to high degree vertices, for large $k$,  $\E\left[\tr\left((A - \E[A])^k\right)\right]$ is much larger than its typical value due to an inverse polynomial chance of the presence of dense subgraphs.%even after removing high-degree vertices, the adjacency matrix of an ultra-sparse random graph has a large high-trace in expectation due to an inverse polynomial chance of the presence of dense subgraphs.

While such issues have been tackled in the past, with significant effort \cite{MNS18, specredemp, Mas14, BLM15, Mossel2015, HS17, LMR21,ding2021robust,Zdeborov__2016, Krzakala_2009, PhysRevE.84.066106, PhysRevLett.107.065701} for analyzing the spectra of \emph{basic} matrices associated with ultra sparse random graphs (e.g., adjacency matrices, non-backtracking walk matrices), in our setting the problem is significantly more involved because of our need to understand substantially more complicated \emph{graph matrices} built from sparse random graphs --- matrices of polynomial dimension with entries that are %symmetric 
low-degree polynomials in the entries of the adjacency matrix of $G(n,d/n)$ and which are invariant (as a function of the adjacency matrix of $G(n,d/n)$) under permutations of $[n]$.

\paragraph{This Work: a fine-grained method for sharp bounds on spectra of graph matrices} A key contribution of this work is a new technique to analyze the spectrum of graph matrices built from ultra sparse random graphs via the moment method. A consequence of these techniques is estimates of  spectral norms for  graph matrices that are tight to within an \emph{absolute} constant factor! Even in the case of \emph{dense} random graphs (let alone the significantly more challenging setting of ultra-sparse random graphs), finding techniques that yield sharp bounds on spectral norms of graph matrices was considered significantly challenging --- a main contribution of ~\cite{MRX20} was finding such bounds for a limited subset of graph matrices.

% Despite the above technical challenges, the past two decades of investigation from different areas of computer science, statistics and physics have in fact procured a solid grasp of algorithmic tractabilities in the possible regime. Intriguingly, in the sparse regime, a class of algorithms inspired from statistical physics, belief propagation and its variants, is conjectured to be optimal up-to algorithmic possibility threshold. However, in a sharp contrast to the success in the upper bounds, rigorous evidence of algorithmic hardness for statistical inference problems in the hard regime remains at large \cite{}.
 
%  For a more detailed overview of recent algorithmic progress obtained , 

% As epitomized in the above, the sparsity of random graph poses inevitable technical challenges, and this is partially why recent investigations into Sum-of-Squares lower bounds have steered cleared of this land. However, as this is the regime where the sharpest algorithmic phase transitions occur, it is also rich of opportunities and mysteries towards rigorously establishing various hypothesis that have emerged, to some extent independently, from different areas of computer science, statistics and physics.
% And it goes without saying for enthusiasts about high-dimensional statistical inference that sparse random graph shall be 
\subsection{Our results}
% We make the first attempt to unify the previously disjoint lines of work in Sum-of-Squares lower bounds that we survey above, and show that higher-degree Sum-of-Squares, cannot \emph{significantly} improve upon the basic spectral techniques for the refutation question of independent set and coloring when the graph comes from $G_{n,d/n}$.

%To state our main result, it is convenient to recall that SoS is a family of increasingly powerful algorithms parameterized by its degree for polynomial optimization problems. 
% In this work, we consider the SoS relaxation for max independent set, and $k$-(multi)- coloring without the constraint that each vertex gets one color. With the duality between SoS and the pseudo-expectation operator, we adopt the pseudo-expectation perspective and to show degree-$\dsos$ SoS program fails, it suffices for us to construct a degree-$\dsos$  pseudo-expectation operator that satisfies the corresponding constraints. For a formal definition of SoS, we defer to [link]. We are now ready to state our main results.
In this work, we prove that for every $d \in \N$, constant degree sum-of-squares strengthening of the independent set axioms fail to certify a bound of $o(n/\sqrt{d})$ on the maximum independent set with high probability when $G \sim G(n,d/n)$. To formulate this precisely, let us recall the notion of pseudo-expectations consistent with a set of polynomial equations:

\begin{definition}[Pseudo-expectation of degree-$\dsos$]
    For any $d_{sos} \in \N$, a degree $\dsos$-pseudoexpectation in variables $x = (x_1, x_2,\ldots, x_n)$ (denoted by $\pE_{\dsos}$) is a linear map that assigns a real number to every polynomial $f$ of degree $\leq d$ in $x$ and satisfies: 1) \textbf{normalization:} $\pE[1] = 1$, and, 2) \textbf{positivity:} $\pE[f^2] \geq  0$ for every polynomial $f$ with degree at most $\frac{\dsos}{2}$. For any polynomial $g(x)$, a pseudo-expectation satisfies a constraint $g = 0$ if $\E[f\cdot g] = 0$  for all polynomials $f$ of degree at most $d_{sos} -deg(g)$.
\end{definition}

We can now describe the sum-of-squares relaxation for independent set in a graph $G$. 

\begin{definition}[Independent Set Axioms] \label{def:ind-set-axioms}
Let $G$ be a $n$-vertex graph. The following axioms describe the $0$-$1$ indicators $x$ of independent sets in $G$: 
\[
\forall v\in V, x_v^2 = x_v  \quad{\text{(Booleanity) } } \,; \]
\[ \forall (u,v)\in E, x_ux_v = 0 \quad{\text{(Independent set)} }
 \]
The degree $d_{sos}$-sum-of-squares relaxation for independent set in $G$ maximizes $\pE[ \sum_i x_i]$ over all pseudo-expectations of degree $d_{sos}$ satisfying the above two axioms. 
\end{definition}

More precisely, we prove:

\begin{restatable}[Main result]{theorem}{main-indset} \label{thm:main-indset}
    There is a $c>0$ such that for every $d \in \N$, with probability $1-o_n(1)$ over $G\sim G_{n,d/n}$, there exists a degree-$\dsos$ pseudo-expectation satisfying the independent set axioms (Definition~\ref{def:ind-set-axioms}) and
    \[ 
\sum_{i\in [n]} \pE[x_i] \geq  (1-o(1)) \frac{n}{c \cdot \sqrt{d} \cdot  \dsos^4}\,.
\] 
\end{restatable}

% We defer the formal statements to \cref{thm:main-indset,thm:main-coloring} after we introduce some required background. 
% \begin{restatable}[Main theorem for (multi)-$k$-coloring]{theorem}{mainthmcol} \label{thm:main-coloring}
% There is a large constant $d_0$, and an absolute constant $c_\eta$, such that for any sufficiently large $n\in \N$, for any $\dsos\in \N$, it holds with high probability (at least $1-o_n(1)$) over $G=(V,E)\sim G_{n,d/n}$ for $d>d_0$, there exists a degree-$\dsos$ pseudo-expectation satisfying
% \[
% \forall v\in V, x_v^2 = x_v  \quad{\text{(Booleanity) } } \,; \]
% \[ \forall (u,v)\in E, \forall c\in [k], x_{u,c}x_{v,c} = 0 \quad{\text{(Independent set)} }\,;
%  \]\[ 
%  \forall v\in V, \forall c\neq c' \in[k], x_{v,c} x_{v,c'} = 0 \quad{\text{(Proper coloring) } }  \]
% \
% and the number of colors $k = O( \sqrt{d} \cdot \dsos^4  ) $.
%   \end{restatable}

\begin{remark}
We note that our techniques likely allow improving the dependence on $d_{sos}$ in the denominator to a linear (instead of quartic) bound. We believe that removing this dependence altogether is possible but requires new ideas. 
\end{remark}
Our construction of the lower bound witness (aka pseudo-moment matrix) is based on pseudo-calibration. Informally speaking, pseudo-calibration provides a guess for the pseudo-expectation with a certain ``truncated" probability density function of a planted distribution (a distribution over graphs containing a large independent set) that is indistinguishable for low-degree polynomials from $G(n,d/n)$. However, as observed in prior works~\cite{JPRTX}, natural planted distributions are \emph{not} low-degree indistinguishable from $G(n,d/n)$ when $d =O(1)$. Nevertheless, it turns out that one can use a certain ad hoc truncation that drops certain carefully chosen terms to obtain a pseudo-expectation that allows us to prove~\cref{thm:main-indset} above.  Our truncation strategy is an upgraded variant of the \emph{connected truncation} first utilized in \cite{JPRTX}. 

In the ultra-sparse regime, our analysis needs to separately consider the vertices that have too high a degree. This trimming of the graph inevitably introduces non-trivial correlations in the graph matrices -- a significant challenge -- that we show how to overcome in our analysis. 

\paragraph{Sharp bounds on spectral norms of Graph Matrices} Our main technical contribution is a new technique that yields sharp bounds on the spectral norms of \emph{graph matrices} --- low-degree matrix valued polynomials in the edge indicator variables of $G \sim G(n,d/n)$. 
%Graph matrices can be thought of as significant generalizations of the adjacency matrix of $G \sim G(n,d/n)$ with highly correlated entries described by low-degree polynomials in the edge indicator variables of $G \sim G(n,d/n)$ (e.g., the adjacency matrix naturally corresponds to an edge in this view). 
The low-degree polynomials themselves are naturally described by graphs. All prior works beginning with ~\cite{BHKKMP16} rely on elegant statements that relate the spectral norms of such graph matrices to natural combinatorial quantities associated with the graphs. However, the precise bounds they achieve turn out to be rather coarse and lose $\poly \log n$ factors. Such losses still turn out to give non-trivial results for problems on dense random graphs. However, they trivialize in the setting of random graphs with average degree $\ll \log n$ --- our principal interest in this work. While this might appear to be a technical issue, finding methods that do not lose such $\poly \log n$ factors was understood to be a major bottleneck in the area. Indeed,  even resolving the case of polylog n degree random graphs took several new ideas in the recent work~\cite{JPRTX}. 

In this work, we finally build methods that overcome the bottlenecks in the prior works and obtain bounds that sharp \emph{up to absolute constants} on the spectral norms of graph matrices. At a high-level, the trace moment method for bounding the spectral norm of a matrix relies on counting the contributions of closed walks on the entries of the matrix. A key new high-level idea in our analysis is a certain \emph{localization} of the walk that allows understanding the contribution of a walk from a single step. This localization allows us to obtain a tractable method to bound total contributions to the weighted count of closed walks with negligible losses.  %The reduction to a single step 
%serves as the backbone of our argument as it allows us to exploit various trade-offs between combinatorial and analytical factors in obtaining a tight norm bound estimate. 

%A few lines about what the key new ideas in the spectral norm analysis might be can go here.

\paragraph{Upgrading the decompose-and-recurse machinery for analyzing pseudo-expedctations} Our analysis of the pseudo-moment matrix we construct requires a substantial upgrade of the machinery for establishing sum-of-squares lower bounds~\cite{BHKKMP19,GJJPR20,potechin2023machinery} via pseudo-calibration. In particular, the strategy in prior works involves ``charging" graph matrices of various shapes (see the technical overview for a more detailed exposition) that arise in the decomposition to positive semidefinite (PSD) terms in the decomposition. Such a charging argument requires a careful count of the terms charged to each PSD term. In the coarse-grained analysis a bound on such a count that is tight up to an exponential in the size of the shape defining the graph matrix suffices. Such an analysis is one of the reasons that even the tightest previous analyses~\cite{JPRTX} loses $polylog(n)$ factors that trivialize the final bounds in the ultra sparse regime, even combined with our tight norm bounds. A key idea that we utilize in this work (that builds on the insight developed from our new techniques in establishing strong spectral norm bounds) involves a careful grouping of terms arising in the graph matrix decomposition of the pseudo-moment matrix so as to avoid the above loss.

% a couple of lines about what might be new here 

% Our main technical contribution is the novel machinery of proving tight spectral norm bounds of \emph{graph matrices}, which combines with our recent progress in the dense-regime for Sum-of-Squares lower bounds via pseudo-calibration framework, opens up the possibility of proving sharp algorithmic phase transitions for a large class of problems beyond the two that we study in this work. We believe this work will serve as a major technical ingredient for future investigation of higher degree Sum-of-Squares lower bounds.

 %However, our work showcases the amenability and versatility of the pseudo-calibration and approximate factorization machinery developed in \cite{BHKKMP19, potechin2023machinery} which was once thought to be "fragile" in front of the obstacles presented by the sparse regime. 

% Beyond the application in Sum-of-Squares lower bounds, recent work \cite{BHKX22} applies graph matrices machinery to investigate stability of numerical algorithms, and we believe our machinery will find further applications in the settin of smoothed analysis.

\paragraph{Organization}
% The work of \cite{kothari2021stressfree} gives the first hardness result for chromatic number in $G_{n,1/2}$ via a reduction from the planted clique problem\cite{BHKKMP19}. In this work, we avoid the reduction step, and achieve this in a single-shot by following the overall recipe of pseudo-calibration, and its application in \cite{JPRTX}. Besides the shared roadmaps evident in the parallel of our main technical lemmas, significant technical departures are needed in proving almost every component due to the challenges of sparsity in our setting as we sketch upon previously.

In ~\pref{sec:technical-overview}, we introduce the preliminaries and notations for our arguments. We showcase the techniques we develop for proving tight norm bounds of graph matrices with formal details in ~\pref{sec: norm-bounds-vertex-encoding}: specifically, we bound the "counting" factor for long walks on graph matrices in ~\pref{sec: norm-bounds-vertex-encoding}, and use it to conclude the final norm bound in ~\pref{sec: maxval-labeling}.

We then apply our norm bounds to PSD analysis of the moment matrix. We describe our moment matrix in ~\pref{sec:moment-matrix-construction} as it is not sufficient to apply pseduo-calibration in a black-box manner due to fluctuation of vertex degree, Moreover, we strengthen the "connected truncation" idea developed in \cite{JPRTX} due to lower-order dependence of norm bounds that can become potentially detrimental in our regime.

Finally, we apply the machinery developed in \cite{JPRTX, potechin2023machinery} with our norm bounds to prove the PSDness of our moment matrix.
% \pravesh{the comments here on relationship to previous works can be weaved into the short description of technical contributions}

\subsection{Proof Plan}
 Pseudo-calibration~\cite{BHKKMP19} provides a general recipe for designing pseudo-expectation as required by Theorem~\ref{thm:main-indset}. As we mentioned before, our analysis involves a modified version of the naive construction given by pseudo-calibration with a certain appropriate truncation (see ~\pref{sec:moment-matrix-construction}). As in the prior works, the main challenge is establishing the positivity property. This property is equivalent to proving the positive semidefiniteness of a certain ``pseudo-moment" matrix $\Lambda$ associated with a pseudo-expectation. As has been the standard approach since~\cite{BHKKMP16}, we proceed to analyze the pseudo-moment matrix $M$ by decomposing it as a linear combination of special ``bases" called graph matrices:
 \[
\Lambda = Id + \sum_{\al \in \tau} \lambda_\al \cdot  M_\al \mper
 \]
 Here, $\tau$ ranges over graph matrices and $\lambda_\al$ are real coefficients on graph matrices $M_\al$ for $\al\in\tau$. In order to understand the eigenvalues of $\Lambda$ and establish positive semidefiniteness, it is natural to understand the spectral properties of the bases $M_{\al}$. In order to focus attention on our main technical contribution, we will dedicate the upcoming overview section to our method for establishing sharp bounds on the spectral norm of such graph matrices. 
 
 % For readers familiar with the analysis for lower bounds against basic SDP for independent set in sparse random graph achieving a value of $\Omega(\frac{n}{\sqrt{d}}) $
 % , the decomposition falls into the above framework as one consider \[ 
 % \Lambda = Id + \lambda_G \cdot M_G
 % \]
 % where $M_G$ is the (normalized) adjacency matrix with some appropriate scaling of $\lambda = O(\frac{1}{\sqrt{d}})$ measuring the size of desired independent-set value. With the above decomposition, it is clear that $\|M_G\| = O(\sqrt{d})$ is the backbone in the analysis for basic SDP. As a teaser, $M_G$ (which we call "line" graph in the following section) will also be a recurrent topic as we illustrate the techniques in our machinery for general matrices. We are now ready to shift gear to more complicated matrices as we reach for higher SoS degree in the hierarchy.

\section{Overview of our sharp spectral bounds on graph matrices} \label{sec:technical-overview}
The main challenge in analyzing graph matrices is that their entries are highly correlated. In particular, an $N\times N$ graph matrix generally has entries that are polynomial functions of $N^{o(1)}$ bits of independent randomness. This makes analyzing the trace powers of graph matrices complicated. As a result, previous analyses of the norms of graph matrices lose poly-logarithmic factors even when the input graph is a dense random graph. Our main contribution is developing a new method for conducting such an analysis that somewhat surprisingly yields estimates that are sharp up to absolute constant factors.
%Graph matrices are correlated random matrices with symmetries that naturally arise as basis functions in a Fourier-like decomposition of complicated random matrices. The entries of such matrices are low-degree polynomials in some underlying independent randomness. In particular, an $N\times N$ graph matrix has entries that are polynomial functions of $N^{o(1)}$ bits of independent randomness. These rather extreme correlations make analyzing their spectral norms challenging. Our main contribution is developing a new method for conducting such an analysis that somewhat surprisingly yields estimates that are sharp up to absolute constant factors. All prior works inherently lose poly-logarithmic factors even when the underlying independent randomness comes from a dense random graph. Such a loss would trivialize the application we are after and was the main bottleneck in obtaining the sum-of-squares lower bounds for ultra-sparse graph problems.  

To describe our new techniques for obtaining the above improved spectral norm bounds, we will focus this overview on two specific and simple examples of graph matrices %only and then towards the end 
and then discuss how our ideas extend more generally. Let's start with the formal definition of graph matrices.

%text here
\begin{definition}[Fourier character for $G_{n,d/n}$]
Let $\chi$ denote the $p$-biased Fourier character, \[ 
\chi_G(1) = \sqrt{\frac{1-p}{p}} =(1+o_n(1)) \sqrt{\frac{n}{d}}, \quad \chi_G(0) = -\sqrt{\frac{p}{1-p}} = - (1+o_n(1))  \sqrt{\frac{d}{n}}
\]    
For a subset of edges $H \subseteq \binom{n}{2}$, we write $\chi_G(H) = \prod_{e\in H} \chi_G(e)$.
\end{definition}

\begin{definition}[Shape]
\label{def:shape}
    A shape $\alpha$ is a graph on vertices $V(\alpha)$ with edges $E(\alpha)$ and %two specially identified subsets of vertices 
    two ordered tuples of vertices $U_{\alpha}$ (left boundary) and $V_{\alpha}$ (right-boundary) %viewed as order-tuples. 
    We denote a shape $\alpha$ by $(V(\alpha), E(\alpha), U_{\alpha}, V_{\alpha})$. 
    
    % with two specially identified subsets of vertices 
    % A shape is an equivalence class of ribbons with the same shape. Each shape 
    % has associated with it
    % a representative $\alpha = (V(\alpha), E(\alpha), U_\alpha, V_\alpha)$, where $U_\alpha,V_\alpha \subseteq V(\alpha)$.
    % Let $W_\alpha \defeq V(\alpha) \setminus (U_\alpha \cup V_\alpha).$
\end{definition}
\begin{definition}[Shape transpose]
    For each shape $\al$, we use $\al^T$ to denote the shape obtained by flipping the boundary $U_\al$ and $V_\al $ labels. In other words, $\al^T = (V(\al), E(\al), U_{\al^T} = V_\al, V_{\al^T} = U_\al)$.
\end{definition}

 \begin{definition}[Embedding]
Given an underlying random graph sample $G$,  a shape $\alpha$ and an injective function $\psi(\al): V(\alpha) \to [n],$ we define $M_{\psi(\al)}$ to be the matrix of size
$\frac{n!}{(n - |U_{\alpha}|)!} \times \frac{n!}{(n - |V_{\alpha}|)!}$
%$n^{\dsos} \times n^{\dsos}$ 
with rows and columns indexed by ordered tuples of $[n]$ of size $|U_{\alpha}|$ and $|V_{\alpha}|$ with a single nonzero entry \[
      M_{\psi(\alpha)}[\psi(U_\al), \psi(V_\al)] = \prod_{(i,j)\in E(\al)}  \chi_G\left(\psi(i), \psi(j)\right).
     \]
     and $0$ everywhere else.
     \end{definition}

%just observe that the below formula in terms of $\alpha$ holds for the two examples above. 
\begin{definition}[Graph matrix of a shape]
    For a shape $\alpha$, the graph matrix $M_\alpha$ is
    \[M_\alpha = \displaystyle\sum_{\text{injective }\psi: V(\alpha) \to [n]} M_{\psi(\alpha)}. \]

    When analyzing the sum of squares hierarchy, we extend $M_{\alpha}$ to have rows and columns indexed by all tuples of vertices of size at most $\frac{\dsos}{2}$ by filling in the remaining entries with $0$.
\end{definition}

\begin{example}[Line-graph graph matrix $M_{line}$]
Define $M_{line} \in \R^{n \times n}$ to be a matrix with zeros on the diagonal and off-diagonal entries $M_{line}[u,v] = \chi_G(u,v)$.
\end{example}
$M_{line}$ is, up to rescaling, the %normalized 
centered 
adjacency matrix of a random graph from $G_{n,d/n}$. We will also use the following $Z$-shape matrix as a slightly complicated example in our discussion in this section. % as our guide to more non-trivial graph matrices.
\begin{example}[ $Z$-shape graph matrix $M_Z$]
Define $M_{Z} \in \R^{n(n-1) \times n(n-1) }$ by $M_Z[(i,j), (k,\ell)] = \chi_G(i,k)\chi_G(j,k)\chi_G(j,\ell)$ whenever $i,j,k,\ell$ are distinct and $0$ otherwise.
\end{example}
 Observe that $M_Z$ has $n^2(n-1)^2$  entries that are low-degree polynomials in the underlying $\ll n^2$ bits of randomness and are thus highly correlated. Note also that $M_Z$ is not a tensor product of $M_{line}$ and thus does not admit an immediate %easy 
description of its spectrum in terms of $M_{line}$\footnote{For the dense case, Cai and Potechin \cite{cai2022mixing} showed that surprisingly, the spectrum of the singular values of $M_Z$ has a relatively simple description in terms of the spectrum of the singular values of $M_{line}$}.
% \begin{remark}
%     Reader may notice that the above is not exactly consistent with our definition of graph matrix for shapes that have columns/rows indexed by subsets. In other words, the corresponding matrix for the shape would have the corresponding entry as the following sum \begin{align*}
%     M[\{i,j\}, \{k,\ell\}] &=\chi_G(i,k)\chi_G(j,k)\chi_G(j,\ell) + \chi_G(i,\ell)\chi_G(j,\ell)\chi_G(j,k) \\&+ \chi_G(j,k)\chi_G(i,k)\chi_G(i,\ell) + \chi_G(j,\ell)\chi_G(i,\ell)\chi_G(i,k)     
%     \end{align*}
%    That said, we focus on particular term $\chi_G(i,k)\chi_G(j,k)\chi_G(j,\ell)$ for illustration, and consider the matrix as rather being indexed by order-tuple for simplicity to enable an explicit argument.
% \end{remark}

\begin{figure}[h!]
    \begin{minipage}{0.48\textwidth}
     \centering
    \includegraphics[width=160pt]{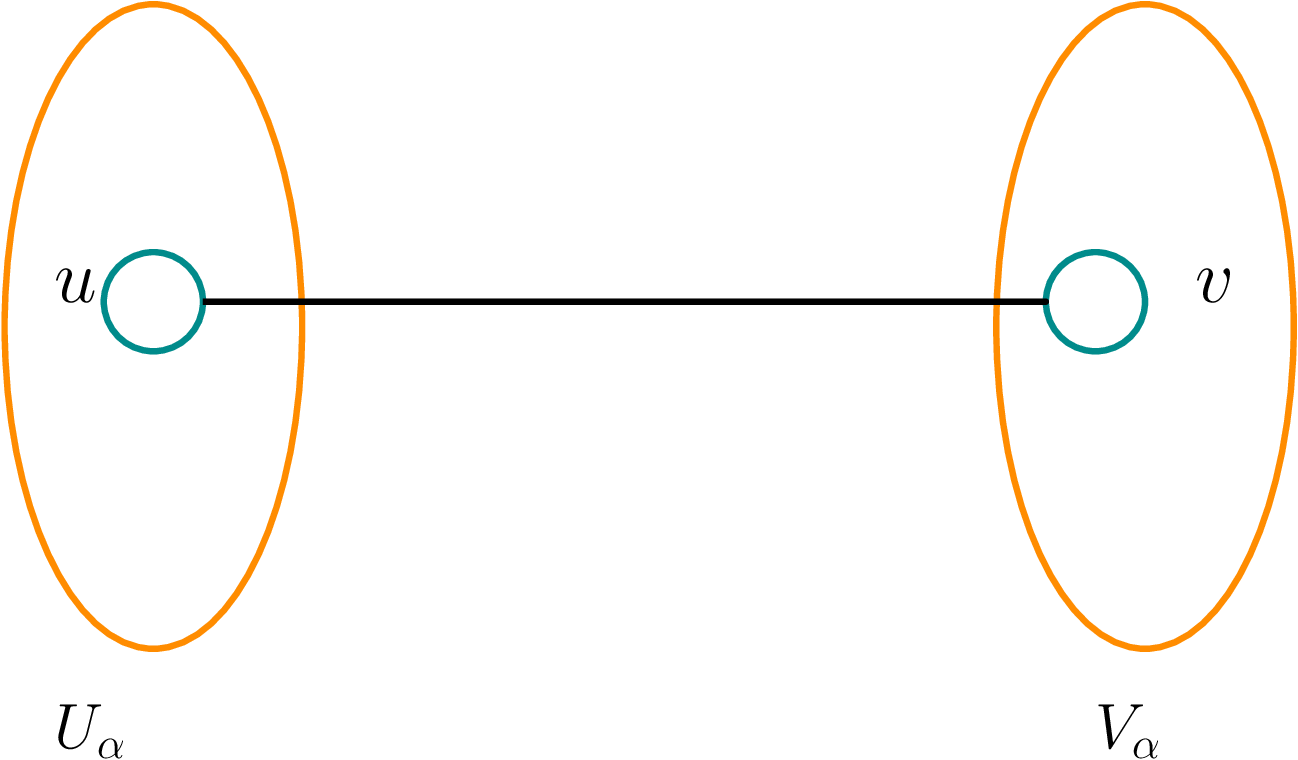}
    \caption{Line-graph \\ $M_{line}[u,v] =\chi_G(u,v) $}
   \end{minipage}\hfill
   \begin{minipage}{0.40\textwidth}
     \centering
       \includegraphics[width=160pt]{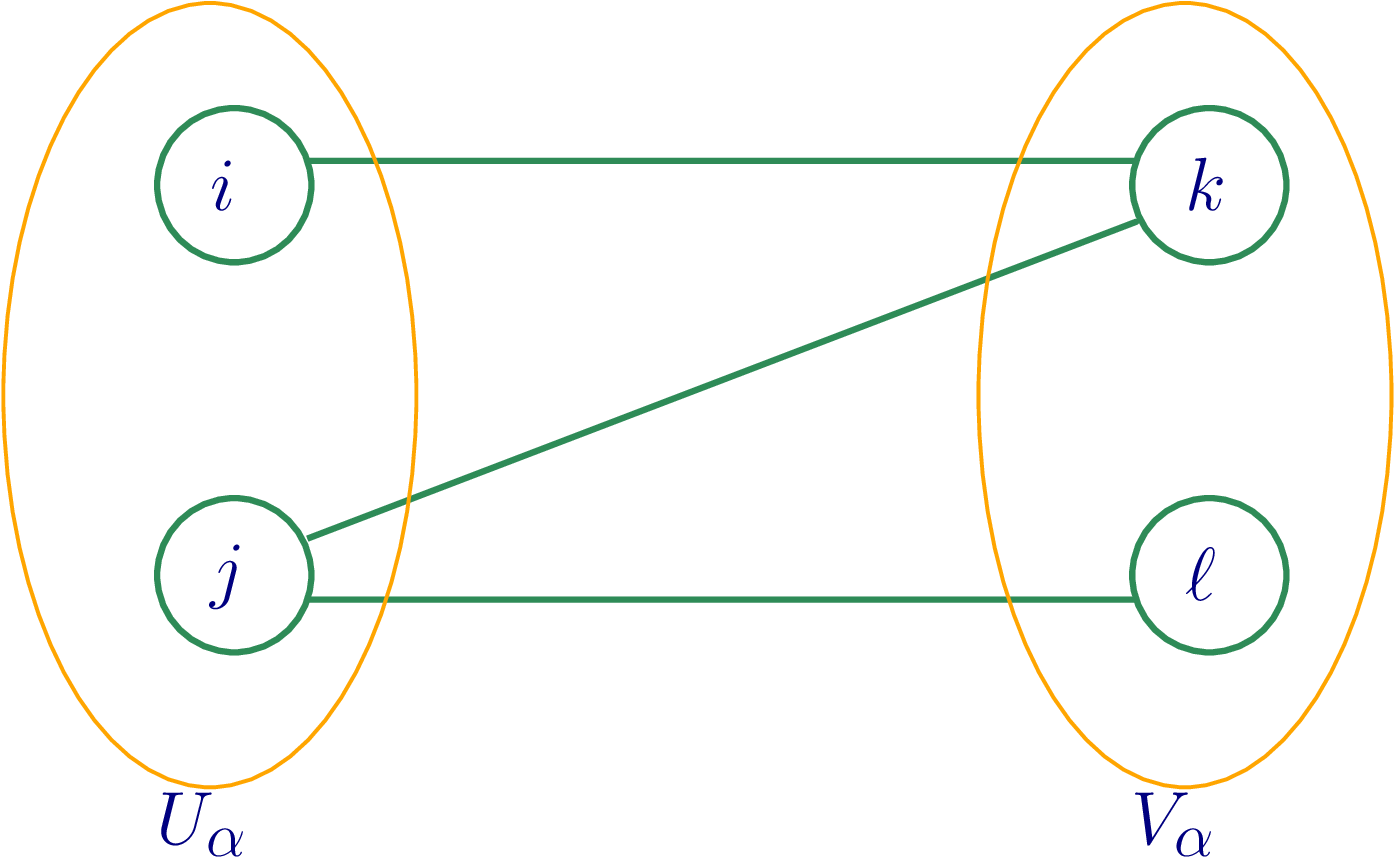}
    \caption{%"Simplified" 
    Z-shape $M_Z[(i,j),(k,\ell)] = $\\ $\chi_G(i,k)\cdot \chi_G(j,k)\cdot \chi_G(j,\ell) $}
   \end{minipage}
\end{figure}

% At this point, we would like to point out that $Z$-shape matrix is a canonical example of graph matrix of our interests as its entries are non-trivial low-degree polynomials of the random graph sample. Notice that correlation across the entries are evident as this is a matrix with $\Theta(n^4)$ non-zero entries while we essentially have $O(n^2)$ bit of randomness from the random graphs. 

% Moreover, despite adjacency matrix (and its variations) have established norm bounds from various other techniques in previous works, the lack of obvious tensor-product structure renders it unclear how one may hope to transfer well-known norm bounds from adjacency-matrix in a black manner for this particular matrix of concern. Nonetheless, we will use two specific graph matrices \emph{line-graph} and \emph{Z-shape} as running examples throughout this section, and prove norm bounds for both of these two matrices that are tight up-to absolute constants to highlight our techniques. 

In this section, we will sketch the following two lemmas that are special cases of our main result on spectral norms of graph matrices and illustrate some of our key new ideas. Our spectral bounds hold when $M_Z$ and $M_{line}$ are defined as function of $G'$ obtained from $G \sim G(n,d/n)$ by removing all edges incident on vertices of degree $\geq O(d)$. 
\begin{lemma}[Informal version of Lemma~\ref{lem:line-graph-bound}, ~\ref{lem: Z-shape-bound}]
	With high probability over the random graph sample $G \sim G_{n,d/n}$, we have \[ 
	\|M_{line}(G')\| \leq O(\sqrt{n})\,,
	\]
	and \[ 
	\|M_{Z}(G')\| \leq O(\sqrt{\frac{n}{d}} \cdot \sqrt{n}^2)\,.
	\]
 where $O(\cdot)$ hides an absolute constant independent of $d$. Here, $G'$ is obtained from $G$ by removing all edges incident to vertices of degree $\geq O(d)$. 
\end{lemma} 
We note that the above bounds are tight up to absolute constant factors (independent of $d$). 

\paragraph{Trace moment method}  Our proof proceeds by analyzing trace moments of matrices using that for any $A$, the spectral norm $\Norm{A}_{spec} \leq \tr((AA^{\top})^q)^{1/2q}$. Taking $q$ to be logarithmic in the dimension of $A$ suffices to get a bound on $\Norm{A}$ sharp up to $1+\delta$ for an arbitrarily small $\delta>0$. The trace power can be expanded as a sum of weighted walks of length $2q$ over the entries of $M_{\alpha}$. Notice that each walk can be viewed as a labeling of a graph composed of $q$ blocks of $\al$ and $\al^T$, in particular, it suffices for us to consider the following labelings of walks.

Each term in the expansion of $\tr((M_{\al}M_{\al}^{\top})^q)$ corresponds to a labeling of the vertices in $q$ copies of $\alpha$ and $\alpha^{\top}$ with labels from $[n]$ (i.e, vertices of the graph $G$) satisfying some additional constraints that correspond to a valid walk. The following definition captures these constraints. 

\begin{definition}[Shape walk and its valid labelings]
	Let $\al$ be a shape and $q\in \N$. For each walk $P$, let $G_P(\al,2q)$ be the \emph{shape-walk} graph of the shape-walk $P$ vertices on vertices  $V_P(\al,2q)$ formed by the following process, \begin{enumerate}
		\item Take $q$ copies of $\al_1,\dots,\al_q$ of $\al$, and $q$ copies of $\al_{1}^T,\dots, \al_q^T$ of $\al^T$;
		\item For each copy of $\al_i$ (and $\al_i^T$), we associate it with a labeling, i.e., an injective map $\psi_i : V(\al) \rightarrow [n]$ (and respectively $ \psi_i' $ for $\al^T$); 
%		\item Each boundary $U_{\al_i}, V_\al = U_{\al^T_i}, V_{\al^T_i}$ is specified by an order-tuple in $[n]$;
		\item For each $i\in [q]$, we require the boundary labels to be consistent as ordered tuples, i.e. $\psi_i(U_{\al}) =\psi'_{i-1}(V_\al)$ and $\psi'_{i}(U_{\al^T}) = \psi_i(V_{\al_i})$;% as tuple-equalities;
		\item Additionally, each such walk must be closed, i.e., $\psi'_q(V_{\al^T}) = \psi_1(U_\al)$ as a tuple-equality;
		\item We call each $\al_i$ and $\al^T_i$ block a block-step in the walk.
\end{enumerate}
For each walk-graph , we associate it with a natural decomposition $P = \psi_1\circ \psi_1'\circ \psi_2\circ\dots\circ\psi_{q}\circ \psi_q'$.
% such that for each $i\in [q]$, $\psi_i$ and $\psi_i^T$ are injective mappings of $\al$ and $\al^T$ with the additional boundary consistency that $\psi_i(V_\al) = \psi_i^T(U_\al)$, and  $\psi_i(U_\al) = \psi_{i-1}^T(V_\al)$.
\end{definition}

\label{sec:norm-bound-overview}

With the walks identified, we may now unpack the trace inequality we are in pursuit of, and offer some motivation for the counting scheme to be described. 
 For any $q\in \N$, 
\begin{align}
	\E_{G\sim G_{n,d/n}} \left[\Tr\left((M_\al M_\al^T)^{2q}\right)\right] &=\sum_{\substack{\text{shape walk} P:V(\al,2q)\rightarrow [n]\\\text{locally injective at each piece}}} \E_G[\prod_{t\in [q]} \prod_{(i,j)\in E(\al)} \chi_G(\psi_t(i),\psi_t(j))  \cdot \chi_G(\psi'_t(i),\psi'_t(j))    ] \nonumber \\
	&=\sum_{\substack{\text{shape walk} P:V(\al,2q)\rightarrow [n]\\\text{locally injective at each piece}}} \E_G[\val_{G}(P)  ] \nonumber
%	\\
%	&=\sum_{P:\text{shape-walk} }\E_{G\sim G_{n,d/n}} \left[\val_G(P)\cdot \cycle \cdot \bdd  \right] \label{eqn:trace-eq}
	\end{align}
%	where we switch from the pruned-random graph sample $\calG$ to the original random graph sample with an extra indicator function that each vertex visited in the walk is bounded degree in the random graph sample $G$ and no $2$-cycle in any small neighborhood.
%
% It is standard by now that once the walk is long enough, i.e., the matrix is raised to a large enough power that is usually $\Theta(\poly\log \text{matrix-dimension} )$ \footnote{where we emphasize the extra dependence on matrix-dimension as our techniques apply to not necessarily matrices of size $n^{O(1)}\times n^{O(1)}$ as typically found.}, the bound on expected trace immediately translates to a high probability bound that holds with probability at least $1-\frac{1}{\poly(n)}$ via Markov's inequality. For our exposition, it suffices for us to restrict our attention to bounding the expected trace quantity from above.

where we use the short form
\[ 
\val_G{(P)} \coloneqq \prod_{e=(i,j)\in E(P)} \chi_{G}(e) = \prod_{e\in E(P)}\chi_G^{\mul_P{(e)}}(e)
\]
where we denote $E(P)$ be the set of edges traversed in the walk $P$, and $\mul_P(e)$ is the multiplicity of $e$ appearing in the walk, where the dependence on $P$ is usually dropped when it is clear from the context. Our goal now is to bound the sum of $\val_G(P)$ over all possible labelings $P$ that are locally injective, and that could possibly occur in the trace power expansion above. 

%We characterize such ``valid" labelings $P$ that can occur in trace expansion by a collection of natural conditions described below.  
%	
%% 	\paragraph{Block walk for graph matrix}
%% To bound the trace, it is useful to expand out the sum and notice each term can be described by a shape-walk as the following.
%
%
%	Each term in the summation of $\Tr\left[(M_\al M_\al^T)^{q}\right]$ is a shape-walk consists of block-steps, and bounding the value and the count of such walks shall be the main subject of our work.
%	

\paragraph{Pruning and Conditioning }

 A complication that arises from our setting is that the high moments of the trace are in fact too large.  
 As mentioned in the introduction, there are two fundamental sources responsible for this issue: 1) fluctuations in the vertex degrees which have a large effect on the norm (so the improved norm bound is false without pruning high degree vertices); %the fluctuation of vertex degree, and the norm is dominated by the high-degree vertices; 
 and 2) rare events happening with inverse-polynomial probability that dominate the contribution to the expected trace power, eg., the event of having a small dense subgraph. To address these two challenges, we work with $\calG$ which is a pruned subgraph of a graph sample from $G\sim G_{n,d/n}$, which is further conditioned on the high probability event that the graph has $2$-cycle free radius $\Theta(\log_d n)$.% at $\kappa = \Theta(\log_d n)$.

\begin{definition}[$2$-cycle freeness]
	Given a graph $G=(V,E)$, we call it a $2$-cycle free graph if for any pair of vertices $u,v \in V(G)$, there are at most $2$ simple paths between $u$ and $v$. Equivalently, each connected component of $G$ has at most one cycle.
 \end{definition}

\begin{definition}[$2$-cycle free radius]
The $2$-cycle free radius of a graph is the largest $r$ such that 
%We call	$\kappa$ the radius of $2$-cycle free neighborhood if 
for any vertex $v$, the induced %vertex-induced 
subgraph of vertices within distance $r$ of $v$ %in its neighborhood within distance $\kappa$ 
is a $2$-cycle free graph.
\end{definition}
 \begin{definition}[Indicator functions]
    Given a set of vertices $V \subseteq V(G)$, let $\bdd[V]$ be the indicator function for each vertex in $V$ having bounded degree, i.e., $\bdd[V]=1$ if every vertex in $V$ has degree at most $c_{degree}\cdot d$ and $0$ otherwise. For a shape walk $P$, we let $\bdd[P] = \bdd[V(P)]$ where $V(P)$ is the set of vertices visited by the shape walk.

    We set $\kappa = 0.3\log_d n$. Let $\cycle[G] = 1$ if $G$ has 2-cycle free radius at least $\kappa$ and $\cycle[G] = 0$ otherwise.
    %Given a walk $P$, let $\bdd[G]$ be the indicator function for bounded degree, i.e. $\bdd[G]=1$ if every vertex in $P$ has degree bounded by $c_{degree}\cdot d$ in $G$, and $0$ otherwise; let $\cycle[G]$ be the indicator function for the absence of nearby $2$-cycle, i.e., $\cycle[G]=1$ if there is no nearby $2$-cycle in $G$, and $0$ otherwise.
\end{definition}
 %In our work, we will set $\kappa \coloneqq 0.3\log_d n$. We also call the event on graph sample $G$ that it has $2$-cycle free radius at least $\kappa$ "no nearby $2$-cycle".
\begin{fact}[\cite{Fri08, BLM15, FM17, Bor19}]\label{fact:2-ccl-bound}
A random graph from $G_{n,d/n}$ has $2$-cycle free radius at least $\kappa = 0.3\log_d n$ with high probability $1-o_n(1)$.
\end{fact}

  Therefore, the distribution we work with needs to to be considerably modified, and we are ultimately interested in the following value,
 \begin{align}
	\E_{\calG}\left[\Tr\left((M_\al M_\al^T)^{2q}\right) \cdot \cycle  \right] &=
	\sum_{\substack{\text{shape walk} P:V(\al,2q)\rightarrow [n]\\\text{locally injective at each piece}}} \E_\calG[\val_{\calG}(P)\cdot \cycle \ ] \nonumber
	\\
	&=\sum_{P:\text{shape-walk} }\E_{G\sim G_{n,d/n}} \left[\val_G(P)\cdot \cycle \cdot \bdd[P] \right] \label{eqn:trace-eq}
	\end{align}

	\subsection{Motivating our machinery: walk global, think local}
	\paragraph{Global vs local counts} The trace moments are typically analyzed by inferring global constraints on contributing walks. For example, in the analysis for tight norm bound for Gaussian random matrices, the dominant contributing walks have a one-to-one correspondence with the Dyck walk with exactly half of the steps going to ``new" vertices, and the other half going to "old" vertices. However, such global analyses become unwieldy once applied to slightly more complicated graph matrices e.g., the $Z$-shape matrix we defined earlier. Our main idea is a new \emph{local} bound that holds on each step of the walk. This local bound generalizes naturally to graph matrices with more complicated shapes while still giving bounds sharp up to an absolute constant.
%
%\paragraph{Also a fear of going too local}  Previous techniques \cite{AMP20, JPRTX} implicitly proceed in a similar manner while they obtain a separate bound for $\vtxcost$ and $\edgeval$, which is bound to result in typically somewhat loose bounds lossy in the log factors as it is not immediately clear whether all walks admit a succinct encoding.
%	
	
	For each block-step (a step corresponds to a labeling of a shape $\alpha$ when bounding weighted walks on a graph matrix $M_{\al}$) in a walk that contributes to \cref{eqn:trace-eq}, we will associate two types of charges:\begin{enumerate}
		\item $\vtxcost$ that is used to identify the labels of the vertices appearing (counting);
		\item $\edgeval$ that captures the expectation of random variables of edges traversed in the walk;	\end{enumerate}	
	 In total, we will bound the contribution due to block step in the walks contributing to  \cref{eqn:trace-eq} as \[ 
	\vtxcost \cdot \edgeval \leq B_q(\al)
	\]
	where $B_q(\al)$ is some desired upper bound for shape $\al$, which we shall refer to as block-value bound.
	An assignment of valid block cost $B_q(\al)$ for each block translates immediately into a bound for the final trace moment. Formally, we define the block-value function as the following,
 \begin{definition}
		For any shape $\al$, and $q\in \N$, for any vertex/edge-factor assignment scheme, we call $B_q(\al)$ a valid block-value function of the given scheme if \[ 
		\E[\tr(M_\al M_\al^T)^q] \leq \text{matrix-dimension} \cdot \auxcost \cdot B_q(\al)^{2q}
		\] 
		for some auxiliary function $\auxcost$ such that  \[ \left(\text{matrix-dimension} \cdot   \auxcost\right)^{1/2q} \leq 1+o_n(1)\,,\] and for each block-step $\mathsf{BlockStep}_i$ throughout the walk,
		\[ 
		\vtxcost(\mathsf{BlockStep}_i) \cdot \edgeval(\mathsf{BlockStep}_i) \leq B_q(\al)
		\]
	\end{definition}
We stress that the block-value function $B_q(\al)$ depends on the shape $\al$, the length-$q$, and the vertex/edge-factor assignment scheme. The bulk of our work is in finding a vertex/edge-factor assignment scheme that produces a minimal valid $B_q(\al)$ --- the \emph{local} component of our argument. %In particular, it boils down into bounding the vertex-factor of the block, coinciding with the common wisdom that combinatorial counting is the crux of obtaining tight bounds in trace-method calculation.

\paragraph{Labeling the steps}
 Our strategy starts by identifying the ''status'' of the random variable used in a given block, which we coin "step-label". 

 \begin{definition}[Edge and step]
	We use ''edge'' to refer to the undirected edge in the underlying random graph sample, and   ''step''  to refer to a directed edge when mentioned in the context of a walk. \end{definition}

The following vertex labeling scheme is a key component of our new local argument. Each step in the walk corresponds to a vertex labeling of the shape. Each such labeling yields a block contributing equal to the product of the characters on edges appearing in the labeled shape. For each such edge, we will use four types of labels in order to help us construct $\vtxcost$ and $\edgeval$.  
\begin{definition}[Step-label]
We categorize the status of each step as the following. For a step whose underlying edge/random variable appears at least twice throughout the walk,
\begin{enumerate}
	\item $F$ (a fresh step): an edge (or random variable) appearing for the first time, and the destination vertex is appearing for the \emph{first} time in the walk;
	\item $S$ (a surprise step/visit): an edge (or random variable) appearing for the first time, and the destination vertex is \emph{not} appearing for the first time in the walk;
	\item $R$ (a return step): an edge (or random variable) appearing for the last time;
	\item $H$ (a high-mul step) : an edge (or random variable) appearing for neither the first nor last time.
\end{enumerate}
For a step whose underlying edge/random variable appears only once throughout the walk, we call the step a \emph{singleton} step, and additionally call its underlying edge a \emph{singleton edge}. We will be able to consider a singleton step as a subclass of $F$ steps in our accounting.
\end{definition} 

\begin{observation}
	Each step receives a single step-label among $\{F,R,S,H, \text{Singleton}\}$. 
\end{observation}

When working with a graph matrix, a single block-step corresponds to several edge steps (random variables). Hence, we extend the step label to a labeling for an entire block-step, 
 \begin{definition}[Block-step labeling]
 	Given a shape $\al$, a labeling $\calL: E(\al)\rightarrow \{F,R,S,H, \text{Singleton} \} $ for a block-step is a collection of step labels for for each edge random variable in $\al$.
 \end{definition}

\paragraph{Vertex appearance and redistribution} We will use a different cost for a vertex label depending on how it appears in a walk. A careful choice of such costs is important.  

For example, in the standard argument for bounding the spectral norm of $n \times n$ Gaussian random matrices, each step, if leading to a new vertex, contributes a value of $n$ as the destination vertex takes a label in $[n]$, while on the other hand, if going to a ``seen" vertex, contributes a value of $1$. In this case, if one applies the block-value bound naively, observe that it would yield a bound of $O(n)$ as opposed to a bound of $O(\sqrt{n})$ with the contribution being dominated by the steps leading to new vertices. For this particular example, this bound can be improved by observing at most half of the blocks attaining a value of $n$ while the other half has a value of $1$ thus geometrically averaging out to $\sim \sqrt{n}$. 

However, such reasoning does not readily fit into our local reasoning framework. We will instead introduce a vertex redistribution scheme. In the walk, we first formally place vertex's "appearance" into three categories.

\begin{definition}[Vertex appearance in a block-step] 
Let $v$ be a vertex in the current block-step. We say that $v$ is making its first appearance if $v$ is not contained in any previous block-steps. We say that $v$ is making its last appearance if $v$ is not contained in any later block-steps. We say that $v$ is making a middle appearance if $v$ appears in both an earlier block-step and a later block-step. %We will say that each labeled vertex appears throughout the walk as "first", "middle", and "last". 
Note that we consider a vertex appearing on the block-step boundary ($U_\al$ or $V_\al$) as appearing in both adjacent blocks.
\end{definition} 
% \pravesh{What's the term being defined in this definition? }
% \jnote{"vertex appearance"...i don't think this is defining anything new, but more like clarifying}
\begin{examples}
 If $v \in V_{i} = U_{i+1}$, $v \in V_{j-1} = U_j$ for some $j > i+2$, and $v$ does not appear anywhere else in the walk then $v$ makes its first appearance in block $i$, makes middle appearances in blocks $i+1$ and $j-1$, and makes its last appearance in block $j$.
 %For a vertex first appears in the right-boundary of block $i$, $V_i$, and last appears in the left-boundary of block $j$, $U_j$, then it will make middle appearances in the left-boundary of block $i+1$, $U_{j+1}$, and right-boundary of block $j-1$, $V_{j-1}$, as well.
\end{examples}

To handle the disparity in the magnitude due to the step-choices, we adopt the following vertex-factor redistribution scheme.

\begin{mdframed}[frametitle = {Vertex-factor redistribution scheme} ]
	Each vertex, when it first appears, requires a label in $[n]$. We redistribute this factor as follows.%, while we redistribute factors as the following, 
 \begin{enumerate}
 	\item Assign a $\sqrt{n}$ factor to the block-step in which the vertex first appears;
 	\item Assign a $\sqrt{n}$ factor to the block-step in which the vertex last appears.
 \end{enumerate}
\end{mdframed} 
 Observe that each vertex picks up in the end a factor of $n$ throughout the walk as it gets $\sqrt{n}$ factor each from its first and last appearance, and thus the vertex factor (for identifying new vertices) is preserved. To illustrate the vertex redistribution scheme, we apply the above machinery to get a loose bound for random $\pm 1$ matrix that is tight up to polylogarithmic factors.
 
\paragraph{Warm up: a loose bound for random $\pm 1$ matrix}	
For the following discussion, let $G$ be an $n\times n$ symmetric matrix with each (off-diagonal) entry sampled i.i.d. from $\{\pm 1 \}$. 
\begin{lemma}[Naive bound on symmetric random matrices via local argument] We can bound its block-value function by
	\[ B_q(G) \leq \sqrt{n \polylog n} \,,\]
	As a corollary, with probability at least $1-o_n(1)$, we have $\|G\| \leq O(\sqrt{n}\polylog n)$.
\end{lemma}
\begin{proof}
	Since each edge contributes a value $1$ in expectation whenever the walk is non-zero, bounding the trace value reduces from bounding weighted closed walks to simply bounding walk-counts. Hence, it suffices for us to focus on vertex-factors for this example. We start by casing the step-label, and suppose we are traversing from (left)  $U$ to (right) $V$, 
	\begin{enumerate}
		\item $F$-step: the left vertex cannot be making its first appearance by definition, and it cannot be last appearance as otherwise the edge appears only once throughout the walk. The destination vertex is making a first appearance, hence contributing a vertex-cost of $\sqrt{n}$;
		\item  $S/H$-step: the left vertex cannot be making its first appearance by definition, and it cannot be last appearance as otherwise the edge appears only once throughout the walk. The destination vertex is making a middle appearance, hence contributing a vertex-cost of $2q = \poly\log n$;
		\item $R$-step: The vertex on the left boundary is potentially making a last appearance, hence contributing a vertex-cost of at most  $\sqrt{n}$; the right vertex cannot be making its first appearance as the edge appears in prior steps, and the vertex is not making its last appearance by definition. It can be further specified at a cost of  $2q = \poly\log n$;
		\item Summing the above, we can bound $B_q(G) \leq 2\sqrt{n}\cdot q = \sqrt{n}\polylog (n) $ by setting $q = \polylog n$ where we remind the reader that we need to set $q=\Omega(\log n)$ to obtain a bound that holds with probability $1-o_n(1)$.
	\end{enumerate}

	\end{proof}

%On a high-level, we exploit the observation that we have two directions for traversing a  particular closed walk, by either traversing clockwise or counterclockwise, it suffices for us to bound the trace by the (geometric) average of two traversals. Concretely, for each block $B_i$ with labeling $\calL$, we will bound its corresponding block value by \[ 
%\vtxcost(B_i)\cdot \edgeval(B_i) \leq \sqrt{\vtxcost_{U\rightarrow V}(\calL)\cdot  \vtxcost_{V\rightarrow U}(\calL)} \cdot  \edgeval(\calL)
%\]
%where we use $\vtxcost_{U\rightarrow V}$ to denote the vertex cost when traversing from $U$ to $V$. 
%For some technical reason in our charging scheme, we reassign the vertex factor for the second visit and last visit: that the second time a leg being used is ''fixed'', while we swap its vertex factor with that of its last appearance.  

%\subsection{R-destination is \emph{free}: shaving the first log}

\paragraph{The challenge of a local argument} We note that previous techniques \cite{AMP20, JPRTX, GT23} for analyzing graph matrices give norm bounds that are tight up to  $O(\log^{O(|V(\alpha)|)} n)$ loss for any $M_{\alpha}$. The main challenge in our new local argument is to obtain bounds tight up to an absolute constant factor.  

For readers familiar with trace-method calculation for the example of G.O.E. matrices, or non-backtracking walk matrices for sparse random graph, observe that the bound for the $R$-step in the above example is lossy when we use a label in $2q$ to specify the destination of an $R$-step. This is indeed the source of the $\polylog$ gap in the above argument. In contrast, should one imagine the walk proceeds in a tree-like manner, we may consider the $F/R$ steps as either walking towards or away from the root and for the case of an $R$ step walking towards the root, the destination of an $R$-step is \emph{fixed} as each vertex has only one edge that is moving closer to the root, avoiding a naive cost of $2q$ to specify the destination. 
	
With the above intuition, one still needs to be careful as there is no reason a priori to focus exclusively on walks that are strictly tree-like, and the $R$-destination is no longer fixed once we start seeing cycles, or more formally, surprise/high-mul steps in the walk. However, as the tight argument for G.O.E. matrices shows, the effect of surprise steps is in some sense \emph{local} as it can be shown that each surprise step can only ``confuse" at most $2$ different $R$-steps, i.e., $R$-steps whose destinations are not unique. To exploit this observation, one may observe that there is a tremendous \emph{gap} opened up by an $S$-step in the warm-up example, in particular, we have a budget of $\sqrt{n}$ per step, while an $S$-step only requires a cost of $2q$ to specify the destination. That said, as opposed to paying just a cost of $2q$ to specify the destination, one may additionally pay a cost of $(2q)^2$ to specify the destination of the $2$ $R$-confusions that the $S/H$-step may incur. By paying this extra cost for $S$-step (as well as for $H$-step in a similar fashion), one can assume that \emph{$R$-destination is fixed} throughout the walk, shaving the polylog overhead. Broadly speaking, we will call the cost used to settle $R$-confusion $\pur$-factors/costs (\emph{potentially unforced-return}) and we defer this to ~\pref{sec:return-cost}.  The above argument is formalized by the following,
	
	\begin{lemma}[Tight bound for random $\pm 1$ matrix] Consider $G$ an $n\times n$ symmetric matrix with each off-diagonal entry sampled i.i.d. from random $\pm 1$ entry. 
	Assuming each $S/H$-step incurs at most $2$ $R$-confusions, 	we can bound its block-value function by
	\[ B_q(G) \leq (1+o_n(1)) \cdot  2\sqrt{n} \,,\]
	for $q \ll n^{O(1)}$. 
	As a corollary, with probability at least $1-o_n(1)$, we have $\|G\| \leq (1+o_n(1)) \cdot  2\sqrt{n} $. This bound is tight even in the leading constant.
	\end{lemma}
	\begin{proof}
		Identical to above, we focus on the vertex factors, and consider the step-label when we traverse the walk from left to right (i.e. from $U_\al$ to $V_\al$).
	\begin{enumerate}
		\item $F$-step: the left vertex cannot be making its first appearance by definition, and it cannot be last appearance as otherwise the edge appears only once throughout the walk. The destination vertex is making a first appearance, hence contributing a vertex-cost of $\sqrt{n}$;
		\item  $S/H$-step: the left vertex cannot be making its first appearance by definition, and it cannot be last appearance as otherwise the edge appears only once throughout the walk. The destination vertex is making a middle appearance, hence can be specified by a cost of $2q$. Moreover, under our assumption that each $S/H$ step incurs at most $2$ $R$-confusions, each of which can be specified by label in $2q$, we have a total cost of \[ \underbrace{2q}_{\text{destination of $S$ step}} \cdot \underbrace{(2q)^2}_{\text{destination of potential R-confusions} } \ll \sqrt{n}\]
		\item $R$-step: the vertex on the left boundary is potentially making a last appearance, hence contributing a vertex-cost of at most  $\sqrt{n}$; the right vertex cannot be making its first appearance as the edge appears in prior steps, and the vertex is not making its last appearance by definition. 
		\item Summing the above, we can bound $B_q(G) \leq \sqrt{n} + (2q)^3 + \sqrt{ n } \leq (1+o_n(1)) \cdot 2\sqrt{n} $ by setting $q = \polylog n$.
	\end{enumerate}

	\end{proof}

    \begin{remark}
        In our general analysis,  we will utilize that a bound $O(1)$ in the exponent of $q$ would have sufficed for us in the above argument. % also in establishing a final for $B_q$ bound that does not have dependence in $q$ provided $q \ll n^{o(1)}$. \pravesh{where are these dangling/floating components defined?} \jnote{rephrased now? does this mean $q\ll n^c$ for any constant $c>0$?}
    \end{remark}
 
The key to the argument above is that each $S/H$ step incurs at most $2$ different $R$-confusions. We note that this is where our machinery departs from a completely "local" argument, which allows us to improve from prior bounds for graph matrices. On a high level, a key component of our analysis relies upon showing \emph{return-is-fixed}, i.e., there are $O(1)$ choices for the destination of an $R$-step, and our machinery follows a rule of thumb inspired by the above example, 
%	Formally, we employ a helper function $\pur$ to capture the effects of $S/H$-steps on rendering confusions for $R$-destinations 
%	assuming that each $S/H$-step can cause at most $2$ $R$-confusions, 
	
	 \begin{center}
		"If the walk gets too messy to encode, there must be ``savings" in the combinatorial factor from previous blocks, that allows us to encode auxiliary information for decoding."
	\end{center}
	To exploit the loss in previous steps, it is no longer sufficient for us to \emph{narrowly} consider a single local step. Towards this end, we employ a global potential function argument to identify the source of such loss throughout the walk, and more importantly, we incorporate the global potential function into our local block-value bound such that our machinery is ultimately a "local" argument. In particular, the only "global" component of our argument is establishing the assumption that \emph{each $S/H$ step incurs at most $2$ $R$-confusions}, and its further generalizations to other steps that come with \emph{gaps} when compared with the dominant block-step (a.k.a. \emph{maximal-value labeling}). 

\paragraph{Further challenges from ultra-sparsity}
The above showcases the bulk of our ideas while extra care is required for our applications for \emph{ultra}-sparse random graphs. To illustrate these challenges, we will use line-graph as an running example and showcase the new ideas needed %should one attempts 
to migrate the local analysis to sparse random graphs.

 For the following discussions, let $M_G$ be an $n\times n$ symmetric matrix with each entry an independent sample from the $p$-biased distribution. For readers familiar with the $O(\sqrt{d})$ spectral radius bound for sparse random graph with average degree-$d$, the desired bound translates into a bound of $O(\sqrt{n})$ under the $p$-biased basis. That said, our goal is to recover an $O(\sqrt{n})$ bound for the following "shape" within our machinery.
\begin{figure}[h!]
     \centering    
    \begin{minipage}{0.48\textwidth}
    \includegraphics[width=160pt]{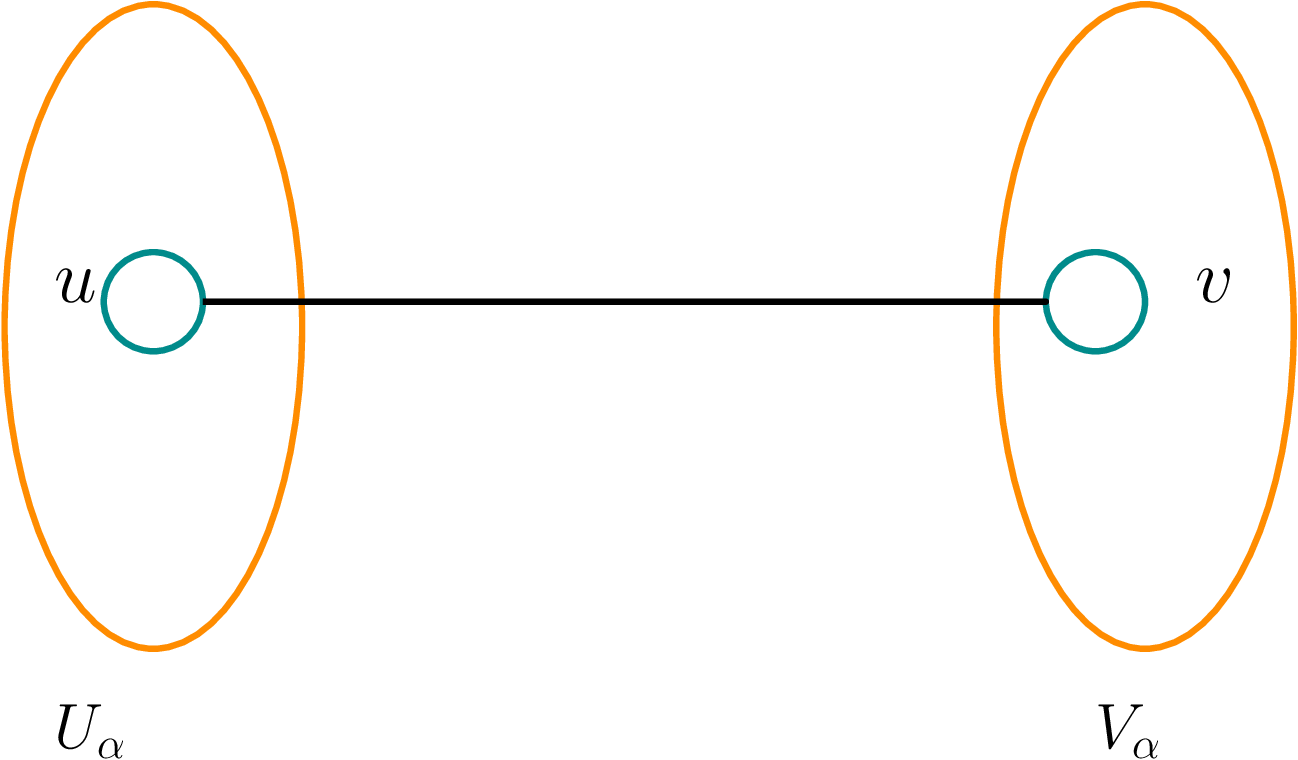}
   \end{minipage}\hfill
    % \caption{Line-graph \\ $M_{line}[u,v] =\chi_G(u,v) $}
\end{figure}

\paragraph{Handling singleton edges due to pruning and conditioning} For staters, one immediate challenge from the ultra-sparse regime is that we are not working directly with a genuine random graph sample in which each edge is an independent sample: instead, we work with random graphs with pruned vertices, as well as conditioning on the event of $2$-cycle-free subgraphs. With this alternative distribution, we need to take into account of \emph{singleton} edges, since it is no longer true the $p$-biased character has mean-$0$,  as opposed to typical trace-method calculation for mean-$0$ random variable. This effect is manifest within our framework when we consider a singleton $F$-step. To illustrate the challenge posed by singleton edges, we first formally describe our edge-factor assignment scheme, inspired by the previous work of \cite{JPRTX}.

\begin{mdframed}[frametitle = {Edge-value assignment scheme }]
For each random variable $x$,
	\begin{enumerate}
		\item $F$-step: for the first time the random variable appears,
		\begin{itemize}
			\item $\mathsf{Singleton}$: for random variable $x$ that appears only once throughput the walk, we assign a factor $\sqrt{\frac{n}{d}}\cdot \frac{d}{n}\cdot \singdecay $ for $\singdecay \leq \exp(-d)$;
			\item $F$: for random variable $x$ that appears at least twice throughout the walk, this gets assigned a factor of $1$;
		\end{itemize} 
		\item $H$-step: we assign a factor of $\sqrt{\frac{1-p}{p}} \approx \sqrt{\frac{n}{d}}$.
		\item $R$-step: we assign a factor of  $1$.
	\end{enumerate}
	\end{mdframed}

At this point, we would like to emphasize that the key to handling singleton edges is that we are able to identify $\singdecay$, an extra decay term coming from edge-value for each singleton edge, which we call \emph{singleton-decay}. Furthermore, we show  $\singdecay \ll \exp(-d)$, providing us with enough slack to offset potential combinatorial blow-up from singleton edges. To highlight the importance of this extra decay term, we consider the local block-value for a singleton $F$ step within our framework.

Hypothetically, suppose there is no singleton decay, we consider the block-step value for the singleton step: the source vertex may be now be making a last appearance, while the destination vertex may be making as well its first appearance. This shall be contrasted with the non-singleton $F$ case in which we pick up only one first appearance factor. Combining the above, we pick a factor of \[
\underbrace{\sqrt{n}}_{\substack{\text{redistributed vertex cost} \\ \text{ for the destination's first appearance}} } \cdot\underbrace{\sqrt{n}}_{\substack{\text{redistributed vertex cost} \\ \text{ for the source's last appearance}} } \cdot \underbrace{\sqrt{\frac{n}{d}}\cdot \frac{d}{n}}_{\text{singleton-edge value without decay} }   = \sqrt{nd}\,.
 \] 
Notice that this yields a block value of $\sqrt{nd}$, which gives us a bound equivalent to the simple row-sum bound. That said, it is crucial to improve the above analysis for singleton step, and in this case, it is clear that the extra decay of $\singdecay \leq \exp(-d)$ comes in handy, and allows us to bound the above by $\sqrt{n}$ as desired. 

The key technical lemma concerning this component is to show the above is a valid assignment scheme, i.e., we indeed pick up extra decay by unpacking the trace-bound for singleton edges. The proof is deferred to the appendix in ~\pref{sec:singleton-decay}.

\begin{proposition}[~\pref{prop: walk-value-factor}]
	The above is a valid edge-value assignment scheme.
\end{proposition}
\paragraph{Challenges from unbounded higher-order moments} Another challenge due to ultra-sparsity is that $p$-biased character has unbounded higher moment.
In our edge-value assignment scheme described above, each $H$-step gets assigned a value of $\sqrt{\frac{1-p}{p}}\approx  \sqrt{\frac{n}{d}}$ and this poses extra challenges for our previous argument of each $S/H$ step incurring at most $2$ $R$-confusions. In particular, we can no longer afford each $H$ edge incurring unforced return since it no longer comes with any extra gap. We address this issue by improving upon the unforced-return bound, and show that each $H$-step does not incur any $R$-confusion in ~\pref{sec:return-cost}. 

Besides the effect of $H$-steps on $R$-confusion, we also need to be careful in our encoding for destination of an $H$-step. In particular, even though each vertex has degree at most $c_{trunc}\cdot d$ under the pruned subgraph, a cost of $O(d)$ is still too much for specifying the destination of an $H$-step. Again, we illustrate this issue by zooming into the block-step analysis for an $H$-step. With a naive bound of $O(d)$, observe that we would pick up a factor of \[
O(d) \cdot \sqrt{\frac{n}{d}} = O(\sqrt{nd})
\]
once combined with the edge-value. This is again a trivial bound comparable to a row-sum bound, however, it also prompts us to improve the encoding cost from $O(d)$ to $O(\sqrt{d})$ for the desired block-value of $O(\sqrt{n})$. In particular, we show that \begin{proposition} \label{prop:root-d-suffices}
A cost of $O(\sqrt{d})$ is sufficient for identifying the destination of an $H$-step. 	
\end{proposition}
 To show that a cost of $\sqrt{d}$ is sufficient, we crucially rely upon the assumption that the random graph sample is $2$-cycle free in any $\Omega(\log_d n)$-radius neighborhood. This is further discussed in ~\pref{sec:high-mul-overview} and ~\pref{sec:2-cycle-use}. The underlying idea also hinges upon another crucial observation for vertex factors, that there are occasions the cost of $O(\sqrt{d})$ may even be spared. To identify such vertices, we make the following formal definition,
 \begin{definition}[Doubly constrained]
 	We call a vertex $v$ doubly constrained by a set of vertices $S$ if there is a path between $i,j\in S$ that passes through $v$.
 \end{definition}
 \begin{proposition}[Informal: doubly-constrained vertices are fixed]
 	Any vertex that is doubly-constrained by a set of known vertices via $H$ steps can be identified at a cost of $O(1)$.
 \end{proposition}

Combining the above components together, we are ready to introduce our desired vertex-factor assignment scheme,

\begin{mdframed}[frametitle = {Vertex-factor assignment scheme: simplified version} ]
	Each vertex requires a label in $[n]$ when it first appears. We redistribute factors as follows: \begin{enumerate}
 	\item Assign a $\sqrt{n}$ factor to the block-step in which the vertex first appears;
 	\item For any middle appearance,\begin{itemize}
 		\item if it is arrived via an $S$ step: assign the corresponding $\pur$ cost and the cost of specifying the vertex (which is again at most $2|V(\al)|q$);
 		\item if it is arrived via an $H$ step and it is not doubly constrained: assign a cost of $O(\sqrt{{d}})$;
 		\item  if it is reached via an $R$ step or it is doubly constrained by an $H$-path: assign a cost of $O(1)$; 
 	\end{itemize}
 	\item Assign a $\sqrt{n}$ factor to the block-step in which the vertex last appears.
 \end{enumerate}
\end{mdframed} 

There are several other technical challenges that arise once we start working with more complicated graph matrices and with matrices that are a sum of a collection of shapes as opposed to a single shape. We defer these discussions to the latter part after we introduce the connection between vertex separators and matrix norm bounds. However, with the ideas already highlighted, we are already able to recover the desired $O(\sqrt{n})$ bound for the line-graph and even get tight norm bounds for some less trivial graph matrices.

\subsection{Toy examples for sparse matrix norm bounds}
\begin{lemma}[Tight bound for scaled adjacency matrix/line graph for  $G_{n,d/n}$] \label{lem:line-graph-bound}
	Let $G$ be a random graph sample from $G_{n,d/n}$, and let $M_G$ be its corresponding matrix in the $p$-biased basis once vertices of degree more than $c_{trunc}\cdot d$ are pruned away. Conditioning on the event $G$ has $2$-cycle free radius at least $\kappa = \Omega(\log_d n)$,  we can bound the block-value by \[ 
	B_q(M_G) \leq O(\sqrt{n})
	\]
	for $q\ll n^{O(1)}$. As a corollary, conditioned on the event of $2$-cycle free radius being $\Omega(\log_d n)$, with probability at least $1-o_n(1))$, we have $\|M_G\| \leq O(\sqrt{n})$ where $O(\cdot)$ hides some absolute constant independent of $d$ .
\end{lemma}
\begin{proof}
	We analyze the block-value by assuming the technical lemmas highlighted in the above section. Again, we start by casing on the step-label, and assume we are traversing from left $(U_\al)$ to right $(V_\al)$.
	\begin{enumerate}
		\item Singleton-$F$ step: the vertex on the left boundary may be making last but not first appearance, and the vertex on the right boundary may be making first but not last appearance. The edge is a singleton edge that gets assigned a value of $\sqrt{\frac{n}{d}}\cdot \frac{d}{n} \cdot \singdecay$, combining the above gives, \[ 
		\sqrt{n}^2 \cdot  \sqrt{\frac{n}{d}}\cdot \frac{d}{n} \cdot \singdecay \leq \sqrt{n}\cdot \exp(-d)\,;
		\] 
		\item Non-singleton $F$ and $R$ step: this is identical to their analogs in the warm-up example of random $\pm 1$ matrix, and gets a value of $\sqrt{n}$ each;
		\item $H$-step: neither vertices are making first nor last appearances. Since it does not incur $R$-confusions, we do not incur any cost of $2q$ (to be contrasted with the warm-up example). Moreover, a cost of $O(\sqrt{d})$ is sufficient to identify the destination, combining with the edge value gives us a bound of \[ 
		O(\sqrt{d}) \cdot \sqrt{\frac{n}{d}} = O(\sqrt{n})\,;
		\]	
		\item $S$-step: this is identical to the analog in warm-up example, and gives a value of $(2q)^3$;
		\item Summing the above gives us a bound of $O(\sqrt{n})$.
		\end{enumerate}
\end{proof}

\paragraph{Z-shape: an entry-level graph matrix}
Now that we have seen how this strategy applies to arguably the most familiar case of the adjacency matrix,  we will now power up the above machinery with the Z-shape matrix as illustrated in the following diagram. 
\begin{figure}[h!]
    \centering
    \includegraphics[width=160pt]{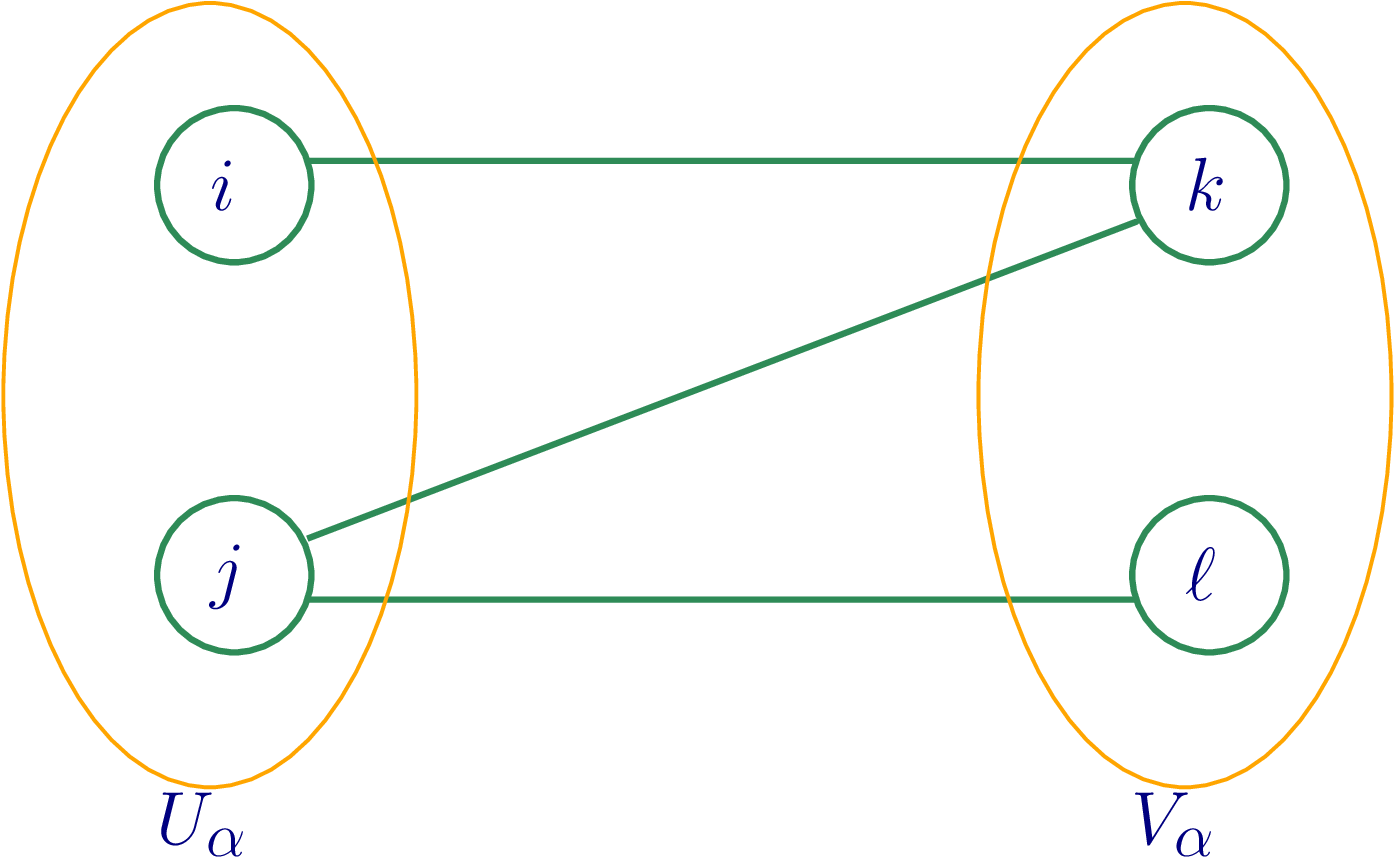}
    \caption{Z-shape\\ $M_Z[(i,j),(k,\ell)] = \chi_G(i,k)\cdot \chi_G(j,k)\cdot \chi_G(j,\ell) $ }
\end{figure}

As defined, $M_Z$ is a matrix in dimension $n(n-1)\times n(n-1)$, and has its rows/columns indexed by ordered pairs of vertices in $[n]$. %It should be pointed out here that this is slightly different from the usual definition for graph matrices that are indexed by sets in the application for Sum-of-Squares lower bounds. However, these are simpler to work with for starters.

\begin{lemma}[Norm bound for $Z$-shape]  \label{lem: Z-shape-bound}
	We can take the block-value of $Z$ shape to be \[ 
	B_q(M_Z) \leq O(\sqrt{n}^2 \cdot \sqrt{\frac{n}{d}})
	\]
 As a corollary, with probability at least $1-o_n(1)$, we have $\|M_Z\| \leq O\left(\sqrt{n}^2 \cdot \sqrt{\frac{n}{d}}\right)$.
\end{lemma}
\begin{proof}
Generalizing from the walk for the adjacency matrix, for the Z-shape, we consider a step in the walk starting from $U_\al$ and analyze the factor-assignment scheme introduced above. 

\begin{enumerate}
	\item Observe that $i,j$ cannot be making their first appearance since they are on the left boundary, and $k,\ell$ cannot be making their last appearance since they are on the right boundary;
	\item For simplicity, we ignore the case of singleton edges as they are analogous to the previous case;
	\item If vertex $k$ is new, we pick up at most a factor of $\sqrt{n}$ for the vertex factor from $k$, and an $O(1)$ factor of $\pur$ since one of the edges $(i,k), (j,k)$ is a surprise visit; since both are non-singleton $F$ edges, $i,j$ do not contribute vertex factors as they are not making their last appearances; we pick up at most another factor of $\sqrt{n}$ from $\ell$  when $j$ is making a middle appearance (either via an $F$ edge to $\ell$, or an $H$ edge). This combines to a factor of $\sqrt{n}^2 q^{O(1)}$ for $q\sim \log^2 n$.
	\item If vertex $k$ is old, i.e., making a middle appearance, \begin{itemize}
		\item In the case $(i,k) = R, (j,k) = H $, $i$ contributes a factor of $\sqrt{n}$ from its last appearance, and the edge $(i,j)$ is ''fixed'' when called upon from $i$ and hence does not contribute any factor; the edge $j,k$ contributes a factor of $\sqrt{\frac{n}{d}}$ and no vertex factor since $k$ has already been identified from the $R$ edge on $(i,j)$; this combines to a total of $\sqrt{n}^2 \cdot \sqrt{\frac{n}{d}}$ assuming $\ell$ is making its first appearance (similar to the case in adjacency matrix), and this is the dominant term;
		\item For the case of $(i,k)= H, (j,k)=R$, it is almost analogous, and we pick up at most a factor of $\sqrt{\frac{n}{d}}$ from the $H$ edge, and $k$ is identified for free since it is reached by an $R$ edge;
		\item  However, we need to further case on whether $j$ is making its last appearance, if not, $\ell$ may be new and contributes a factor of $\sqrt{n}$ (either via vertex factor, or $H$ edge-value combined with vertex factor of $\sqrt{d}$), and this gives a total of $\sqrt{n}\cdot \sqrt{\frac{n}{d}}$. 
		\item In the case where $j$ is making its last appearance, there is no vertex factor from $j$ and $\ell$ cannot be new, hence, it does not contribute a factor of $\sqrt{n}$ (notice $(j,\ell)$ must be $R$). This combines to a total of $\sqrt{\frac{n}{d} }$ (notice the tremendous gap from $\sqrt{n}^2\sqrt{\frac{n}{d}}$ even though there may not be any surprise visit locally);
		\item In the case both $(i,k), (j,k)$ are $H$ edges, we observe that vertex $k$ can be identified at a cost of $\sqrt{d}$ and we additionally pick up an edge-value of $\sqrt{\frac{n}{d}}^2$ from two $H$ edges; moreover, we may pick up an extra factor of $\sqrt{n}$ from $\ell$ since $j$ is making a middle appearance; that said, this combines to a factor of $\sqrt{\frac{n}{d}}^2 \cdot \sqrt{n}\cdot  O(\sqrt{d}) \leq O(\sqrt{\frac{n}{d}}\sqrt{n}^2 ) $. 
	\end{itemize}
	\item Summing over the above choices gives us a bound of $O(\sqrt{n}^2\cdot \sqrt{\frac{n}{d}})$ as there are $O(1)$ step-labelings.
\end{enumerate}
\end{proof}

It should be pointed out that this toy example conveys an essential feature coming from our edge-value assignment scheme: despite the lucrative vertex factor of $n$ for a new vertex, it is not always optimal to go to one! Moreover, we want to emphasize this is a local argument as we are traversing the walk and deciding the next step from the current boundary.

\subsection{Recap of graph matrix norm bounds from prior works}

Before we dig into our new results, we first travel back in time and remind the reader of what is known about graph matrix norm bounds previously. The discovery of vertex-separator being the combinatorial quantity that controls spectral norm bounds of graph matrices was first identified in the investigation of the planted clique problem for $G_{n,1/2}$, 
\begin{definition}[Vertex-separator]
	Given a shape $\al$, we call $S$ a vertex-separator for $\al$ if all path between $U_\al$ (the left boundary) and $V_\al$ (the right boundary) contain at least one vertex in $S$.%that does not intercept with $S$.
\end{definition}
\begin{theorem}[Informal statement of dense norm bounds \cite{AMP20, GT23}  ]
	For a shape $\al$, with high probability over the choice of underlying random graphs from $G_{n,1/2}$, the following bound holds, \[ 
	\|M_\al \| \leq \tilde{O}\left( \max_{S:\text{separator for }\al } \sqrt{n}^{|V(\al)\setminus S|} \cdot \sqrt{n}^{iso(\al)}\right)
	\]
	and more specifically, $\tilde{O}$ hides polylogarithmic factors of $\polylog(n)$ in the dimension, and \emph{polynomial} factor dependence $\poly(|V(\al)|)$ in the size of the shape (which is presumed to be $\tilde{O}(1)$ as the size of shape is limited to $\poly\log(n)$ for planted cliques).
\end{theorem}

More recently, the above bound has been extended into the somewhat sparse regime when $d=\Omega(\poly\log n)$,
\begin{theorem}[Informal statement of sparse norm bounds \cite{JPRTX, GT23} ]
	For a shape $\al$, with high probability over the choice of underlying random graphs from $G_{n,d/n}$, the following bound holds, \[ 
	\|M_\al \| \leq \tilde{O}\left( \max_{S:\text{separator for }\al } \sqrt{n}^{|V(\al)\setminus S|} \cdot \left(\sqrt{\frac{n}{d}}\right)^{E(S)} \cdot \sqrt{n}^{iso(\al)}\right)
	\]
	with $\tilde{O}$ hiding the identical polylog dependence on $n$ and polynomial factor in the size of shape as the previous version.
\end{theorem}
One may notice that the above bound in fact continues to hold in the ultra-sparse regime of $d=O(1)$, however, these bounds become less meaningful here due to the loss of $\polylog$ factors that render the bound comparable, if not inferior,  to the naive bound from a row-sum bound, losing the usual \emph{square-root} gain. 

Although the norm bound falls short in the regime we are interested in due to polylog factors, the prior works indeed carry an important message about the polynomial-factor dependence.\begin{enumerate}
	\item Each vertex outside the separator contributes a factor of $\sqrt{n}$ (if not isolated);
	\item And in the sparse regime, each edge inside the separator contributes a factor of $\sqrt{\frac{n}{d}}$.
\end{enumerate}
A starting question for our sparse norm bounds is whether these leading-factors continue to be the only dominant terms in the ultra-sparse regime, or whether they get overwhelmed by the ''prior'' lower-order dependence such as the polylog factors. On a high level, we resolve the question with the following answer,\begin{enumerate}
	\item The above two factors are indeed the only dominant terms even in the ultra-sparse regime;
	\item  Most lower-order dependence is unnecessary while some lower-order dependence of $\polylog(n)$ factors is needed for shapes due to floating and dangling structures.
\end{enumerate}

\subsection{Our main theorems for graph matrix norm bounds}
%\paragraph{More fine-grained definitions for shape}
To enable us to obtain a more fine-grained understanding of the spectral behavior graph matrices, it turns out to be crucial that we give a more detailed decomposition of each component in a shape, as they are now the determining factor for matrix norm bounds below $\polylog$ order dependence that has been overlooked in prior works.
\begin{definition}[Floating vertex, and floating component]
	For a given shape $\al$, we call a connected component $Q\subseteq V(\al)$ a floating component if $Q$ is not connected to $U_\al \cup V_\al$. For each floating component, we arbitrarily pick a vertex to be the first vertex of the component in the encoding. Let $E_\al(\float)\subseteq E_\al$ be the set of edges in a floating component.
\end{definition}
\begin{definition}[Floating factor]
    We use $\float(\al)$ to capture the extra overhead for norm bounds due to floating components.
\end{definition}

%\begin{definition}[Simple path]
%	We call a path simple if it does not traverse any edge more than once. However, it may traverse a vertex more than once.
%\end{definition}
\begin{definition}[Dangling and non-dangling vertex]
	For a given shape $\al$, we call a vertex $v\in V(\al)$ a ''dangling'' vertex if it is not on any simple path from $U_\al$ to $V_\al$; otherwise, we call the vertex ''non-dangling''.
\end{definition}
\begin{remark}
    Note that any vertex of degree $1$ outside $U_\al\cup V_\al$ is a dangling vertex.  
\end{remark}
The following is an example of a shape with a dangling vertex.
\begin{figure}[h!]
     \centering
    \includegraphics[width=160pt]{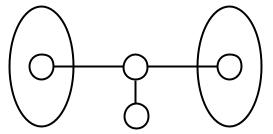}
     \caption{Line-graph with a dangling vertex \\ $M_{dang}[u,v] =\sum_{t_1, t_2\neq u,v} \chi_G(u,t_1)\cdot \chi_G(v,t_1) \cdot \chi_G(t_1, t_2)$} 
\end{figure}

\begin{definition}[Dangling branch]
	 Given a shape $\al$, for any non-dangling vertex $v$ incident to dangling vertices, consider the collection of dangling-vertices reachable from $v$ without passing through any other non-dangling vertex, by possibly removing excess edges that lead to visited vertices, this is a tree of dangling-vertices rooted at $v$, and we can pick an arbitrary order to go through the leaves, and identify a tree-path to the leaf such that each path-segment is a branch whose \begin{enumerate}
	 	\item  branch-head is either a vertex in some previously identified branch or $v$;
	 	\item  branch-tail is a leaf.
	 \end{enumerate} 
	 Moreover, we write $\branch(\al)$ as a collection of dangling branches in $\al$.
\end{definition}	
\begin{definition}[Dangling factor]
    We use the factor $\dang(\al)$ to capture the extra overhead for norm bounds due to dangling branches.
\end{definition}
\begin{remark}
	For clarity, the reader should keep in mind the example where each non-floating degree-$1$ vertex in $V(\al)\setminus U_\al\cup V_\al$ is dangling and a branch-tail.
\end{remark}

For some $\delta>0$, for any shape $\al$ with $|V(\al)|\leq n^{\delta}$, we obtain the following theorems,
\begin{theorem}[User-friendly version without dangling/floating vertices]
    Given a shape $\al$ with no floating or dangling component, %and let $G$ be a sample from $G_{n,d'/n}$ with high-degree vertices (degree at least $d=\ctrunc \cdot d$) removed, let $\cycle$ be the indicator function for the absence of nearby $2$-cycle in $G$, 
    we define \begin{align*}
		B_q(\al) = \max_{S:\text{separator}}  \cnorm^{|V(\al)|} \cdot \cnorm^{|E(S)|} \cdot \sqrt{n}^{|V(
	\al)\setminus S|}\left(\sqrt{\frac{n}{d}}\right)^{|E(S)|}\sqrt{n}^{|I(\al)|}
	\end{align*}
    where $\cnorm>2$ is an absolute constant independent of $d$ and $\al$ and $I(\al)$ is the set of isolated vertices in $\al$ outside $U_\al\cup V_\al$. 
    
    Recall that $\cycle$ is the indicator function for $G$ having 2-cycle free radius at least $\kappa = 0.3\log_d n$. For any $\epsilon>0$, and for any $q=\Omega(\sqrt{|U_\al||V_\al|} \log^2n)$ and $q \leq n^{O(\delta)} $, 	\[ 
	\E_{G}\left[\Tr\left((M_\al M_\al^T)^{q}\right)\cdot \cycle \right] \leq \bigg((1+\epsilon)B_q(\al) \bigg)^{2q}%\,.
	\]
%	 \[ 
%	\E_G[\Tr(M_\alM_\al^T) ] 
%	\]
where $G$ is sampled from $G_{n,d/n}$ with high-degree vertices (degree at least $d=\ctrunc \cdot d$) removed.

As a corollary, with probability at least $1-o_n(1)$ , \[ \|M_\al\|\leq (1+\eps) \cdot B_q(\al)\,.\] 
\end{theorem}

%\begin{theorem}[User-friendly with dangling/floating vertices]
%	Given a shape $\al$, let $\float$ be the number of floating components outside the separator, and let $V(\float)$ be the number of of vertices in the floating components, let $\dang $ be the number of dangling vertices, we define 
%	\begin{align*}
%		B_q(\al) &=   \cnorm^{E(\al)} \max_{S:\text{separator}} \sqrt{n}^{|V(\al)\setminus S|}\left(\sqrt{\frac{n}{d}}\right)^{E(S)}\sqrt{n}^{I(\al)}
%\cdot \sqrt{d}^{\dang }\cdot  \sqrt{2|V(\al)|q}^{|\float|}
%	\end{align*} 
%for some absolute constant $\cnorm>2$ (independent of $d$ and $\al$), and $I(\al)$ is the set of isolated vertices in $\al$ outside $U_\al\cup V_\al$. 
%For any $\epsilon>0$, and for any $q=\Theta(\sqrt{|U_\al||V_\al|} \log^2n)$, 	\[ 
%	\E_{G}\left[\Tr\left((M_\al M_\al^T)^{q}\right)\cdot \cycle \right] \leq \bigg((1+\epsilon)B_q(\al) \bigg)^{2q}
%	\]
%\end{theorem}
\begin{theorem}[Full-power] \label{thm:norm-theorem}
    Given a shape $\al$, letting $\float$ be the floating factor and letting $\dang(\al\setminus S)$ be dangling factor outside the separator $S$, we define 
	\begin{align*}
		B_q(\al) &=    \max_{S:\text{separator}} \cnorm^{|V(\al)| } \cdot \cnorm^{|E(S)|} \cdot  \sqrt{n}^{|V(\al)\setminus S|}\left(\sqrt{\frac{n}{d}}\right)^{|E(S)|}\sqrt{n}^{|I(\al)|}
\cdot \dang(\al\setminus S ) \cdot \float(\al) 
	\end{align*} 
where $\cnorm>2$ is an absolute constant independent of $d$ and $\al$, $F(\al)$ is the collection of floating components in $\al$, $\branch(\al\setminus S )$ is the collection of dangling branches outside the separator, and $I(\al)$ is the set of isolated vertices in $\al$ outside $U_\al\cup V_\al$.

Additionally, the dangling/floating factor are bounded by the following, 
\[ 
	\dang(\al\setminus S) \leq \prod_{i\in \branch(\al\setminus S)} \min(\sqrt{2|V(\al)|q}, \sqrt{d}^{|\branch(i)|})\,,
	\]
and \[ 
	\float(\al)\leq \prod_{C_i\in F(\al)} \max \left(2|V(\al)|q_\al, \sqrt{n}\cdot \singdecay^{|E(C_i)|} \cdot 1_{V(C_i)\cap S = \emptyset } \right) \,,
	\]
For any $\epsilon>0$, and for any $q=\Theta(\sqrt{|U_\al||V_\al|} \log^2n)$ and $q \leq n^{O(\delta)}$, we have the following bound	\[ 
	\E_{G}\left[\Tr\left((M_\al M_\al^T)^{q}\right)\cdot \cycle\right] \leq \bigg((1+\epsilon)B_q(\al) \bigg)^{2q}%\,.
	\]
 where $G$ is sampled from $G_{n,d/n}$ with high-degree vertices (degree at least $c_{degree} \cdot d$) removed.
 
 As a corollary, with probability at least $1-o_n(1)$, \[ \|M_\al\|\leq (1+\eps) \cdot B_q(\al)\,.\] 
\end{theorem}

\begin{remark}
    Throughout this work, we will pick $q = c (|U_\al||V_\al|)\log^2 n $ for any shape $\al$.
\end{remark}

\begin{remark}
For the purpose of our application in Sum-of-Squares lower bounds in this work, a size bound of $|V(\al)|\leq O(\log^{2} n)$ already suffices as we truncate at $|V(\al)|\leq c\cdot \dsos\log n$. However, we state the stronger bound here as we would like to emphasize that the norm bound analysis is crucially no longer the bottleneck for proving SoS lower bounds anymore, even in the \emph{ultra-sparse} regime.
\end{remark}
\begin{remark}
    The bottleneck of our probabilistic norm bound concentration comes from the conditioning of a random graph sample having 2-cycle free radius at least $\kappa = 0.3\log_d n$. %not having nearby-$2$-cycle. 
    Conditioned on that event, our bound holds with probability at least $1-n^{-100}$, and can be readily strengthened to $1-n^{-\omega(1)}$.
\end{remark}
    As a careful reader may point out that in the user-friendly version of our result, the block-value bound $B_q$ despite being stated with dependence in $q$, does not ultimately depend on $q$. However, the final bound in the full-power version does have dependence in $q$. We opt to keep the notation $B_q$ for consistency. On the other hand, we would like to emphasize that the dependence on $q$ factor is the usual "log"-factor for random matrix norm bounds, and our main result reveals that the exact dependence on $q$ can be improved for a large family of shapes, though some dependence is also necessary for specific shapes due to special structural properties which we call "dangling" and "floating".

\begin{example}[Shape with log factors]
    Consider the following shape of a single floating triangle with $|U_\al|=|V_\al|=|U_\al\cap V_\al|=1$, i.e. this is a diagonal matrix with entries \[M[i,i]= \sum_{\text{injective }\phi(a,b,c) \in [n]\setminus\{i\} } \chi_G(\phi(a),\phi(b))\cdot \chi_G(\phi(b),\phi(c)) \cdot \chi_G(\phi(a),\phi(c)) \] 
    for any i.
    
    Applying our main theorem to this shape gives a norm bound of $\|M\| \leq O(1) \cdot \sqrt{n}^3 \log^2 n$ as the separator is mandatory as $U_\al = V_\al$, and we pick up an extra floating factor of $q = O(\log^2 n)$. Note that some log factor here is needed as for any fixed diagonal, the entry is a sum of mean $0$ and variance $\Theta(n^3)$, and we are looking for a bound that holds with probability $1-o_n(1)$. 
     \begin{figure}[h]
     \centering
    \includegraphics[width=150pt]{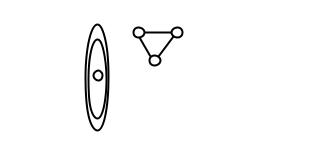}
     \caption{Floating triangle} 
     \end{figure}
 \end{example}
\begin{remark}
    For our application, for shapes with floating and dangling factors, we do not optimize for the tight polylog dependence (or generally, dependence in $q$) once they are present.
\end{remark}

\subsection{Bridging vertex separator, step-labeling and norm bound}
% More than just illustrating the power of our factor assignment scheme, these two examples we showcase above may also serve for the broader theme in exposing the connection between graph matrix norm bounds and vertex separator.

In the following section, we will show how we rediscover the polynomial dependence as predicted from prior norm bounds in the ultra-sparse regime, especially in the context of our new technique via step-labeling. We first remind the reader that our main theorem (using the vanilla version) gives the following bound for any shape $\al$,  \begin{align*}
		B_q(\al) = \max_{S:\text{separator}}  \cnorm^{|V(\al)|} \cdot \cnorm^{|E(S)|} \cdot \sqrt{n}^{|V(
	\al)\setminus S|}\left(\sqrt{\frac{n}{d}}\right)^{|E(S)|}\sqrt{n}^{|I(\al)|}
	\end{align*}

Let's now zoom back into the examples that we just discuss, and see how the previous bound may fall out of an application of the main theorem in its full generality. We first consider the following construction of vertex-separators from the edge-labeling for the upcoming block,
% We first consider the construction of the following set given the edge-labeling for the upcoming block,
\begin{enumerate}
\item Include any vertex incident to both \emph{non-singleton} $F$ and $R$ edges into $S$;
		\item Include any vertex in $U_\al$ incident to some \emph{non-singleton} $F$ edge into $S$;
		\item Include any vertex in $V_\al$ incident to some $R$ edge into $S$;
		\item Include any vertex incident to an $H$ edge into $S$.
\end{enumerate}
\begin{figure}[h!]
    \begin{minipage}{0.48\textwidth}
     \centering
    \includegraphics[width=160pt]{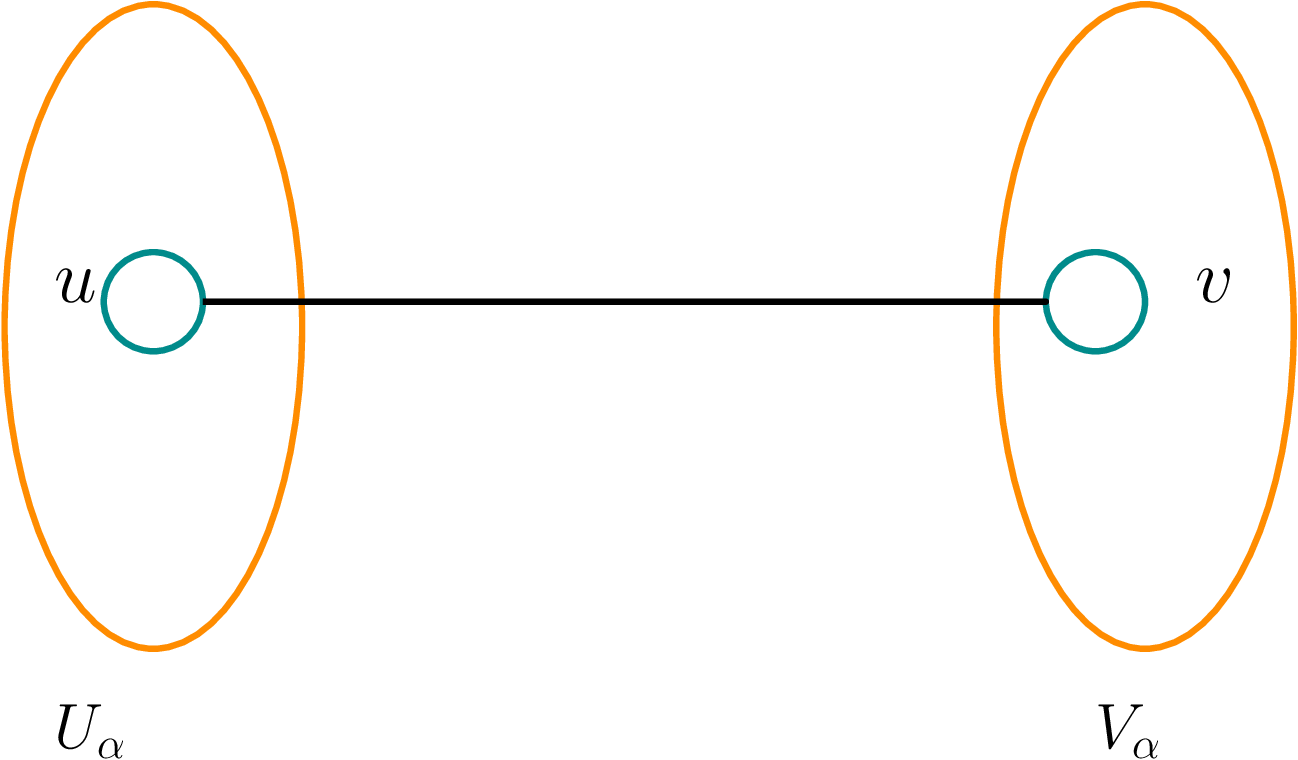}
    \caption{Line-graph \\ $M_{line}[u,v] =\chi_G(u,v) $}
   \end{minipage}\hfill
   \begin{minipage}{0.48\textwidth}
     \centering
       \includegraphics[width=160pt]{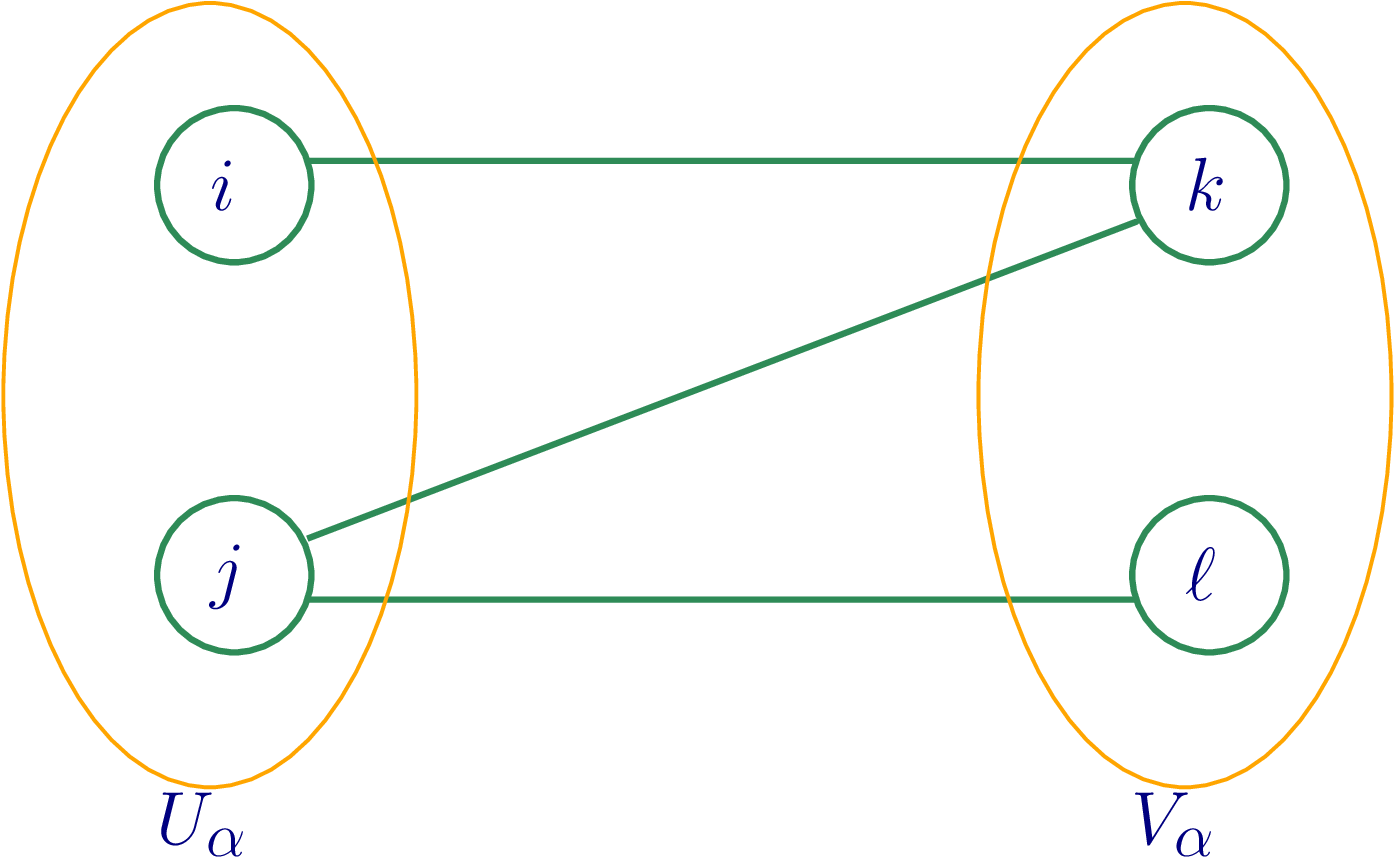}
    \caption{Z-shape\\ $M_Z[(i,j),(k,\ell)] = \chi_G(i,k)\cdot \chi_G(j,k)\cdot \chi_G(j,\ell) $ }
   \end{minipage}
\end{figure}
	
\paragraph{Adjacency matrix revisited} In the case of adjacency matrix, the maximal-value edge-labeling is either a (non-singleton) $F$ or an $R$ edge, in which case, one (but not both) of the vertices is included into the separator, while the other vertex remains outside and thus gives a factor of $\sqrt{n}$. In the case of $H$ edges, we note that this can be reduced into the $F/R$ case by turning the edge into an $F/R$ which yields the same block-value while we now have exactly one vertex outside the separator again. As a result, in all these edge-labelings, there is a vertex separator either $u\in U_\al$ or $v\in V_\al$ that can be constructed and
the block-value is dictated by the corresponding vertex separator: each vertex outside the separator continues to contribute a factor of $\sqrt{n}$ as they do in the dense case. In this case, applying the main theorem gives a bound of 
 \begin{align*}
		B_q(M_{line}) = \max_{S:\text{separator}}  \cnorm^{|V(\al)|} \cdot \cnorm^{|E(S)|} \cdot \sqrt{n}^{|V(
	\al)\setminus S|}\left(\sqrt{\frac{n}{d}}\right)^{|E(S)|}\sqrt{n}^{|I(\al)|} = \cnorm^3 \cdot \sqrt{n} = O(\sqrt{n})
	\end{align*}

\paragraph{$Z$-shape revisited} The example of adjacency matrix, despite its simplicity, falls short in conveying another important message that each edge inside the separator contributes a factor of $\sqrt{\frac{n}{d}}$. However, this phenomena has already manifested itself in the $Z$-shape matrix. Following the dominant edge-labeling in the $Z$ shape that gives labels $R, H,F$ to all three edges from top to bottom, the separator is thus given by $(j,k)$ and the edge inside contributes a factor of $\sqrt{\frac{n}{d}}$ while the two vertices outside contribute a factor of $\sqrt{n}$ each.

 In this case, applying the main theorem gives a bound of 
 \begin{align*}
		B_q(M_{Z}) = \max_{S:\text{separator}}  \cnorm^{|V(\al)|} \cdot \cnorm^{|E(S)|} \cdot \sqrt{n}^{|V(\al)\setminus S|}\left(\sqrt{\frac{n}{d}}\right)^{|E(S)|}\sqrt{n}^{|I(\al)|}=O\left(\sqrt{n}^2 \cdot \sqrt{\frac{n}{d}}\right)
\end{align*}
\paragraph{Singleton edges affecting the separator?} It should be acknowledged that the idea of vertex separators falls slightly short in capturing the potential effect due to singleton edges. In the prior settings, any edge that appears only one throughout the walk would immediately zero out the contribution from the particular walk as this is taking expectation over symmetric random variables with mean $0$. However, this is not true in our ultra-sparse setting since we apply conditioning on the degree boundedness of each vertex and the absence of dense subgraph, which unfortunately destroys the perfect balance of the random variable- it is no longer mean $0$. That said, as we observe in the previous examples, the bound continue to hold as long as we get a singleton decay that is at most $\frac{1}{\sqrt{d}}$. Put in the context of vertex separators, this reveals that despite the fact that the above construction of $S$ from any edge-labeling does not immediately produce a vertex separator for the shape, its value may still be bounded if we can consider an alternate edge-labeling that 1) yields a higher block-value, and 2) corresponds to an honest vertex separator of the shape if it does not contain any singleton edge.

In our main technical analysis, we show that this is indeed true for components that are connected to $U_\al$ and $V_\al$, while singleton edges do indeed have a non-trivial influence when the shape contains floating components, in particular, tree-like floating components. Fortunately in our setting for SoS lower bounds, this turns out to influence a collection of shapes that are previously already ''small'' in norm, and thus, does not pose qualitively significant differences for our application.

\paragraph{$R-H-F$ structure lemma}  More than just illustrating the power of our factor assignment scheme, these two examples we showcase above may also serve for the broader theme in exposing the connection between step-labeling and vertex separator, and ultimately how this connection enables us to obtain tight norm bounds. As a by-product of our analysis that proceeds by understanding the status of each edge, as opposed to simply by casing on the status of each vertex whether inside the separator, we in fact obtain a more fine-grained understanding of how a walk that maximizes the norm should behave. This is encapsulated in the following structure lemma that we discover. For the clarify of presentation, we restrict our attention to the most interesting collection of shapes that do not contain floating and dangling component, while it can be extended to capture those shapes in a straightforward fashion. \begin{lemma}[Structure lemma for graph matrix (See formal discussion in  \pref{sec:root-n-outside})] Given a shape $\al$, and for a fixed separator $S$ of $\al$, fix the traversal direction from $U_\al$ to $V_\al$, the maximum-edge-labeling is given by \begin{enumerate}
    \item Label each edge that can be reached from $U_\al$ without passing through $S$ an $R$ edge;
    \item Label each edge inside the separator an $H$ edge;
    \item Label each edge that can be reached from $V_\al$ without passing through $S$ an $F$ edge.
\end{enumerate}
\end{lemma}
As evident in the two examples we showcase, any $F$ edge before reaching the separator in fact corresponds to a surprise visit so that we can immediately improve the value locally by flipping it to an $R$ edge, assuming there is a large enough gap created by surprise visits, which is $2|V(\al)|q_\al \ll n^{\delta}$ for some tiny constant $\delta>0$ versus the $\sqrt{n}$ contribution if we avoid the surprise visit. Analogously, any $R$ edge that is used after we pass through the separator render a gap we can improve as flipping it to an $F$ edge can be shown to give us an additional vertex contribution of $\sqrt{n}$. Finally, each edge in $E(S)$, if turned to an $H$ edge from $F/R$, now contributes a factor of $\sqrt{\frac{n}{d}}$ which again allows us to attain a higher block-value.

\subsection{Overview of proof for unforced-return bounds and high-mul balance}
At this point, it should be noted  that there is an adjustment which we make to this high-level picture. For the polynomial factors, it is most convenient to have the $R$ edges be the last time the edge appears. However, for analyzing the factors of $d$ and $q$, it is more convenient to instead have the $R$ edge be the second time that the edge appears and have the subsequent appearances be $H$ edges. We make this switch for the remainder of the analyzing vertex-encoding-cost in the norm bound analysis.
% the swap should be only for bounding vertex-cost? we need to revert back later in smvs component
%With this high-level picture 
With this in mind, we now return to fill in the technical details that we defer from the earlier section, in particular, \begin{center}
    Why can an $R$-step be specified in $O(1)$, i.e., why is an edge being used the second-time forced?
\end{center}and additionally,
\begin{center}
    Why is a cost of $\sqrt{d}$ sufficient to specify an $H$ edge, which has already appeared $\geq 2$ times?
\end{center}

\subsubsection{Forced-return bound}
% \paragraph{Swapping last-use cost with second-use cost} At this point, we remind the reader that an $R$-step as defined earlier corresponds to an edge (or a random variable) appears for the last time in the walk. However, for an edge that makes a middle appearance in the walk, when bounding the vertex-factor for specifying this edge, we will swap the the step-label for its second and last appearance. In other words, we will instead show an edge appearing for "the second-time" is fixed (which would have been assigned an $H$-step), and assign its original budget as an $H$-step to its last appearance. With this swap-operation, the main focus for bounding the vertex-cost is now reduced to the "second"-appearance of an edge.

\paragraph{Starting from line-graph}
To see that an edge being used the second-time is fixed, we may first trace back to what happens if the walk takes place on a tree where things are considerably simpler. In this case, each step in the walk either goes away from the root or backtracks towards it, and moreover, each vertex has only $1$ edge which backtracks to the root (which is the edge that discovers the vertex). If we further know that each edge appears only twice in the walk, i.e., no branch is traversed repeatedly, then this is indeed the second-appearance of the edge . As a result, return is forced if the walk proceeds on a tree, or more generally, on an arbitrary graph where the walk proceeds in a tree-like manner.

However, although a sparse random graph is locally tree-like, since we take a long walk of length $q=\Theta(\log n)$, it is unrealistic for us to impose a tree-like assumption on the walk. That said, to generalize the argument for cycles, we observe that each cycle formed in the walk results from an edge leading to a visited vertex, which we call a \emph{surprise visit}, and such an edge can be much more succinctly encoded using a label in $q = \Theta(\log n)$ than its counterpart of a new edge leading to a new vertex. From this perspective, the tremendous gap opened up prompts us to show that as long as each surprise visit does not create ''too many'' unforced-returns, as we  we can then use the gap from corresponding surprise visits to encode auxiliary data so that each return, despite the possible confusion caused by surprise visits, continues to be fixed.

Towards this end, we focus on how the number of unforced return legs from a vertex $v$ may change if we depart and arrive back at the vertex again assuming there is no $H$ edge involved in the walk and each edge appears at most two times. Suppose we depart the vertex by using a new $F$ edge, the subsequent arrival either closes an $R$ edge, or it is a surprise visit which comes with a gap. Moreover, the subsequent surprise visit may see an increment of at most $2$ unforced-return from the particular vertex $v$: with one being the most recent departure $F$ edge, and another one being the surprise visit in the current arrival. That said, each surprise visit can be used to charge at most $2$ unforced-returns. This argument turns out to be surprisingly powerful beyond the simple line-graph case, and this formalized in Section \ref{subsect:return-for-one-in-one-out } and further extended to capture the influence from $H$ edges by considering $H$-component as opposed to a single vertex.

\paragraph{Extending to branching vertices with mirror copy}

The above argument addressees unforced-return restricted to vertices that push out at most $1$ $F$ edge when it is on the boundary, however, this is not always the case for graph matrices. For example, consider the middle vertex in the following shape: fix the traversal from left to right, it pushes out $2$ edges in each block. For convenience, we call each such vertex a \emph{branching} vertex. They fall out of reach of the argument we outline above, and we indeed need an upgraded argument to take into account of such vertices, as they are also crucially why the graph matrix norm bounds in the dense setting is governed by \emph{Minimum Vertex Separator}, i.e., the separator of the shape with the fewest number of vertices. 

\begin{figure}[h!]
    \centering
    \includegraphics[width=160pt]{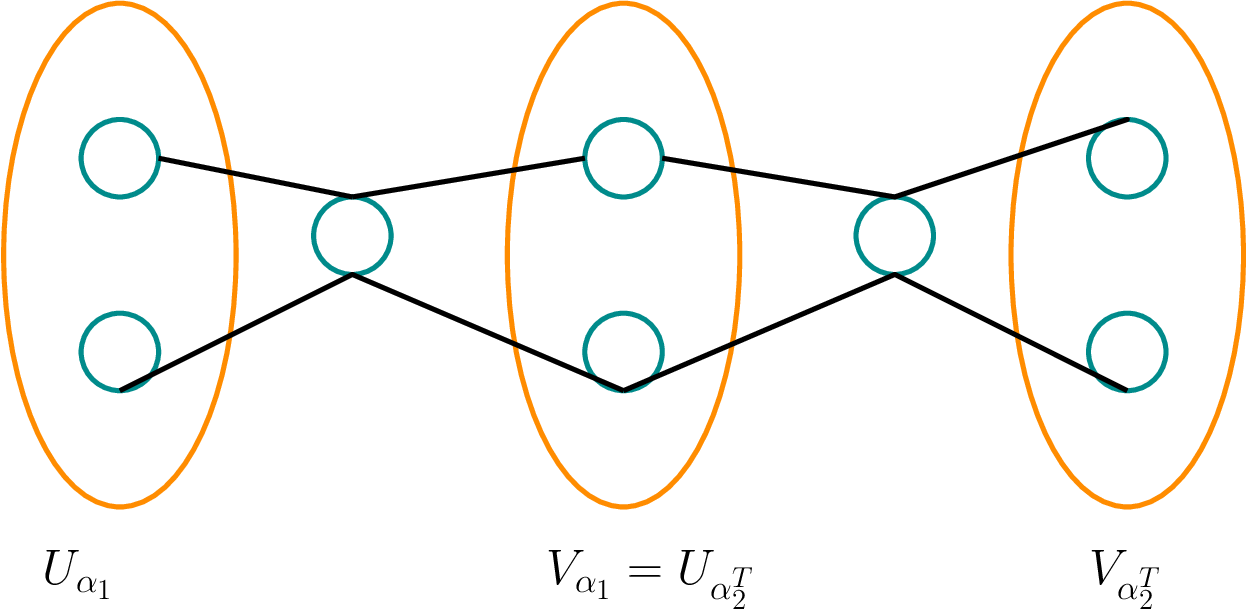}
    % \caption{Z-shape\\ $M_Z[(i,j),(k,\ell)] = \chi_G(i,k)\cdot \chi_G(j,k)\cdot \chi_G(j,\ell) $ }
\end{figure}
It turns out that such vertices can be easily addressed, at least in the case when apply trace-method for a single shape, since we know that each block of $\al$ is followed by its symmetric copy $\al^T$ in the subsequent block, and this enables us to use a completely local argument for understanding the influences of such vertices. \begin{center}
    For each branching vertex, look at its copy in the subsequent block, either all the paths are mirror-copies of each other, or there is a surprise visit (since locally we have cycles).
\end{center}
With the above observation, in the case each path is a mirror-copy of each other, i.e., each path goes to the boundary and immediately backtracks, each $F$ edge that gets opened up is immediately closed, and hence does not contribute any confusion to unforced-returns. On the other hand, suppose some path does not backtrack immediately, there is a surprise visit along the path that can be assigned to \emph{charge} the corresponding confusion it may incur, and we formalize this via the \emph{missed-immediate-return} function.

\subsubsection{High-mul balance: $\sqrt{d}$ is sufficient for $H$-step}\label{sec:high-mul-overview}
With the unforced return bound settled, we may now proceed to potentially the last missing component to complete our matrix norm bounds for graph matrix of a fixed shape. Under our bounded-degree graph assumption of the underlying random graph sample, it is not immediately true that we only consider walks such that each vertex has bounded degree. Throughout this section, for convenience, let $d$ be the threshold of degree truncation. It is not true that we constrain our walks to use at most $d$ edges for each vertex: in fact, each vertex may have $q\sim \Theta(\poly\log(n))$ edges in the walk, however, the bounded degree assumption allows us to deduce a decay whenever some vertex has degree more than $d$: in particular, our edge-value bound holds by assuming each edge appears in the random graph, and thus gives a spike of $\sqrt{\frac{n}{d}}$. However, given the prior that some edge is missing, we can in fact assign a decay for each edge that is missing, and creates a tremendous gap as each non-edge contributes a factor of $-\sqrt{\frac{d}{n}}$ instead. That said, it still suffices for us to assume each vertex to be incident to at most $d$ edges in the walk.

However, this is not quite sufficient: for example, in the case of adjacency matrix, this allows us to avoid a pessimistic bound of $O(\poly\log n)$, while not entirely sufficient towards getting a meaningful bound of $O(\sqrt{d})$, as a priori, a label in $[d]$ is needed for each $H$ edge which would then lead to the trivial bound of $O(d)$ equivalent to a row-sum bound. That said, getting the extra $\emph{square-root}$ gain is crucial in the ultra-sparse regime. We now restate the proposition from earlier section, 
 \begin{proposition}[Restatement of \pref{prop:root-d-suffices}]
A cost of $\sqrt{c_{degree} \cdot  d} =  O(\sqrt{d})$ is sufficient for identifying the destination of an $H$-step. 	
\end{proposition}

Towards this goal, we apply our conditioning that the graph has no nearby-$2$-cycle at radius $c\log_d n$ for some constant $c>0$, and give a potential-function argument for showing a cost of $\sqrt{d}$ is sufficient. For the example of the adjacency matrix, assuming there is no surprise visit involved (as otherwise we get an extra gap), we can split the walks into intervals of at most $\frac{c}{2}\log_d n$. For each interval, suppose all edges used throughout the interval have been revealed by the beginning of the interval, it suffices for us to identify a label in $[q]$ that represents the end of the interval, which is possible as we assume each edge has been revealed. Since the walk is locally $2$-cycle free within the interval, a label in $[2]$ is sufficient to identify the non-backtracking segment from the start to the end of the interval, and observe that
\begin{enumerate}
    \item We can partition the chunk of $H$ edges into either main-path (non-backtracking segment) and off-track segment attached to the main-path (backtracking segment);
    \item Each $H$ edge on the main-path is fixed once the start and end point is identified;
    \item Each backtracking edge requires $d$ going away off-track. This is fixed when traversing back towards the main path and thus give an average of $\sqrt{d}$;
    \item The additional label in $[q]$ is then offset by the long chunk as we have $q^{O(1/\log_d) n} \leq 1+\eps$ for $q\leq n^{\delta}$ for some constant $\delta \leq O(\frac{1}{\log d})$. 
\end{enumerate}

\subsection{Graph matrix norm bound for grouped shapes}

\paragraph{Single trace-method for a collection of various shapes}
As we shed light upon in the overview for our PSDness analysis, our moment matrix %is constituted 
consists of $\Omega(\exp(\dsos\log n))$ many shapes. Despite the intuition in prior works that apply conditioning to reduce dense shapes to sparse shapes, the analysis in \cite{JPRTX} inevitably loses extra $\polylog(n)$ factors due to its edge-reservation idea. On a high-level, for ease of the technical analysis, it is helpful in the prior works to assume the shapes that arise in the charging to have various ''nice'' properties that preserve the prescribed \emph{minimum vertex separator}. This proceeds in the manner of putting aside $O(V(\al))$ edges, while applying conditioning on the rest. Unfortunately, this inevitably leads to a loss of $\polylog$ factors as there may still be $O(|V(\al)|)$ excess edges, as each of which requires a potential $|V(\al)|=\Theta(\poly\log n)$  factor for encoding.

Though it is foreseeable that a more powerful conditioning lemma equipped with combinatorial charging argument would carry through, we observe in this work that it is in fact more convenient to group shapes together (as well as with their shape coefficients from pseudo-calibration) and apply one single trace-method calculation as our spectral norm bounds machinery is amenable to be applied on a sum of graph matrices. In particular, for our final charging, we consider graph matrix for grouped shapes defined as the following,

\begin{definition}[Graph matrix for grouped shapes]
    Consider $\calP$ a collection of shapes, we define the graph matrix of grouped shapes of $\calP$ as \[
    M_{\calP} = \sum_{\al\in \calP } M_\al
    \]
\end{definition}

% It turns out that for convenience of our technical analysis to potentially avoid the blow-up from a union-bound over various shapes that are to be counted, it would be more helpful to apply trace-method (in a single trial) on a sum of graph matrices (in particular, with pseudo-calibrated coefficient as well in our final application). We formalize this via a partition of shapes into active profiles which crudely partition shapes based on which vertices in $U_\tau$ and $V_\tau$ are incident to vertices outside $U_\tau\cap V_\tau$. We defer the formal definition to the technical analysis. 

The main motivation behind this grouping is the observation that the cost of identifying ''excess edges'' in the outer union bound is in fact the fundamentally the same as having \emph{surprise visits} in our trace-method calculation, which we know by now that each comes with a gap! That said, we can hope to offset the drawing cost by the gap from surprise visit, and as we will soon see, the cost of encoding these ''excess edges'' is in the end $O(1)$!

To shed light into our analysis, we note that as we apply trace-method calculation on $M_{\calP}$ directly without partitioning on the shapes, a block-walk may indeed behave as the following,

\begin{figure}[h!]
    \centering
    \includegraphics[width=320pt]{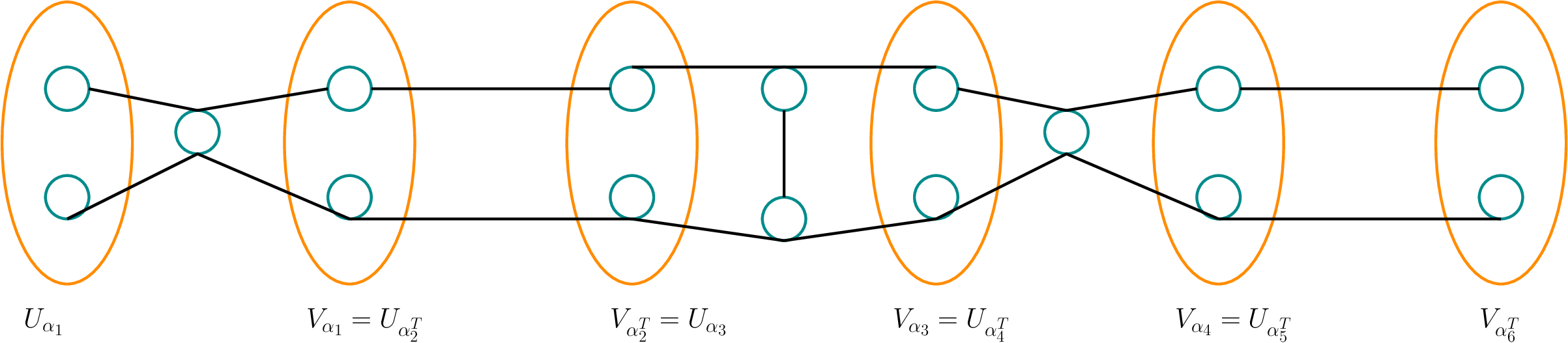}
     \caption{ Block walk for grouped shapes }
    \label{fig:perm_jump}
\end{figure}

\paragraph{Strategy overview for bounding drawing factor}
At this point, we are ready to extend the block-value analysis to capture additionally the cost of drawing out a shape, and avoid the ''potential'' encoding cost should one try to apply a triangle inequality in a naive way. 

We remind the reader that the encoder knows the shape to-be drawn out, and the decoder has a map of the random graph sample being explored up-to the current timestamp in the walk as well as the edge-labeling of the upcoming block. The task for the decoder is to assign each active vertex in the current boundary a label to help draw out the upcoming shape.

 Out strategy proceeds as the following,
 \begin{enumerate}
	\item We extend the block-value bound to additionally take into account of the cost of drawing out the upcoming shape;
	\item Given the labeling $\calL_\al$ for the current $\al$ block (which is known to the encoder), we consider a subshape $\al_{H}$ obtained by removing edges from $\al$ that receive an $F$ label in $\calL$, and remove isolated vertices not in $U_{\al_{H}}$, and this produces us a list of wedges $W =  \{W_i\}$; 
%	\item Assign a label for each wedge in $\al_{H}$, including isolated vertices in $U_{\al_H}$, and let this list be $L$;   
	\item (Bundle) For each wedge $W$, we first identify it as an edge set $W\subseteq  \binom{n}{2}$ (via ~\pref{claim:bundle-bound}) ;
	\item (Draw) The rest of the shape can be easily drawn out as it either requires a single constant (i.e. going to a new vertex), or it leads to a seen vertex with a new edge and hence is accounted for as a surprise visit.
\end{enumerate}

\paragraph{Branch-chasing: forced-return bound without mirror copy} Interestingly, as we extend our block-value bound in order to apply graph matrix norm bounds on a collection of shapes as opposed to a single shape, some component of our prior analysis in fact slightly falls short: in particular, our forced return argument for branching vertex. Note that previously as we apply trace method for a single shape, we know that each branching vertex that pushes out multiple edges in the current block inevitably faces a pull-back in the subsequent block, and we use this observation to give a lightweight local argument that such branching vertex does not interefer with out forced return bound. However, as shown in the example of \cref{fig:perm_jump}, this is not true for branching vertex anymore (for example, the middle vertex in $\al_1$ pushes out two edges while they don't get immediately pulled back in the subsequent block). In fact, not only is it not true that pull-backs do not immediately happen in the subsequent block, they may never happen either:  consider the example where the first block consists of branching vertices pushing out, while any subsequent block proceeds via parallel disjoint paths, and therefore, the paths never rejoin.

\paragraph{Chase confusion} We now make this issue concrete. Consider the following example at \cref{fig:chase-confusion}, suppose $v$ at the first block branches out to $a$ and $b$, even though the path to $a$ immediately backtracks, the path to $b$ keeps pushing forward. At the end of the second block, if an $R$ edge is used from $v$ for the third-block, it is potentially unclear which edge it points to, as it can be either pointing to $s_1$ or $b$, which is then considered a $R$-chase as we picture it as a walker chasing another. 

Put simply in the example of tree, our prior unforced return bound can be seen as saying for each edge, when traversed the second time, it must be traversed in the opposite direction from when it is first created, as seen in the example of backtracking towards the root (unless there is a cycle that can be charged to it in the general case).  However, in this setting, this is no longer true as we may traverse the edge $(v,b)$ twice in the same direction, while there is no surprise visit. 

To address this issue, we note that unless there is a surprise visit that can again be charged to this, we can ask each walker on the matrix boundary a single question "whether you are being chased?"
Note that for a single walker, this is a question that can be answered in $[2]$ as it is unnecessary to specify which particular walker is chasing. As we additionally maintain a branch-chasing graph that is a tree, we show in ~\pref{sec:branch-chasing} that a single constant cost per active walker on the boundary (i.e. per vertex on the matrix boundary), is sufficient to determine which edge is being used to chase. As a result, $R$-chase remains to be fixed up to factor of an absolute constant per vertex in each block. 

\begin{figure}[h!]
    \centering
    \includegraphics[width=320pt]{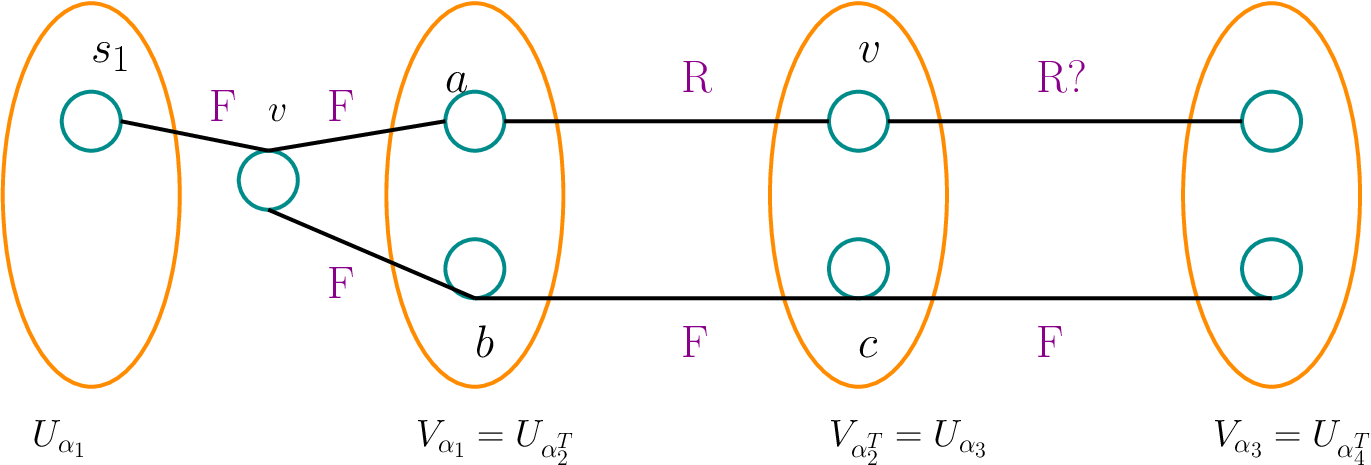}
     \caption{ Chase confusion }
    \label{fig:chase-confusion}
\end{figure}

   \section{Tight matrix norm bounds: vertex encoding cost}
\label{sec: norm-bounds-vertex-encoding}
 
%We will bound the trace on each length-2 block by, \[ 
%\max_{\calL:\text{ labeling of the edges}, |F|\geq |R|} \val(\calL)\cdot \cnt(\calL)
%\]
As illustrated in our technical overview, our starting point here is step(edge)-labeling,
\begin{definition}[Swapped Step label for vertex-encoding cost]
	 A labeling $\calL$ for a walk assigns each edge in the length-$q$ walk a label in $\{F,R,S, H\} $, with \begin{itemize}
		\item $F$ (a fresh step): an edge (or random variable) appearing for the first time, and the destination vertex is appearing for the \emph{first} time in the walk;
  	\item $R$ (a return step): an edge (or random variable) appearing for the second time;
	\item $S$ (a surprise step/visit): an edge (or random variable) appearing for the first time, and the destination vertex is \emph{not} appearing for the first time in the walk;
	\item $H$ (a high-mul step) : an edge (or random variable) appearing for neither the first nor the second time.
	\end{itemize}
	By default, we assume $\calL$ assigns the labels according to the traversal from $U_\al$ to $V_\al$; analogously, we denote $\calL^T$ if we traverse $\calL$ in a reverse direction from $V_\al$ to $U_\al$. 
\end{definition}
It should be noted here that as we work with vertex-encoding cost in this section, we consider $R$-step for an edge making its second-appearance.
\begin{remark}
	With some abuse of notation, we will also use $\calL$ (and analogously $\calL^T)$ to denote the labeling of a walk restricted to a a single block. 
\end{remark}
We now describe our encoding scheme given the edge-labeling. The encoding works by maintaining a set $W_t$ of the vertices in the current block revealed to us already, and walk along edges greedily to expand $W_t$ until all vertices in the walk are encoded, i.e., $W_t =V((\al\circ \al^T)^q)  $ . 
\subsection{Vertex encoding scheme}

% We first observe that for any shape $\al$, we can decompose the shape into $\al = \tilde{\al} \cup \branch(\al) \dang(\branch)$ where $\tilde{\al} $ is a shape obtained from $\al$ by removing any dangling branch, and $\branch(\al)$ is the collection of dangling branches in $\al$. With this decomposition, our encoding proceeds in the following two-step process \begin{enumerate}
% 	\item Step 1: Encode the constraint graph restricted to the non-dangling part of \emph{length $q$};
% 	\item Step 2:  Fix an order to traverse and attach dangling branches, then follow the fixed order to encode the dangling branches and attach them to the non-dangling constraint graph.
% \end{enumerate}
% In either step, we appeal to the following greedy vertex-encoding procedure for identifying vertices in the corresponding constraint-graph (i.e., assigning them labels in $[n]$).

For a fixed set $W_t$ at time-stamp $t$, and a given shape $\al$, we consider the edges across the cut produced by $W_t$, let $B_t\subseteq V((\al\circ \al^T)^q)\setminus W_t$ be the set of unknown vertices incident to some cut edge, and let $(i,j)$ be some edge across the cut with $i\in W_t, j\notin W_t$,

\begin{mdframed}[frametitle = {Vertex encoding procedure }]

\begin{enumerate}
	\item By default, we set $W_0=U_\al$: any vertex in $U_\al$ is the starting point revealed to us from the last block, hence no encoding needed;
	\item  \label{enu:return} if there's an $R$ edge $(i,j)$ across the cut, we expand $W$ along this edge o.w.;
	\item  \label{enu:long-high-mul-path} if there is a path of $H$ edges connected to $W$ across the cut of length at least $\log_d 2q|V(\al)|$, split the path into chunks of length at most $\kappa$ and encode each chunk's boundary in $2q|V(\al)|$ with an extra constant in $[2]$ indicating cycle orientation using $2$-cycle freeness;
	\item \label{enu:doubly-constrained} if there is some vertex $v\in B_t$ that has (at least) two paths of $H$ edges connected to $W$, the vertex is doubly-constrained (unless the path crosses $U_\al$ or $V_\al$), and all vertices along this path can be encoded using a label in $[2]$ to indicate the orientation of the cycle altogether by using $2$-cycle freeness; (specifically, if the path has a combined length of $\ell \gg \kappa$, we need to encode additional vertices in $|2qV(\al)|$ with a label in $[2]$ for every $\kappa$ vertices to identify the entire cycle, however, this is $(1+o(1))$ per vertex on average); o.w.,
	\item if there are multiple high-mul paths crossing $U_\al$ or $V_\al$, and it is a mirror copy, we expand  along these high-mul paths until the entire mirror copy is encoded into $W$; o.w.,
	\item if there is an $H$ edge $(i,j)$ across the cut, we expand $W$ along this edge;
	\item \label{enu:surprise-visit} if there is an $F$ edge $(i,j)$ across the cut for $j\notin U_\al\cup V_\al $, and $j$ is incident to some $H$ edge or $R$ edge outside the cut, or $j\in U_\al\cup V_\al$ while $j$ is incident to some $H$ edge outside the cut, we expand along this edge; o.w.,
	\item expand along the $F$ edges;
	\item we additionally require all vertices in the current length-$2$ block to be decoded before decoding vertices in the new block (except the vertices decoded by doubly-constrained vertices crossing the boundary or by mirror copy);
	\item (Excess $F$ edge clean-up) for each length-$2$ block, if there is any $F$ edge that does not get encoded as the cut edge, we assign an extra label in $|2qV(\al)|$ for each such edge indicating their second-time being used (if any). 
 \end{enumerate}
\end{mdframed}
\begin{remark}
	The above encoding scheme applies in a straightforward way for traversing walks in a reversed direction from $V_\al$ to $U_\al$.
\end{remark}
We would like to remind the reader that cases 2-5 allow the next vertex to be encoded in $O(1)$ factor, case 6 allows us to specify the next vertex using a factor of $d$ , case 7 allows us to specify the next vertex using a factor of $O(q_\al|V(\alpha)|))$, and case 8 requires the full factor of $n$ .

In the remaining of this section, we bound the vertex-encoding cost of the above encoding scheme.

\subsection{(Non-dangling) Second-use return is forced: overview} \label{sec:return-cost}
%For each vertex, consider the following quantity 
%\begin{proposition}
%	 Let $\calL$ be a maximal-value labeling, for any vertex $v\in V(\al)$ and any $H_v$ an $H$ component in $\calL$, \[ 
% \branch(H_v) \leq \dang(H_v) + 1
%\]
%\end{proposition}

In this subsection, we identify the cost of identifying a vertex arrived using some $R$ edge, i.e., an edge that appears for the second-time. At a high-level, the vertex-choice for the destination of such an edge is fixed: for example, when we are traversing a tree from some root, each vertex only has $1$ edge that goes towards the root. That said, things might get slightly more complicated when the walk is no longer tree-like. Intuitively, any confusion for determining the vertex reached by an $R$ edge can be settled with a cost of $O(d)$ using a degree bound, however, this would render us a trivial bound of $\sqrt{nd}$ analogous to a row-sum bound. 

To get around with this, we note that any such confusion can be charged with a gap in the previous walk, which offers us a gap to encode auxiliary data for settling such confusion. Concretely, we relate the potential confusion in determining destination of a return edge via \emph{surprise visits} \footnote{The argument in this section also relies upon another quantity \emph{missed-immediately-return} which can be seen as a proxy to surprise visits} , i.e., edges leading to visited vertices in the walk which correspond to some gap in our walk.

\begin{definition}[Potential-unforced-return factor $\pur$]
	At time $t$, for each vertex $v$, let $\pur_t(v)$ be defined as the following,\begin{enumerate}
		\item Each return leg closed from $v$ contributes $1$ to $\pur_t(v)$;
		\item Each unclosed $F$ edge currently incident to $v$ contributes $1$;
		\item An extra $-1$ as we can maintain a counter for each vertex such that there is always one potential return-leg presumed to be forced.	\end{enumerate}
		Let $H_{v,t}$ be the component reachable using subsequent $H$ edges from $v$ at time $t$, and define \[ 
		\pur_t(H_v) \coloneqq \sum_{i\in H_{v,t}} \pur_t(i)
		\] 
\end{definition}
\begin{remark}
	The definition of $\pur_t(v)$ does not require $v$ to be on the boundary at time $t$ (similarly for $H_{v,t}$), while we are only interested at the desired invariant when they are at the boundary.  
\end{remark}
%%\begin{definition}[Anchor vertex]
%	We call a vertex $v$ an anchor vertex if it opens up multiple $F$ edges at some single time-stamp in the walk.
%\end{definition}
\begin{definition}[Surprise visit counter] For a time-stamp $t$, for each vertex $v$ and $H_{v,t}$ defined above, let $s_t(v)$ be the number of surprise visits arriving at $v$ revealed by time $t$; we analogously define $s_t(H_v) = \sum_{i\in H_v} s_t(i)$.
\end{definition}
%\begin{definition}[Missed-immediately-return factor counter $\mir$]
%	For a vertex $v$ that opens up multiple $F$ edges from the boundary, let $\mir_t(v)$ be the number of $F$ edges it opens up as an anchor vertex so far in the walk that are not immediately closed when $v$ appears on the boundary again. Analogously, we define $\mir_t(H_{v,t})$ for the corresponding $H$-component.
%\end{definition}
%\begin{remark}
%	For adjacency matrix or line graph, $\mir_t(v) = 0$ throughout the walk as each vertex may open up at most $1$ $F$ edge when at boundary.\end{remark}
%
\subsection{Starting from the basics: a forced-return argument for 1-in-1-out}
\label{subsect:return-for-one-in-one-out }
\begin{proposition}\label{prop: pur-invariant}
	Throughout the walk where any vertex opens up at most $1$ edge at the boundary,  we maintain the invariant \[ 
	\pur_t(H_{v,t}) \leq 2\cdot s_t(H_{v,t}) 
	\]
	for $H_{v,t}$ the $H$-component reachable from the current boundary vertex $v$.
	
%	In general, for a walk with no dangling vertex while a vertex (or an $H$ component) may open up more than $1$ $F$ edge at the boundary, we maintain the invariant \[ 
%	\pur_t(H_{v,t}) \leq 2\cdot s_t(H_{v,t}) + \mir_t(H_{v,t})
%	\]
\end{proposition}
\paragraph{Overview: unforced-return ignoring $H$ edges }
For starters, we observe for a single vertex $v$, if there is no $H$ edge being used throughout the walk and each vertex opens up at most $1$ edge on the boundary, \begin{enumerate}
%	\item Since each vertex opens up at most $1$ $F$ edge on the boundary, $\mir$ is $0$ throughout;
	\item For the base case, when $v$ is first revealed, the return from $v$ must be using the edge that explores $v$, and we have $\pur(v)=0$;
	\item Each departure from $v$ may add $1$ potential-unforced-return edge as it adds an unclosed incidental $F$ edge to $v$, and departure from $v$ using $R$ edge does not change the $\pur$ factor at $v$;	\item However, the subsequent arrival may close an $F$ edge that is previously contributing as an unclosed $F$ edge incident to $v$ (though not necessarily the $F$ edge opened by the last departure), and hence $\pur(v)$ does not change from the last time when $v$ is at the boundary if arriving via an $R$ edge;
%\item In general, if it opens up $b_v>1$ many $F$ edges the last time on the boundary, notice each of the newly open legs that remains open while we arrive at vertex $v$ again (regardless of the arrival followed by $H$ or $R$ edge) is a missed immediate-return leg; 	%	\item  In general, if it opens up $b(v)$ $F$ edges that are not immediately closed in the last departure, and the subsequent arrival closes an open $F$ edge (to avoid surprise visit), we pick up a change of at most $b(v)-1$ for $\pur(v)$ (where this is an upper bound as multiple $F$ edges may get closed simultaneously in the arrival as well);

%
	\item The arrival, if using a new $F$ edge, is a surprise visit arriving at $v$ and contributes an $1$ to $s_t$ and an $1$ to $\pur$ on LHS;  the bound holds if we assign an increment of $2$ via the one single surprise visit.
	\end{enumerate}
	
	\paragraph{Incorporating $H$ edges}
	To extend our above argument for using $H$ edges, we give an argument based on $H$-components which may evolve over time. We first appeal to the following two claims building upon the prior observation that arriving at a vertex via an $R$ edge does not increase its $\pur$ factor compared to when it last appears on the boundary (for the case only a single edge may be opened up), \begin{claim}[Calling an $R$ edge does not change $\pur$ on vertices already in the current $H$-comp]
		Let $v$ be the current boundary at time $t$, let $H_{v,t}$ be the $H$-component of the boundary at time $t$, and an $R$ edge is called from $H_{v,t}$ (concretely from $v$) \[ 
		\sum_{i\in H_{v,t}} \pur_{t+1}(i) \leq \sum_{i\in H_{v,t}} \pur_t(i)
		\]  
		In other words, despite $H_{v,t}$ is not necessarily the same as $H_{v,t+1}$ since calling an $R$ edge may expand the current $H$ component, it does not increase $\pur$ factor restricted to the vertices in the current $H$ component.
	\end{claim}
	\begin{proof}
		To see the upper bound, observe that calling an $R$ edge from an $H$ component, restricted to the vertices already in the $H$ component prior to the closing of this $R$ edge, we simply have an edge contributing as an unclosed $F$ edge shifted to a closed return leg called from the component. The inequality follows when we close an $R$ edge between vertices already in the same $H$ component.
	\end{proof}
		\begin{claim}[Getting merged into an $H$ component via an $R$ edge does not change $\pur$ on the annexed vertices]
		Let $v$ be the boundary at time $t$, and let $t'$ be the subsequent time-mark such that $H_{v,t}$ is $H$-connected to the boundary (at $t'$) again and an $R$ edge is used at this step, i.e., $H_{v,t}\subseteq H_{u,t'}$, the $H_{v,t}$ component gets merged into a $H$ component via an $R$ edge at the current step for $u$ the boundary vertex at $t'$ followed by an $R$ edge \[ 
		\sum_{i\in H_{v,t}} \pur_{t'}(i) = \sum_{i\in H_{v,t}} \pur_t(i)
		\]
%		In general, if $H_{v,t}$ opens up $b(H_{v,t})> 1$ not-immediately-closed $F$ edges when it appears at the boundary last time, we have \[ 
%		\sum_{i\in H_{v,t}} \pur_{t'}(i) = \sum_{i\in H_{v,t}} \pur_t(i) + b(H_{v,t}) 
%		\]
	\end{claim}
	\begin{proof}
		Since we depart from $H_{v,t}$ at time $t$ via an $F$ edge, we have $\pur_{t_m}(H_{v,t}) = \pur_{t}(H_{v,t_m})+1$ for any $t<t_m<  t'$; however, since we arrive back at a vertex in $H_{v,t}$ at time $t'$ by closing an $R$ edge, we have $\sum_{j\in H_{v,t}} \pur_{t'}(j)  = \sum_{j\in H_{v,t}}\pur_{t'-1}(j)-1 = \sum_{j\in H_{v,t}}\pur_{t}(j)$, giving us the desired.
	\end{proof}
	
	We now proceed to prove the main proposition.
	\begin{proof}[Proof to  Proposition~\ref{prop: pur-invariant}]
	Suppose this is true throughout the walk so far, and we are currently at $H_{v,t}$, and move to $u$ from $v$ at time $t$, we need to show the above for $H_{u,t+1}$. By construction, it suffices for us to consider whether this is an $F/R$ edge.  Notice this is immediate if $v$ is a new vertex, as $H_{u,t+1} = \{u\}$, and there is no potential-unforced-return from $u$ at this time. Otherwise, if the $v\rightarrow u$ leg opens up a new surprise visit,
		\begin{enumerate}
		\item If $u\in H_{v,t}$, the new $F$ edge that gets open may contribute $1$ to both $\pur_{t+1}(v)$ and $\pur_{t+1}(u)$ while $\pur$ factor of other vertices in $H_{v,t}$ does not get changed. That said, we pick up a gain of $2$ on the LHS, while the surprise visit also now contributes to $s_{t+1}(H_{v,t+1})$, giving an increment of $2$ on the RHS as well;	
		\item If $u\notin H_{v,t}$, let $t_u$ be the last time vertex $u$ appears at a boundary $H$ component, and call this component $H_{u,t_u}$, in addition to the factors contributing to $\pur_{t_u}(H_{u,t})$, the departure at time $t_u$  opens up $1$ $F$ edge from $H_{u,t_u}$. Plus the new surprise visit that we just open up, we have \begin{align*}
		\pur_{t+1}(H_{u,t+1}) &= \pur_{t_u}(H_{u,t_u}) +  2 \\&
		\leq 2\cdot s_{t_u}(H_{u,t_u}) +  2\\&=	 2\cdot s_{t+1}(H_{u,t}) 	
		\end{align*}
		where we observe that $s_{t+1}(H_{u,t}) = s_{t_u}(H_{u,t_u}) + 1$ as we just pick up a new surprise visit, and $H_{u,t_u} = H_{u,t}$ as a set of vertices.
		
%
%		and our arrival at $H_{u,t}$ via the surprise visit at time $t$ again contributes a gain of $2$ to the $RHS$ i.e., $\pur_{t+1}(H_{u,t}) = \pur_{t_u}(H_{u,t})+ b(H_{u,t_u})$, while we pick up a change of $2$ on the RHS as above.  
\end{enumerate}
		If $v\rightarrow u$ is an $R$ leg that closes some $F$ edge, let $t_u$ be defined as earlier,
		 \begin{enumerate}
		 	\item  $u\notin H_{v,t}$: This yields a merge of $H$ components, and we can partition $H_{u,t+1}$ into vertices $H$-reachable either via $u$ or $v$ at time $t$,  \begin{align*}
				\pur_{t+1}(H_{u,t+1})& = \sum_{i\in H_{v,t}} \pur_{t+1}(i) + \sum_{j\in H_{u,t}} \pur_{t+1}(j)
			\end{align*}
		By our claim that calling $R$ edge does not increase $\pur$ on the vertices in the current $H$ component, \begin{align*}
			 \sum_{i\in H_{v,t}} \pur_{t+1}(i)  \leq  \sum_{i\in H_{v,t}}\pur_t(i) =\pur_t(H_{v,t})\leq  2\cdot s_t(H_{v,t}) 
		\end{align*}
		Since $u\notin H_{v,t}$, for component $H_{u,t} = H_{u,t_u}$, it gets annexed into a new $H$-comp via an $F$ departure and a $R$ return, \[
		\sum_{j\in H_{u,t}}\pur_{t+1}(j) = \sum_{j\in H_{u,t}} \pur_{t_u}(j) = \pur_{t_{u}}(H_{u,t_u}) \leq 2\cdot s_{t_u}(H_{u,t_u}) 		 \] 
%		 If multiple $F$ edges get open up, among which $b(H_{v,t})$ are not immediately-return, we have \begin{align*}
%		 	\sum_{i\in H_{v,t}} \pur_{t'}(i) = \sum_{i\in H_{v,t}} \pur_t(i) + b(H_{v,t}) &\leq 2\cdot s_{t_u}(H_{u,t_u}) + \mir_{t_u}(H_{u,t_u})  +b(H_{v,t}) \\&= 2\cdot s_{t_u}(H_{u,t_u}) + \mir_{t+1}(H_{u,t_u}) 
%		 \end{align*}
%		 where we observe any edge in $b(H_{v,t})$ is a newly gained not-immediately-return leg. 
Combining the above yields 
\begin{align*}
		 	\pur_{t+1}(H_{u,t+1}) &\leq 2\cdot s_{t+1}(H_{v,t}) + 2\cdot s_{t+1}(H_{u,t_u})  = 2\cdot s_{t+1}(H_{u,t+1})  
		 \end{align*}
		 as $H_{v,t}$ and $H_{u,t_u}$ form a partition for $H_{u,t+1}$ for $u\notin H_{v,t}$, and the surprise visits inherit over time.
		 \item $u\in H_{v,t}$: This is a return within an $H$-component, and we have \[
		 \pur_{t+1}(H_{u,t}) \leq  \pur_{t}(H_{v,t})
		  \]
		  since calling a return leg within the same $H$ component does not increase $\pur$ factor (in fact, a decrease by $1$ as the corresponding $F$ edge contributes a factor $1$ to both endpoints in the $H$ component, while after being closed, it only contributes as a closed $R$ edge to one endpoint).
		  \end{enumerate}
	\end{proof}

\subsection{Branching with mirror copy: a heuristic argument}\label{sect: return-mir}
\paragraph{Branching} The above argument crucially relies upon that each vertex can only encode at most $1$ edge when it is on the boundary each time, however, this is not true in the general setting for graph matrix, and particularly why rough matrix norm bounds depend crucially on \emph{minimum vertex separator} of the underlying graph. That said, we note that the argument can be extended in a lightweight manner to the general case when the matrix is indexed by order-tuple, i.e., there is no permutation-jump involved in the walk.

For concreteness, we call a vertex an \emph{anchor} vertex if it encodes multiple (out-going) edges on the boundary at a single time in the walk. The crux of the argument observes that we can further constrain each such vertex to have each of its $F$ edge opens up in the branching to be closed before the vertex appears again on the boundary, as otherwise there is a gap from surprise visit that we can identify locally in the constraint graph, and this is captured by the counter $\mir$ to be defined momentarily.

\paragraph{Return bound for tuple-indexed walk}
Let's start with some definitions in this setting, and fix the traversal direction from $U_\al$ to $V_\al$.
\begin{definition}[Anchor vertex]
For a vertex $v$ in the shape $\al$ (or $\al^T$) and an edge-labeling $\calL$, we call it an \emph{anchor} vertex if it pushes out more than $1$ non-dangling edge, i.e., an edge on a simple path from $v$ to the next boundary.  \end{definition}
 \begin{definition}[Mirror copy]
 	For any vertex $v$ in $\al$ (or in $\al^T$), we let $\sigma(v) \in V(\al^T)$ or (in $V(\al)$) be the counterpart of $v$ obtained by shape-transpose. 
 \end{definition}
%For starters, we observe that the linearized constraint graph can decomposed as having $U_\al =V_\al$ vertex disjoint paths from $U_1$ to $V_q$ attaching path-segments between anchor vertex and its mirror copy. 

\begin{definition}[Missed-immediate-return factor counter $\mir$]
	For a vertex $v$ that opens up multiple $F$ edges from the boundary, let $\mir_t(v)$ be the number of $F$ edges it opens up as an anchor vertex so far in the walk that are not immediately closed when $v$ appears on the boundary again. Analogously, we define $\mir_t(H_{v,t})$ for the corresponding $H$-component.
\end{definition}
\begin{remark}
	For adjacency matrix or line graph, $\mir_t(v) = 0$ throughout the walk as each vertex may open up at most $1$ $F$ edge when at boundary.\end{remark}

% For starters, consider the linearized constraint graph obtained by forgoing the trace-walk condition $U_1 = V_q$ (as an equality of order-tuples), and this is a graph of $q$ blocks of $\al$ and $\al^T$. We observe that the constraint graph can be encoded via the following process, and use this process to identify the charging of potential-surprised-visit with $\mir$ factor. For the following discussion, fix the traversal direction from $U_1$ to $V_q$.

% \begin{mdframed}[frametitle = {Decomposing (linearized) constraint graph into path-segments}]
% \begin{enumerate}
% 	\item We start with at most $U_\al=V_\al$ \emph{directed} vertex-disjoint paths from $U_1$ to $V_q$;
% 	\item Attach  \emph{directed} path-segments from an anchor vertex to its mirror copy in the subsequent block; 
% 	\item Attach path-segments between anchor vertices in the same block;
% 	\item Attach \emph{directed}(dangling) path-segment from an anchor vertex to a dangling vertex in the same block.
% \end{enumerate}
% \end{mdframed}

\begin{proposition}
	For each block, process the anchor-vertices pair according to their distance in increasing order in $\al\circ\al^T$. For each starting anchor vertex being processed, each path-segment crossing the boundary starting with $F$ edges is either immediately backtracking, or we can assign a surprise visit from the path segment (or its adjacent one) to the starting anchor vertex for its corresponding $\mir$ factor. Moreover, each surprise visit is assigned to at most $2$ $\mir$ factors.
\end{proposition}
\begin{proof}
		 We first observe that for any path-segment that departs from the starting anchor vertex with $F$ edge may only return to the anchor vertex (with label in $[n]$) within the path-segment at the corresponding mirror copy location (aka. the ending anchor vertex) if the starting anchor vertex (labeled in $[n]$) ever appears again in the path segment. This follows from the injectivity assumption: the starting anchor vertex may only appear in the current block, and unless there is a surprise visit in the segment, the path returning to the starting anchor vertex lands the starting anchor vertex (its label in $[n]$) at the ending anchor vertex location in the constraint graph. That said, for any anchor vertex that gets returned too early, aka. prior to the corresponding mirror copy location, the branch that returns to the current vertex gets assigned a surprise visit and each other path-segment that opens up new $F$ branches departing from the anchor vertex also gets assigned a surprise visit.		  
		 
		 	 Next, we observe that the return to the starting anchor vertex at its mirror location is mandatory for any segment starting with $F$. Suppose the current path-segment does not return, the starting anchor-vertex picks up $1$ $\mir$ factor from the path-segment. Suppose there is no surprise visit in the current segment (i.e., the $\mir$ factor from the segment is uncharged yet), we claim that this must be the first processed path-segment between the current pair of starting and ending vertices, as otherwise the current path-segment has to arrive at some known vertex at the mirror-copy location which gives a surprise visit to assign to the corresponding $\mir$ factor. However, notice for any anchor vertex with path-segment added to its mirror copy in its subsequent block, there are at least two edge-disjoint paths (given by distinct path segments) to its mirror copy. For concreteness, call the second path-segment processed the adjacent of the first segment. Notice these two path segments are disconnected in the graph sample once the starting anchor vertex is removed, while these two paths need to merge eventually at the mirror copy location. By assumption, the mirror copy is not the starting anchor vertex, hence there must be a surprise visit in the second path-segment and we can charge it to the $\mir$ factor of both the first and second path-segments. 
	 
	Furthermore, to see that a factor of $2$ is sufficient, notice each subsequent path-segment between the current pair of starting and ending anchor vertex either comes with a surprise visit or it is an immediate backtracking segment itself that does not contribute $\mir$ factor to the starting anchor vertex. This completes the proof of our proposition. 
	\end{proof}

%Branch-chasing
%
%
%
%Below is the branch-chasing argument
\subsection{Branch-chasing using R is free}\label{sec:branch-chasing}
\paragraph{Troubles for yje immediate-return based argument}
To recap our prior argument, we essentially give a special treatment for \emph{branching} vertices, i.e., those in the constraint graph that may encode multiple edges across the boundary, and we claim that for any such vertex, each $F$ edge that gets opened up when departing from this vertex must be immediately closed before this vertex appears again in the boundary: in particular, the branching vertex needs to match in its mirror-copy location unless there is a surprise visit that offers a gap for charging potential unforced return due to the vertex opening up multiple $F$ edges simultaneously in a single-shot. 

However, this heuristic falls apart when we apply it for grouped shapes, as we no longer have injectivity assumption across the boundary: in particular, mirror copy is no longer well-defined. That said, as highlighted in the overview section, we can consider a branching vertex that simultaneously opens up multiple $F$ branches, and one of them returns sooner than others, causing a potential confusion when calling an $R$ edge at the anchor vertex, that it may refer to any of the $F$ edges leading to the alive $F$ branch (and the $F$ edge leading to the anchor vertex before any of the $F$ branches get opened up).

\paragraph{Branch-chasing graph}
We start by considering an auxiliary graph constructed as the following. It should be noted that \emph{boundary} in the following section typically refers to the boundary of the walk unless otherwise specified. Notice that each vertex on the boundary gives us a tiny constant decay via their edges. However, the core of our argument relies upon coloring the edges such that red edges are the only possible source of confusions, and furthermore, we maintain the invariants that each vertex is not incident to \emph{too many red edges} and each black directed edge leads to a path to some walker on the current boundary.
 
\begin{mdframed}[frametitle = {Constructing and maintaining the branch-chasing graph}]
\begin{enumerate}
	\item This is a directed graph with each vertex receiving a label in $[n]$, moreover, we can color the edges into \emph{red} and \emph{black}, such that we are interested in the desired invariant each vertex is incident to only one out-going red edge unless there are further surprised visits assigned to this vertex;
	\item Observe that there is no confusion if there the incident red-edge is \emph{unique};
% 	\item In this graph, we are primarily interested in unclosed $F$ edges and vertices that split out multiple $F$ edges at a single time in the boundary;
	\item We start with an empty graph and build the graph as the walk proceeds;
	\item For each $F$ edge, we open up a \emph{directed black} edge from the original vertex to the vertex it arrives at;
	\item If the $F$ edge that gets opened up is a surprise visit, i.e., it leads to a previously explored vertex, we flip all the \emph{black} unclosed $F$ edges that have a directed path in the branch-chasing graph to the surprise destination to \emph{red}. 
\end{enumerate}
\end{mdframed}
With this graph in hand, we first see how it may help us in deciding $R$ edges that arise in the chasing process when there are multiple walkers.
\begin{observation}
	 	If there is no dangling branch, any black edge is on a directed path to the current boundary of the walk, i.e., the boundary of the branch-chasing graph is equivalent to the walk boundary.
\end{observation}
\begin{corollary}
	Each chase using a black edge is fixed.
\end{corollary}
\begin{proof}
	Even though each vertex may be incident to multiple black edges, observe that the subgraph cut across by any black edge is a rooted tree, in particular, with some \emph{active} walker in the current boundary. We then consider the following process, at each (block) step in the walk,  it suffices for us to go through the vertices on the current boundary and query
	\begin{enumerate}
		\item whether it is a \emph{branching} vertex looking for $R$ chase;
		\item whether it is a vertex assigned to \emph{beacon guidance} for $R$ chase;
	\end{enumerate}
	For each step, for each walker at some branching vertex looking for $R$ chase, the encoder can assign the  boundary vertex on the path promised by the out-going $R$ edge as the following,
	\begin{enumerate}
		\item For vertex $v$ looking for beacon-guidance, consider the subtree reached by the desired outgoing edge, assign beacon guidance for any branching vertex with $R$ chase query contained in the subtree;
		\item Assign an arbitrary vertex in its subtree reached by the desired out-going edge not assigned for beacon-guidance to $v$. 	\end{enumerate}
	
	 To decode, since the branch-chasing graph (restricted to black edges) is a directed tree (forest), there exists some query node whose (out-going) subtree is query-node free; by out encoding, there is also promised to be at least one guidance node in that subtree, assign the outgoing edge pointing to the (potentially multiple) guidance beacon to the query-node. Remove their colors (temporarily), and repeat this process.
\end{proof}
\begin{remark}
	It should be pointed out that we only aim to identify the edge-set used by $R$ chase from this process, as any further confusion is then resolved using edge-decay and automorphism factor.
\end{remark}

We now proceed to handle the potential confusion caused by surprise visit in the branch-chasing graph, i.e., red edges. Throughout the process, a surprise visit may turn black edges into red edges, and call this path the corresponding flip path of the surprise visit. 
\begin{proposition}
Consider the path of edges that have their colors flipped, this \emph{flip path} passes through at most $2$ vertices that are incident to \emph{out-going} red edge before the current flip.
\end{proposition}
\begin{proof}
	Without loss of generality, suppose the surprise visit proceeds from $a\rightarrow b$, we observe that any edge that gets flipped by this surprise visit. Let $u_a$ be the closest vertex on the flipped path to $a$ that is already incident to some out-going red edge before the current flip, and suppose there is some other vertex $v_a$ that also has a black path to $a$ before the current flip. We observe that it needs to pass through $u_a$ as the subgraph restricted to black edges is a directed tree. However, consider the supposedly black path from $v_a$ to $u_a$, we claim that it must already be red before the current flip as $u_a$ has an out-going red-edge, which by construction, requires the $v_a$ to $u_a$ path to be flipped to red. Therefore, there can be at most $1$ vertex incident to an \emph{out-going} red edge before the current flip that has a flip-path to $a$; applying the above argument to the flip-path leading to $b$ then gives us the desired bound.
\end{proof}
\begin{corollary}
	Assigning each surprise visit to two of the vertices its flipped path passes through that may be incident to more than $1$ out-going red edge captured by the factor $\red$,  we can maintain the following invariant for each vertex \[ 
	\text{out-going red edges incident to $v$} \leq 1 + w_t(v)
	\]
	where $w_t(v)$ accounts the number of surprise visits assigned to vertex $v$ for creating new excess red edge originating from $v$.
\end{corollary}

\subsection{Return is still ''forced'' with dangling branches}
% We now put the above arguments together, concretely, we combine the argument for $1$-in-$1$-out and $R$-chase being free argument for walks with permutation-jump. We start by taking into account our prior observation of $R$-chase being free into $\pur$ factor, the potential function that controls the number of potential-unforced-returns throughout the walk. 
Before we dig into our final bound, we first see how having dangling edges may affect our forced return bound. Notice that having dangling edges indeed brings $\pur$ factor that does not get assigned to any surprise visit, however, our goal in this section is to show that beyond the ''anticpated'' $\pur$ factors, dangling edges do not contribute further confusion.

\paragraph{Attaching dangling branches and dangling factor assignment}
Recall that the return for each dangling branch is not necessarily fixed, and our focus here is to see each dangling branch, despite how many vertices it may pass through, only incur $1$ factor of potential-unforced-return.

 Let's start by recalling dangling branches (in a shape) can be encoded in the following manner: we can recursively attach a dangling branch to some identified vertex in the shape. And this begets our dangling factor assignment scheme:in each block-step, we assign each dangling factor to the vertex it attaches to, and update the $\pur$ function for each vertex on the newly attached segment via the $1$-in-$1$-out argument. At this point, it should be clear that the $\pur$ factor is associated with the number of dangling branches as opposed to dangling vertices. To formalize this intuition, we introduce attachment edge, as this is the edge part of a dangling branch that comes with extra $\pur$ factor in our argument. Notice not all the edges in a dangling branch come with extra $\pur$ factor, and this is crucial for us in obtaining a norm bound with extra factor only depending on the number of dangling branches as opposed to the number of dangling vertices.
 
\begin{definition}[Attachment edge]
    Given a shape and a block-step walk of the corresponding shape, for each dangling branch, call the edge an attachment edge of the corresponding branch if its origin vertex is identified using $R/H$ edge while it spits out some $F$ edge part of the dangling branch.
\end{definition}
 \begin{observation}
          For any dangling branch, at any block-step of the walk, there is at most one attachment edge unless there is a surprise visit locally along the branch. 
 \end{observation}
 We now upgrade our edge-coloring scheme in the branch-chasing graph to take into account dangling branch. For each dangling branch that gets encoded, we will color its $F$ edges as the following scheme,
 \begin{mdframed}[frametitle = {Coloring scheme for dangling branches}]
 \begin{enumerate}
     \item Following the encoding order that we start from some non-dangling vertex or landing vertex if dangling, and encode the edges along the dangling branch;
     \item We restrict our attention to $F$ edges along the dangling branch, and
     \item Color each attachment edge \emph{yellow};
     \item Color each non-attachment edge \emph{green}.
 \end{enumerate}
\end{mdframed} 
\begin{remark}
   An out-going green edge is an edge opened as part of dangling branched while it does not get assigned dangling factor.
\end{remark}
\begin{claim}
For any vertex $v$, let $s_{\dang,t}(v)$ be the number of surprise visits arriving at $v$ by time $t$ that is part of any dangling branch, we have \[ 
\text{out-going green edges from $v$ by time $t$} \leq 1+ \green_t(v) \]
and furthermore, \[ 
\green_t(v) \leq s_{\dang,t}(v)
\]
\end{claim}

\subsection{Ultimate unforced return bound}
We now put the above arguments together, concretely, we combine the argument for $1$-in-$1$-out and $R$-chase being free argument for walks with permutation-jump. Towards that end, we need to modify the definition of $\pur$ factor to take into account our prior observation that return is still forced if it is an $R$ chase,
\begin{definition}[Time $t$] We can associate a particular time-mark to each edge . That said, we use time $t$ to represent the moment before the edge at the corresponding time-mark is used.	
\end{definition}
\begin{definition}[Upgraded $\pur$ factor] At time $t$ (before the edge of time $t$ is used), for each vertex $v$, let $\pur_t(v)$ be defined as the following,\begin{enumerate}
		\item Each return leg closed from $v$ contributes $1$ to $\pur_t(v)$;
		\item Each unclosed \emph{red} or \emph{green} $F$ edge currently incident to $v$ edge \emph{in the branch-chasing graph} contributes $1$;
		\item A $-1$ as we can maintain a counter for each vertex such that there is always one potential return-leg presumed to be forced;
		\item An extra $-1 \cdot 1_{v \text{ has seen outgoing red edge by } t} $ as we can maintain a counter for each vertex such that there is always one outgoing red edge presumed to be forced, and this is only needed when $v$ has outgoing red edge; 
		\item An extra $-1 \cdot 1_{v \text{ has seen outgoing green edge by } t}$ as again we can maintain a counter for green edges such that there is always one out-going green edge presumed to be forced, and this is also only needed when $v$ has outgoing green edge.
		\end{enumerate}
		Let $H_{v,t}$ be the component reachable using subsequent $H$ edges from $v$ at time $t$, and define \[ 
		\pur_t(H_v) \coloneqq \sum_{i\in H_{v,t}} \pur_t(i)
		\] 	
\end{definition}
\begin{remark}
	We highlight that the only modification is at that each $F$ edge can now be potentially unforced only if it is a red edge in the branch-chasing graph at the moment.
\end{remark}
\begin{definition}
	Call an edge protected if it appears as a \emph{black} edge in the branch-chasing graph or it is attached as an attachment edge for some dangling branch.
\end{definition}

\begin{lemma}[Ultimate forced-return bound] Throughout the walk, we maintain the invariant \[
\pur_t(H_{v,t}) \leq  s_t(H_{v,t}) + w_t(H_{v,t}) +\dang_t(H_{v,t}) +\green_t(H_{v,t})  \]
 for $H_{v,t}$ the $H$-component reachable from the boundary vertex $v$ at the end of time $t$ (after the edge at time $t$ is used),  where \begin{enumerate}
 	\item 
  $s_t$ is the surprise visit counter that counts the number of surprise visits arriving at $H_{v,t}$;
  \item 
  $w_t$ is the counter that assigns surprise visits to vertices with multiple red edges in the branch-chasing graph in $H_{v,t}$;
  \item $\dang_t$ is the counter that counts the number of dangling branches attached to $H_{v,t}$, i.e., yellow edges;
%  \item $r_{\dang_t}$ is the counter that counts the number of times we arrive at the component $H_{v,t}$ via an attachment edge of some dangling branch attached to $H_{v,t}$ (i.e., the reverse of some dangling attachment edge) ;
%  \item $b_t$ the counter that counts the number of surprise visits encoded from some vertex in $H_{v,t}$ when it has out-going edge in the branch-chase graph; 
%  \item $r_s$ is the counter that counts the number of surprise visits incident to $H_{v,t}$ (only needed when a surprise visit is used the second time);
 \end{enumerate}
\end{lemma}
\begin{proof}
For clarity, in the case a vertex splits out multiple walkers at the block-step boundary at time $t$, we can process them in an arbitrary order, and notice each $F$ edge splits out does not affect the desired invariant at time $t$ since each such edge that appears in the branch-chasing graph with a walker on the boundary via local injectivity of graph matrix.

% Suppose this is true for vertex $v$ when it last appears in an $H$-connected component on the boundary at time $t$, and we arrive at $v$ again at time $t'$ from $u$ using a leg $u\rightarrow v$ , we are to show the invariant at the end of current time-mark $t'$. We consider how $v$'s incidental $F$ edges may change from $t$ to the end of $t'$. Notice there may be multiple edges in the current step lead to $v$ at time $t'$, however, we can process the walkers in order and therefore assume in each time-mark only one edge is being used.
 
 Suppose the current edge being used goes from $u$ to $v$, let time $t$ be the last time mark $v$ appears in the $H$ component at the boundary (before the edge at time $t$ gets used), and call the component $H_{v,t}$. We start by partitioning the unclosed $F$ edges incident to $H_{v,t}$ at time $t$ according to whether they contribute to $\pur_t(v)$, i.e. whether they appear as a directed edge going out from some vertex in $ H_{v,t}$ in the branch-chasing graph at time $t$. Let those that contribute (i.e., unprotected) be in $O_{t}(v)$ and the others (i.e., protected edges) be in $B_{t}(v)$.  Notice the last departure from $H_{v,t}$ (not necessarily from vertex $v$) at time $t$ may create one edge that is either in $B$ or in $O$ at time $t+1$. Our argument relies upon keeping track of how $B$ and $O$ may change between these two time-marks, and we now case on the status of the leg $u\rightarrow v$ at the current time-mark $t'$ (after the edge at $t'$ is used),

\begin{enumerate}
	\item Again, by construction, it suffices for us to consider the edge used at time $t'$ be $u\rightarrow v$ as either $F$ or $R$, since $H$ edge is automatically captured when we consider $H$-components;
	\item If the walker arrives via an $R$ leg, and assume $v,u$ not in the same $H$ component before this leg as the other case is immediate via prior argument that calling an $R$ edge within the same component cannot increase $\pur$ factor. Therefore, it suffices for us to assume $u$ and $v$ being in disjoint $H$ components at time $t'-1$, i.e., before the use of the leg $u\rightarrow v$, and we can further restrict our attention to $H_{v,t'-1}$ as \[ 
	\pur_{t'+1}\left(H_{v,t'}\right) = \pur_{t'+1}(H_{u,t'}) + \pur_{t'+1}(H_{v,t'}) \leq \pur_{t'}(H_{u,t'}) + \pur_{t'+1}(H_{v,t'})
	\]
	since calling an $R$ edge from an $H$-component cannot increase the $\pur$ factor on the original vertices, i.e. \[ 
	\pur_{t'+1}(H_{u,t'})\leq \pur_{t'}(H_{u,t'}) 
	\]
	\item  Let $E_t$ be the edge that gets opened up at time $t$ when departing from $H_{v,t}$;
	\item Consider the edges that may now contribute to $\pur_{t'}(H_{t,v})$ while not accounted for at time $t$, let this set of edges be $\Delta_{t'-t}(H_{v,t})$: they are either the edge $E_t$ if it is opened as a dangling edge, or an edge that shifts from $B_{t}$ to $O_{t'}$ (as well as potentially non-dangling $E_t$ that shifts from black to red), i.e., they are protected on the branch-chase graph at some time between $t$ and $t'-1$ as \emph{black} edges,  while no longer so at $t'$ since they are now \emph{red} edges; 	
	\item For the edges that shift from black to red within the time interval $[t,t']$, we process them in an arbitrary order of their out-vertex in $a\in H_{v,t}$, and for each vertex's first edge being processed, we additionally case on whether the out-vertex $a$ has out-going red edge at time $t$: if not, the first edge flips $1_{a \text{ ever has outgoing red edge}}$ and this offsets the gain to the LHS; otherwise, any such edge assigns $a$ a factor of surprise visit via a gain in $w_{t'}(a)$ over $w_{t}(a)$, yielding  \begin{align*}
	& |\Delta_{t, t}(a) | \leq -1 \cdot \left( 1_{a \text{ has seen out-going red edge by } t' } - 1_{a \text{ has seen out-going red edge by } t } \right) + \left(w_{t'}(a) - w_{t}(a)\right)	
%	&= -1 \cdot \left( 1_{a \text{ has seen out-going red edge by } t' } - 1_{a \text{ has seen out-going red edge by } t } \right) + \left(w_{t'}(a) - w_{t}(a)\right)	
	\end{align*}
	\item If $E_t$ was opened up as dangling edge, and moreover, as a dangling attachment edge, we pick up a gain in $\dang_{t'}(H_{v,t})$ over $\dang_t(H_{v,t})$ that can offset the gain from edge. Otherwise, $E_t$ is a  a green edge, i.e., part of some dangling branch while not an attachment edge, $E_t$ is then captured by the factor of $\green_t(v)$: it is either forced as the sole green edge, or it gets assigned surprise visit factors via the function $\green_t(v)$, letting $E_t$'s out-going vertex be $a$, we have \begin{align*}
&1_{E_t \text{ is opened as a dangling edge} } \\& \leq \left(\underbrace{\dang_{t'}(a) - \dang_{t}(a)}_{\text{dangling attachment} } \right) + \bigg(-1\cdot 1_{E_t \text{ is the first out-going green edge from } a} + \left(\green_{t'}(a)- \green_t(a)\right) \bigg) 		
	\end{align*} 
%		\item  In this case, we conider whether  $E_t$ is a black edge at time $t'$: it either remains protected or it gets assigned a factor via a gain of surprise visit in $w_{t'}(H_{v,t'-1})$ over $w_t(H_{v,t})$ which offsets the gain of $1$ on the LHS from $E_t$;
%	\item Notice the arguments are identical to the $1$-in-$1$-out if the edge that gets closed is not protected in time $t$, and we have $|O_t(v)| = |O_{t'}(v)|$ though they may differ as sets;	
%\item Vertices in $H_{v,t}$ incident to green edges while not so anymore at time $t'$ may contribute a gain of factor $1$ in that gain of $\pur_{t'}(H_{v,t})$ over $\pur_t(H_{v,t})$, while noticing each such vertex gets assigned a drop of at least one green edge that ceases to contribute to  
\item From the above, for the $H$-component restricted to vertices $H$-connected to $v$ at time $t'$, suppose an $R$ edge is used, (and noticing $H_{v,t'} = H_{v,t}$)
	\begin{align*}
		&\pur_{t'+1}(H_{v,t'}) - \pur_{t}(H_{v,t'}) \leq  \pur_{t'}(H_{v,t'}) - \pur_{t}(H_{v,t'})\\&\leq |\Delta_{t'-t}(H_{v,t})| +1_{E_t \text{ is opened up as a dangling edge}} \\&\leq \sum_{a \in H_{v,t} } -1 \cdot \left( 1_{a \text{ has seen out-going red edge by } t' } - 1_{a \text{ has seen out-going red edge by } t } \right) + \left(w_{t'}(a) - w_{t+1}(a)\right) \\
		&+ \underbrace{\dang_{t'}(H_{v,t}) - \dang_{t}(H_{v,t})}_{\text{added dangling attachment for } E_t} + \bigg(-1\cdot 1_{E_t \text{ is the first out-going green edge from } a} + \left(\green_{t'}(a)- \green_t(a)\right) \bigg) 	  \\
		&\leq \sum_{a \in H_{v,t}} (w_{t'}(a) - w_t(a)) + (\dang_{t'}(a) - \dang_t(a)) + (\green_{t'}(a) - \green_t(a))
		\\&= \sum_{a \in H_{v,t}} \Delta_{t'-t}\circ w(a) + \Delta_{t'-t}\circ \dang(a) + \Delta_{t'-t} \circ \green(a)
				\end{align*}
		where we introduce the shorthand \[ 
		\Delta_{t_1, t_2}\circ f(v) \coloneqq f_{t_1}(a) -f_{t_2}(a)
		\]
		for $f \in \{s,w,\dang,\green\}$ and $t_1> t_2$.
%		Furthermore, recall we are considering the annexation of $H$-components (i.e. $v,u$ not $H$-connected at time $t'$) , combining with the factor of $H$-component restricted to $u$ at time $t'$ yields, \begin{align*}
%			&\pur_{t'+1}(H_{v,t'+1}) = \pur_{t'+1}(H_{u,t'}) + \pur_{t'+1}(H_{v,t'})\\ &\leq \pur_{t'+1}(H_{u,t'})+  
%		\end{align*} 
	\ Combining with the component connected to $u$ at time $t'$ gives, (assuming $u,v$ again not in the same $H$-component at time $t'$) \begin{align*}
		&\pur_{t'+1}(H_{v,t}) = \pur_{t'+1}(H_{u,t'}) + \pur_{t'+1}(H_{v,t'})\\
		&\leq \pur_{t'}(H_{u,t'}) + \pur_{t}(H_{v,t'})  + \left(\pur_{t'+1}(H_{v,t}) - \pur_{t}(H_{v,t})\right) \\
		&\leq s_{t'}(H_{u,t'}) + w_{t'}(H_{u,t'}) + \green_{t'}(H_{u,t'}) + s_{t}(H_{v,t}) + w_{t}(H_{u,t}) + \green_{t}(H_{u,t'}) \\
		&+ s_{t'+1}(H_{v,t'})-s_{t}(H_{v,t'}) +\sum_{a \in H_{v,t}} \left(\Delta_{(t'+1), t} \circ w(a) + \Delta_{(t'+1),t}\circ \dang (a)  + 	\Delta_{(t'+1), t}\circ \green(a) \right) \\
		&\leq s_{t'+1}(H_{v,t'}) +  w_{t'+1}(H_{v,t'+1}) +  \dang_{t'+1}(H_{v,t'+1}) +  \green_{t'+1}(H_{v,t'+1})
		 \end{align*}
	where we use the property that the counter functions $\{s,w,\green, \dang \}$ is non-decreasing over time and factorizes over connected components. 
	\item On the other hand, if the leg $u\rightarrow v$ is an $F$ edge, again assuming $u,v$ not being in the same $H$-component at time $t'$ and the other case is analogous, restricted to the $H_{v,t'}$, we have 
	\begin{align*}
		&\pur_{t'+1}(H_{v,t'}) - \pur_{t}(H_{v,t}) \leq |\Delta_{(t'+1) , t'}(H_{v,t})| + |\Delta_{t',t}(H_{v,t})|   + 1_{E_t \text{is opened up as a dangling edge}} \\
		&\leq |\Delta_{(t'+1) , t'}(H_{v,t})| + \sum_{a \in H_{v,t}}  \Delta_{t',t}\circ w(a) + \Delta_{t',t}\circ \dang(a) + \Delta_{t',t} \circ \green(a)
	\end{align*}
	where we notice the factors restricted to the last two terms are identical to the previous calculation when $R$ edge is used, and we can focus on the first term $|\Delta_{(t'+1) - t'}(H_{v,t})|$. Notice this set of edges are either the surprise visit we just opened up at time $t'$, or the edges that switch from black to red due to the current surprise visit (eg. potentially the edge $E_t$ that remains black at time $t'$ while now becoming red), and \begin{align*}
		|\Delta_{(t'+1) , t'}(H_{v,t})| &\leq s_{t'+1}(H_{v,t'})-s_{t'}(H_{v,t'}) + \sum_{a\in H_{v,t'}} |\Delta_{(t'+1),t'}(a)|
			\end{align*}
			Combining the above yields \begin{align*}
		\pur_{t'+1}(H_{v,t'}) - \pur_{t}(H_{v,t}) 		\leq &s_{t'+1}(H_{v,t'})-s_{t}(H_{v,t'}) \\ +&\sum_{a \in H_{v,t}}  \Delta_{(t'+1), t}\circ w(a) + \Delta_{(t'+1), t}\circ \dang(a) + \Delta_{(t'+1), t} \circ \green(a)
			\end{align*}
		Combining with the invariant for $H_{v,t}$ at time $t$ yields, \begin{align*}
		\pur_{t'+1}(H_{v,t'+1}) &= \pur_{t'+1}(H_{v,t'}) \leq \pur_{t}(H_{v,t})	+ s_{t'+1}(H_{v,t'})-s_{t}(H_{v,t'}) \\ +&\sum_{a \in H_{v,t}}  \Delta_{(t'+1), t}\circ w(a) + \Delta_{(t'+1), t}\circ \dang(a) + \Delta_{(t'+1), t} \circ \green(a)\\
		&\leq s_{t'+1}(H_{v,t'+1}) + w_{t'+1}(H_{v,t'+1}) +  \dang_{t'+1}(H_{v,t'+1}) +  \green_{t'+1}(H_{v,t'+1})
		\end{align*}

		%	\item if the walker arrives via an $F$ edge, any such edge is a surprise visit and contributes to both the gain of  of $\pur_{t'}(H_{v,t})$ over $\pur_t(H_{v,t})$ and the gain of $s_{t'}(H_{v,t})$ over $s_t(H_{v,t})$; moreover, if the arriving walker also departs $H_{v,t}$ using $F$ edges at time $t$, i.e. it breaks an edge on branch-chas
\end{enumerate}
\end{proof}

\paragraph{Putting things together for the final unforced return bound throughout the walk}
It suffices for us to consider each $H$-connected component when they last appear in the walk. Let $\{H_i,t_i\}$ be a sequence of $H$-connected components at the end of the walk, and time $t_i$ their last appearance time-mark, for a shape-walk $W$, we have \begin{align*}
	\pur(W) &\leq \sum_{H_i, t_i} \pur_{t_i}(H_{t_i})\\
	&\leq s_{t_i}(H_{i}) + w_{t_i}(H_i)+ \dang_{t_i}(H_{i}) +\green_{t_i}(H_i) \\
	&\leq 4\cdot s(W)+  \dang(W)
\end{align*} 
where we observe that each surprise visit may contribute a $1$ to the final bound via $s$ counter, and  a factor of $2$ via $w$ counter, a factor of $1$ via $\green$, and each $\dang$ attachment contributes a $1$ via $\dang$.
%\paragraph{Putting together unforced-return bound for set-indexed walk}

\subsection{(Local) 2-cycle-freeness can be useful} \label{sec:2-cycle-use}
\begin{proposition} \label{prop:long-high-mul-path}
	Given the starting vertex, a path of $H$ edges of length $t$ can be encoded at a cost of \[ 
	(2\cdot 2q|V(\al)|)^{\lceil \frac{t}{\kappa}\rceil }
	\]
\end{proposition}
\begin{proof}
	Since the graph is $2$-cycle-free at every $\kappa$ neighborhood, it suffices to split the long walk into $\lceil \frac{t}{\kappa} \rceil$ chunks, and encode the vertex at the end of each chunk; additionally, since there are at most $2$ paths between the start and end of each chunk, a label in $[2]$ to identify the orientation of the cycle is sufficient to recover the entire path.
\end{proof}
\begin{proposition}
	For any un-encoded vertex that has at least two paths of $H$ edges to some known vertices, the vertex is doubly-constrained, and the entire paths can be encoded at a cost of \[ 
	2\cdot (2\cdot 2q|V(\al)|)^{\lfloor \frac{t}{\kappa}\rfloor }
	\]
	\end{proposition}
\begin{proof}
	The proof is identical to the above, except we no longer need to identify an extra endpoint for the end of the last chunk to use local $2$-cycle-freeness.
\end{proof}
\begin{remark}
	Throughout our work, we will take $2q\cdot|V(\al)|\approx n^\epsilon$ and $\kappa\approx 0.3\log_d n$, which would give us \[
	(2q\cdot|V(\al)|)^{1/\kappa }\leq c
	 \]
	 for some constant $c>0$ independent of $d$ provided $\eps<O(\frac{1}{\log d}) $ where we take \[ 
	 q=\Theta(\sqrt{|U_\al||V_\al|} \log^2 n)
	 \]
\end{remark}

In addition to the main claim that gives an additional factor of $d$ for each dangling vertex, we observe $2$-cycle freeness can be exploited to improve that bound.

\begin{proposition}[Improved vertex factor for dangling $R/H$ vertices] \label{prop:improved-dangling}
Let $\al$ be a given shape, and let $S$ be a separator for some irreducible locally optimal labeling, we let $\dang(\al\setminus S)$ be the dangling factor for $S$,
	\[ 
	\dang(\al\setminus S) \leq \prod_{i\in \branch(\al\setminus S)} \min(2|V(\al)|q, d^{|\branch(i)|})
	\]
\end{proposition}
Observe that by possibly removing some edges incident to some dangling vertex if necessary, $\dang(\al\setminus S)$ can be partitioned into a collection of trees rooted at some non-dangling vertex in $\al$, and moreover, we call each vertex dangling in $\dang(\al\setminus S)$ a branch point if it is either a leaf, or has degree bigger than $1$. For each branching point $i\in \branch(\al\setminus S)$, let $\branch(i)$ denote the vertices on the path from last branching point (or the nearest non-dangling vertex) to $i$.  We are now ready to prove the proposition.

\begin{proof}
	Notice for each branch, a label in $2|V(\al)|q$ is sufficient using Proposition~\ref{prop:long-high-mul-path}; alternatively, a label in $d$ is sufficient for identifying each ''explored'' dangling vertex.
\end{proof}

%	 The above is almost sufficient to give us the desired proposition, except we note that the factors inside the proposition comes with an extra square root. To see this, we case on whether the vertex-factor is used to indicate a second-use path or a high-mul path. 
%
%
%	For the last appearance of an edge that lead to dangling vertices/branches, recall that we exchange the vertex factor for the second-use and the last appearance of an edge, the factor of $2|V(\al)|q$ or the factors of $d$ is only needed when the edge is being visited the second time (that it is not needed when it first appears),  hence taking the average with when the edge appears for the first time gives us the desired square root.
%	
%	For the case of middle appearances, notice each edge comes with a value of $\sqrt{\frac{n}{d}}^{\branch(i)}$, while each time a label in $2|V(\al)|q$ or a sequence of labels in $d$ is needed, giving us a total of \begin{align*}
%		\sqrt{\frac{n}{d}}^{\branch(i)}\min(2|V(\al)|q, d^{\branch(i)}) &\leq 
%	\sqrt{n}^{\branch(i)} \min(\frac{2|V(\al)q}{\sqrt{d}^{\branch(i)} }, \sqrt{d}^{\branch(i)}) \\&\leq \sqrt{n}^{\branch(i)}\min(\sqrt{2|V(\al)|q} ,\sqrt{d}^{\branch(i)})
%	\end{align*}
%	Recall that the factors of $\sqrt{n}^{\branch(i)}$ is already taken into account by the vertex factors of $\sqrt{n}$ for each vertex outside the separator, hence, each block's dangling blow-up is still at most $\min(\sqrt{2|V(\al)|q} ,\sqrt{d}^{\branch(i)})$ for each branch. Combining the argument for all dangling branches gives us the desired proposition.

\paragraph{Balance in $d$}
A priori, we need to show for each high-mul (non-separator) leg in the walk, only a $\sqrt{d}$ cost is sufficient for encoding, so that we can assign a cost of at most $\sqrt{n}$ for each step as \[
\sqrt{\frac{n}{d}}\cdot \sqrt{d} = \sqrt{n} \]
%\begin{proposition} [\textbf{Local $d$-balance}] For each mirror high-mul walk, each vertex in the path can receive a label in $[d]$ until it crosses the boundary, hence the vertices in the mirror walk can also be encoded in $\sqrt{d}$ per vertex on average; 
%\end{proposition}
%\begin{proof}
%	Since the vertices in $V(\al^T)$ are promised to receive the same labels in $V(\al)$ as their counterparts, it suffices to encode the vertices till $V_\al$, each of which requires a cost of $d$; hence, overall a cost of $\sqrt{d}$ is sufficient for each vertex in the path except $v^T$, which does not require any label given the start $v$.
%\end{proof}

\paragraph{Potential function for an individual walker} Let's start with a potential function argument for individual walker.

\begin{proposition}[\textbf{$\sqrt{d}$-balance}] Each $H$ edge requires at most a label in $O(\sqrt{d})$ for identifying its destination vertex.
\end{proposition}

\begin{proof}
We first give a brief argument for the case when there is a single walker in the graph matrix (eg. line graphs), and extend this argument to the case when multiple walkers are involved.
 \paragraph{Assigning a potential function for a single walker}
 First, we observe that it suffices for us to consider an interval of the walk of length $0.1\log_d n$ using $F/R/H$ edges, and let $s$ be the starting point of the interval and $t$ be the destination. Observe that $t$, as it is reached by $H$ edge at the end of the current interval and since there is only a single walker, there are two cases: either it is already revealed in the walk and can be specified with a label in $2|V_\al|q$, or it is explored via an $F/R$ segment in the subsequent interval. 

 In the first case, pick a pit-stop vertex $B$ to be the destination $t$. Otherwise, pick the pit-stop $B$ to be instead the last vertex on the $s-t$ path known to the walker by the beginning of the interval. That said, for any interval of length $0.1\log_d n$, the non-backtracking path from $s$ to $B$ is fixed using $2$-cycle freeness. In the case when $B\neq t$, note that the walker opens up the path from $B$ to $t$ in the subsequent interval via an $F/R$ segment before it uses these edges as an $H$-step. That said, it suffices for the walker to specify a direction towards or away from $t$ on the main-path when a new edge is opened up after $B$, which is again $O(1)$ per-step.
 
 The length $0.1\log_d n$ is picked so that we can specify the destination vertex for free once we have
 $\left(2|V(\al)|q_\al\right)^{1/0.1 \log_d n}<O(1)$ for $|V(\al)|\leq n^\delta$. That said, we can decompose the upcoming interval of $H$ edges as either walking on the main-path (non-backtracking segment from $s$ to $t$), and off-track path attached the main-path (backtracking segment). Each $H$ edge on the main-path does not require any encoding cost, and each $H$ edge on the backtracking segment contributes $d$ when it is heading outward of the main-path and $O(1)$ when heading inwards.

\paragraph{Assigning a potential function for a group of walkers}
In the trace-walk for general shapes, the walker is no longer unique as we may have arms splitting out from \emph{branching} vertices, and similarly, arms merging as well. In other words, this corresponds to a walker being spawned and dying in the walk. We now show the potential function argument for individual walker applies well for this general setting.

For starters, recall that constraint graph can be linearized and decomposed as the following, 
\begin{mdframed}[frametitle = {Decomposing (linearized) constraint graph into path-segments}]
\begin{enumerate}
	\item We start with at most $U_\al=V_\al$ \emph{directed} vertex-disjoint paths from $U_1$ to $V_q$;
	\item Attach  \emph{directed} path-segments from a \emph{branching} vertex to a vertex on a revealed path-segment, or a vertex in $V_q$.
	\item Attach path-segments between vertices in the same block;
	\item Attach \emph{directed}(dangling) path-segment from an anchor vertex to a dangling vertex in the same block.
\end{enumerate}
\end{mdframed}

Notice that it suffices for the encoder to know such assignment of potential function. For each walker, if it remains alive at the end of its subsequent $\kappa$-interval, i.e., it does not intercept other revealed path-segments, then the potential function can be assigned as an individual walker. 

On the one hand, for a walker that intercepts a revealed path-segment (i.e. it is a short-lived walker who does not remain in the walk for a $\kappa$-interval), it suffices for us partition the segment into $\kappa$-intervals again, and it suffices for us to restrict our attention to a short-lived walker, that it only takes a walk in an interval of length much less than $\kappa$, i.e., we do not have sufficient edge-decay to identify a label in $2|V(\al)|q_\al$. However, we note that the edge-set of non-backtracking $H$ edges can be identified in $O(1)$ via our wedge-bundling argument in the forthcoming ~\pref{lem: bundling_cost}, and this fixes the potential function for the short-lived walker before it \emph{dies}, i..e., joins another walker's segment.

On the other hand, for a walker that remains alive at the end of its subsequent $\kappa$-interval, i.e., it does not intercept other revealed path-segments, then the potential function can be assigned as an individual walker by extending the argument for the single walker. It is no longer true that the path from $s$ to $t$ must be revealed by this particular walker, however, this happens when another walker explores the main-path from the revealed pit-stop to $t$ and we observe that we may allow the pit-stop to "move" along the direction towards $t$ throughout this interval, and thereby specifying the edges on the main-path from $s$ to $t$. 

To move the pit-stop vertex for a specific walker, we apply the wedge-bundling argument in the forthcoming ~\pref{lem: bundling_cost} again. Concretely, apply this argument where each walker looking for a new pit-stop engages in the process as its current pit-stop vertex (as opposed to its current vertex location), and for the walker that explores the segment on the path from the pit-stop to $t$ for some other walker, it participates as the current vertex location. Now observe that it suffices for us to identify the short paths between these walkers, which can be done in $O(1)$ per walker/vertex at each step.  

\end{proof}

\section{Tight matrix norm bounds: block value and maximal-value labeling}
\label{sec: maxval-labeling}
%\begin{proposition}
%	Consider the length-$q$ walk as consist of $q$ blocks of $\al\circ\al^T$, given the edge labeling of each block $B_1,...,B_q$, notice the final trace is given by \begin{align*}
%		\Tr\bigg((M_{\al} M_{\al^T})^q\bigg) = \sum_{\text{shape-walk path} P}n^{|\comb(P)|}\cdot \E_G[\edgeval(P)\cdot \cycle \cdot \bdd  ]
%	\end{align*}
%	
%	
%	there are at most $\vtxcost_q$ many terms of expected value $\edgeval_q$, where 
%	\[ \vtxcost_q = \prod_{i=1}^q \vtxcost(B_i) \]
%	and \[ 
%	\edgeval_q= \prod_{i=1}^q \edgeval(B_i)
%	\]
%\end{proposition}
In this section, we show that the block value is determined by the vertex separator. On the one hand, this identifies the dominant term of the trace being governed by the SMVS; on the other hand, the structure lemma also provides us with a primitive to identify the gap from our trace-walk that offers slack for encoding auxiliary data.

 Towards this end, we start by constructing a vertex separator out of any locally-optimal labeling, and then relate the value of the labeling to the desired norm bound. However, before we proceed to construct a vertex separator, we need to understand the structure of a locally optimal labeling.
 \subsection{Set-up for block-value bounds}
 Before we dig into the core of our block-value bounds, we first make some (re-)definitions of our vertex-factors. In the previous section, we use $R$ edge to represent the second-time an edge appears, and crucially, we show that, a vertex, if pushing out an $R$ edge, the choice (unless dangling) is \emph{fixed}. On the other hand, it would be convenient for us to consider $R$ edge as the last appearance of an edge appearance due to potential assymmetries of the shape, and we note that for any edge, we can assume exchange its potential vertex factors for the second time and last time appearance, and thus assume each vertex, when pushing out an $R$ edge that corresponds to the edge's last appearance, is fixed. 
 
 In the following section, we will now think of $F/H/R$ corresponding to the first/middle/last appearance of an edge, which also motivates the following assignment scheme for edge-value depending on edge-labels, whose formal expansion is deferred to \cref{sec:edge-val-assignment},
	\paragraph{Edge-value assignment scheme}
	\begin{mdframed}
For each random variable $x$,
	\begin{enumerate}
		\item $F$-step: for the first time the random variable appears,
		\begin{itemize}
			\item $\mathsf{Singleton}$: for random variable $x$ that appears only once throughput the walk, we assign a factor $\sqrt{\frac{n}{d}}\cdot \frac{d}{n}\cdot \singdecay $ for $\singdecay \leq \exp(-d)$;
			\item $F$: for random variable $x$ that appears at least twice throughout the walk, this gets assigned a factor of $1$;
		\end{itemize} 
		\item $H$-step: we assign a factor of $\sqrt{\frac{n}{d}}$.
		\item $R$-step: we assign a factor of  $1$.
	\end{enumerate}

 \end{mdframed}
	 	\paragraph{Vertex-factor redistribution scheme}
	Up front, each new vertex requires a label in $[n]$ when it first appears in the walk. However, we note that we can redistribute this factor throughout the walk and \begin{enumerate}
		\item Assign a factor of $\sqrt{n}$ when the vertex appears for the first time;		\item Assign another factor of $\sqrt{n}$ when the vertex appears for the last time.
	\end{enumerate}
	It can be readily verified that each vertex still gets assigned a factor of $[n]$, modulo the cost of identifying destination of $R$ edges and $H$ edges pointing to this vertex. 
	
	\paragraph{Building the separator}
	For each block, given an edge-labeling $\calL$ (assumed to be traversing from $U\rightarrow V$), we build the separator $S(\calL)$ as the following (and ignore the dependence on $\calL$ when it is clear),\begin{enumerate}
		\item Include any vertex in $U_\al \cap V_\al$ into $S$;
		\item Include any vertex incident to both non-singleton $F$ and $R$ edges into $S$;
		\item Include any vertex in $U_\al$ incident to some non-singleton $F$ edge into $S$;
		\item Include any vertex in $V_\al$ incident to some $R$ edge into $S$;
		\item Include any vertex incident to an $H$ edge into $S$.
	\end{enumerate}

 \subsection{Bounding the block-value from the separator}
 Given a separator $S$ and its corresponding edge-labeling $\calL$, it suffices for us to bound \[
 \edgeval(\calL ) \cdot \vtxcost(\calL)
  \]
  for a given block.
  
  \subsubsection{Factors outside the separator} \label{sec:root-n-outside}
   Notice any vertex outside $S$ is making its first or last appearance, and any vertex making both first and last appearance must be either isolated or it is a vertex outside $U_\al \cup V_\al$ that is also incident to singleton edges only, that said, restricting the vertex cost to the vertices outside $S$, we have \[
  \vtxcost(\calL )[V(\al)\setminus S] \leq \sqrt{n}^{|V(\al)\setminus S|} \sqrt{n}^{|I(\al)\cup F_S(\calL)| }
   \]
   where we let $I(\al)$ be the set of isolated vertices in $\al$, and $F_S(\calL)$ the set of vertices outside $U_\al\cup V_\al$ that are incident to only singleton $F$ edges in $\calL$.
  \begin{lemma}[Singleton-edges come with decay unless floating] For a given edge-labeling $\calL$ with non-floating singleton edges outside the separator $S$, there is an edge-labeling $\calL'$ with separator $S'$ such that $S\subseteq S'$ , $\calL'$ has no non-singleton floating edges, and  \[
  \edgeval(\calL)\cdot \vtxcost(\calL) \leq \edgeval(\calL') \cdot \vtxcost(\calL')
   \]
   Moreover, $S'$ is a separator for $\al$, and restricting to the factors outside the new separator,
  	\[ 
  	\edgeval(\calL')[E(\al)\setminus E(S')] \cdot \vtxcost(\calL )[V(\al)\setminus S] \leq \sqrt{n}^{|V(\al)\setminus S'|} \cdot \sqrt{n}^{I(\al)} \cdot \left(\sqrt{n}\exp(-d)\right)^{|\float(\al\setminus S') | } \cdot \pur(\calL')[V(\al)\setminus S']
  	\]
  	where with some abuse of notation, we let $|\float(\al\setminus S )|$ be the number floating components in $\al$ that are not reachable from $S$, and $ \pur(\calL')[V(\al)\setminus S]
$ denote the cost to settle $\pur$ factor incurred by edges leading to vertices outside the separator $S'$
  \end{lemma}
  \begin{proof} 
   Before we dig into the analysis, it would be helpful for us to remind the reader of the following definitions.  
    \begin{definition}[$U$-wedge and $V$-wedge]
    	We call any connected component that is only reachable from $U_\al$ (or $V_\al$) a $U$-wedge (or a $V$-wedge).
    \end{definition}
  We adopt a two-step strategy, let $V_c(\al\setminus S)$ be the set of non-floating vertices outside the separator, and let $V_f(\al\setminus S)$ be the set of floating vertices that are also not reachable from $S$.
  	\begin{enumerate}
  	\item If $\calL$ has singleton-$F$ edge in any $U$-wedge, we can locally improve the value by turning the $F$ edge into $R$ edge in $\calL'$;
  	\item Furthermore, for the non-floating component that are neither $U/V$-wedges, we modify $\calL'$ from $\calL$ by turning any non-floating singleton $F$ edge into a non-singleton edge;
  	\item We then bound the value by considering the factors on 1) $U/V$-wedges, 2) components reachable from some $U-V$ path, 3) floating components.
  	\end{enumerate}
  
 We now proceed to bound the first inequality such that \[ 
 \edgeval(\calL)\cdot \vtxcost(\calL) \leq \edgeval(\calL') \cdot \vtxcost(\calL')  \]
   We observe the following, \begin{enumerate}
  	\item For vertices not incident to any singleton edges, it suffices for us to adopt the same encoding for it in $\calL'$ as in $\calL$, and thus it contributes the same vertex factor to both sides;
  	\item $S\subseteq S'$ since any vertex in $S$ remains in $S'$, while $S'$ may potentially include more vertices as there may be $U_\al \setminus S$ vertices which get pushed into $S'$ since its incidental $F$ edge gets changed into a non-singleton edge;
  	\item For any edges inside $E(S)$, it remains as a separator edge in $E(S')$ and thus contributes the same edge-value;
  	\item  Any non-singleton edge in $\calL$ contributes the same value in $\calL'$;
  	\item Any singleton edge in $\calL$ contributes a value of $\sqrt{\frac{d}{n}}\cdot \singdecay$ to $\calL$, while a value of $1$ to $\calL'$;
  	\item Any $\pur$ factor's contribution to $\vtxcost$ remain the same on both sides.
  	\end{enumerate}  	
  	
  	 The lemma follows by the following two claims. 
  	 \end{proof} 	  
  	\begin{claim}[Non-floating vertices are offset by singleton decay]
  		\[ 
  \sqrt{n}^{|V_c(\al)\setminus S|}\cdot 	\left(\sqrt{\frac{n}{d}}\cdot \frac{d}{n} \cdot \singdecay\right)^{|\# \text{of non-floating singleton edges}(\calL) |} \cdot \sqrt{n}^{|F_S(\calL) \cap V_c(\al\setminus S)|} \ll \sqrt{n}^{|V_c(\al)\setminus S'| } 	\]
  	\end{claim}
  	\begin{proof}
  		This follows by considering a traversal process from $U_\al$ and $V_\al$, as any non-floating vertex is reachable from $U\cup V$. Any such vertex is reachable using a singleton $F$ edge, and assign its edge-value to the particular singleton vertex it leads to gives us \[
  			\left(\sqrt{\frac{n}{d}}\cdot \frac{d}{n} \cdot \singdecay\right) \cdot \sqrt{n} =\sqrt{d}\cdot \exp(-d)\ll1
  		 \]
  		 Moreover, for each vertex in $S'\setminus S$, note that up front it gives a factor of $\sqrt{n}$ to the LHS and $1$ on the RHS. However, notice that any such vertex has a path to $V_\al$ (otherwise it is not added to the separator $S'$), there is at least one edge along the path that goes from some singleton vertex that remains uncharged (as we charge each edge to the vertex it leads to, and vertex in $V_\al$ by definition does not contribute a $\sqrt{n}$ blow-up as it is not making both first and last appearance). Combining the uncharged edge's value with the vertex factor gives us a factor of \[ 
  		{ \sqrt{\frac{n}{d}}\cdot \frac{d}{n}\cdot \singdecay} \ll  1
  		 \] 
  	\end{proof}
   As a result of this claim, we observe that for a locally optimal edge-labeling, each path between $U,V$  does not use singleton-$F$ edges, and thus it necessarily through some vertex in the separator we construct from the above out the edge-labeling. Therefore, we have recovered the connection between vertex separators and norm bounds of in the dense regime, and moreover, this allows us to bound the block value of each separator. Moreover, it is sufficient for us to switch to the vertex-separator perspective from edge-labeling as any locally optimal edge-labeling must now correspond to some "honest" separator of the shape.

    \begin{corollary}
 	Any non-isolated vertex outside the separator contributes a vertex factor of $\sqrt{n}$ (modulo their corresponding floating and dangling factors via $\pur$). They contribute an extra factor of $\sqrt{n}$ if isolated.
 \end{corollary}
   
  	\begin{claim}[Tree-like floating component gives singleton blow-up]\label{claim:floating-factor}
  \[		 \left(\sqrt{\frac{n}{d}}\cdot \frac{d}{n} \cdot \singdecay\right)^{|\# \text{of floating singleton edges}(\calL) |} \cdot  \sqrt{n}^{|V_f(\al\setminus S)| } \leq (\sqrt{n}\cdot  \singdecay)^{|\float(\al\setminus S)|} \] 	\end{claim}
  \begin{proof}
  	The previous argument for handling singleton edges crucially replies upon being able to assign each vertex to an edge that comes with singleton decay. However, this is no longer true for floating component not reachable from the separator, especially tree-like floating component: for example, consider an edge floating around in $\al$. Should the edge be a singleton edge in the walk outside the separator, it contributes a factor \[
	n^2 \cdot \sqrt{\frac{d}{n}}\cdot \singdecay \approx n^{1.5} \exp(-d) \gg \sqrt{n^2}\cdot (2|V(\al)|q_\al)	 \]
	where we recall $\sqrt{n^2}\cdot (2|V(\al)|q_\al)	$ is the cost we would expect if there is no conditioning. That being said, we may have at most one $\sqrt{n}\cdot \singdecay$ blow-up for each component only for the first vertex in the floating component outside the separator, as the rest of vertices in the floating component can be charged by the singleton edge leading to it by the previous argument, and this proves the claim above.
	%	\textbf{Non-singleton floating component} The argument is identical to the non-floating case except for the first vertex in the component (when being traversed the second-time or later), giving us a bound of \[ 
%		\sqrt{n}^{V|(C_i)|}\cdot \dang(C_i) \cdot \underbrace{2|V(\al)|q_\al}_{\float(C_i)}
%		\] 
%		\end{enumerate}
%		Combing the above, for each floating component, it suffices for us to take \[ 
%		\float(C_i) \leq \max(\sqrt{n}\cdot (\singdecay)^{V(C_i)-1}, 2|V(\al)|q_\al) \leq \sqrt{n}\cdot \exp(-d)
%		\]
  \end{proof}
  	\begin{remark}
		A natural question is whether the $\sqrt{n}$ blow-up from conditioning is tight. We remark that it is possible as it comes down to bounding $\E[\chi_e \cdot 1_\calF ]$ where we recall $\calF$ is the conditioning on bounded-degree, and we can ignore the conditioning on $2$-cycle for this part. \begin{align*}
			\E[\chi_e \cdot 1_\calF ] &= \Pr\left[1_\calF \cap \left(\chi_e=\sqrt{\frac{1-p}{p}} \right)\right] \cdot \sqrt{\frac{1-p}{p}} + \Pr\left[1_\calF \cap \left(\chi_e=-\sqrt{\frac{p}{1-p}} \right)\right] \cdot (-\sqrt{\frac{p}{1-p}})\\
			&=\Pr\left[1_\calF|\chi_e=\sqrt{\frac{1-p}{p}}\right]\sqrt{\frac{1-p}{p}}\cdot p - \Pr\left[1_\calF|\chi_e=-\sqrt{\frac{p}{1-p}}\right] \cdot \sqrt{\frac{p}{1-p}}\cdot (1-p)\\
			& = (\sqrt{(1-p)p}) \left( \Pr\left[1_\calF|\chi_e=\sqrt{\frac{1-p}{p}}\right] - \Pr\left[1_\calF|\chi_e=-\sqrt{\frac{p}{1-p}}\right]\right)
			\end{align*}
			To bound the difference in the probability, we note that this is equivalent to resampling all the edges but $e$, what's the probability that at least one of the two endpoints of $e$ hit degree $10d-1$ (suppose we truncate at degree $10d$), and an $\exp(-d)$ bound here is the best we can hope for, which would give a bound of $\E[\chi_e\cdot 1_\calF] \approx \sqrt{\frac{d}{n}} \cdot \exp(-d)$ . Notice we need an extra $\frac{1}{\sqrt{n}}$ decay to offset the blow-up for a singleton $F$ component.
	\end{remark}

\subsubsection{Bounding the floating factor}\label{subsec:floating-bound}

	\begin{claim}
		It suffices for us to assign a factor of $\sqrt{n}\exp(-d)$ for each floating component, i.e., $\float(\al) \leq \left(\sqrt{n}\exp(-d)\right)$ or a factor of $(2|V(\al)|q_\al)^{O(1)}$, eithe of which can be bounded by $\left(\sqrt{n}\exp(-d)\right)$.
	\end{claim}
	\begin{proof} We case on whether the floating component contains a cycle. Observe that the edges in the same floating component have to receive the same $F/R/H$ label under the edge-labeling, otherwise it can be locally improved.
	\begin{enumerate}
		\item \textbf{Singleton floating component}  We first consider $F$-component, where some extra care is warranted as singleton edges may now be troublesome. The previous argument for handling singleton edges crucially replies upon being able to assign each vertex to an edge that comes with singleton decay. However, this is no longer true for floating component, especially tree-like floating component: for example, consider an edge floating around in $\al$: should the edge be a singleton edge in the walk, it can contribute \[
	n^2 \cdot \sqrt{\frac{d}{n}}\cdot \singdecay \approx n^{1.5} \exp(-d) \gg \sqrt{n^2}\cdot (2|V(\al)|q_\al)	 \]
	where we recall $\sqrt{n^2}\cdot (2|V(\al)|q_\al)	$ is the cost we would expect if there is no conditioning. That being said, we may have one $\sqrt{n}$ blow-up, giving a bound of \[ 
	\sqrt{n}^{|V(C_i)|}\cdot \underbrace{\sqrt{n}\cdot \singdecay^{V(C_i)-1} }_{\float(C_i)} \cdot 1_{C_i\cap S =\emptyset } 	\]	
	where we have the indicator $1_{C_i\cap S =\emptyset } $ as this can only happen when $C_i$ is outside the separator.
	
		\item \textbf{Non-singleton floating component} The argument is identical to the non-floating case except for the first vertex in the component (when being traversed the second-time or later), giving us a bound of \[ 
		\sqrt{n}^{V|(C_i)|}\cdot \underbrace{\dang(C_i) \cdot 2|V(\al)|q_\al}_{\float(C_i)}
		\] 
		\end{enumerate}
  \begin{remark}
      We bound $\dang(C_i) \cdot 2|V(\al)|q_\al \leq |2|V(\al)q_\al|^{O(1)}$ loosely in the final bound as we consider each landing at the floating component a surprise visit, which comes with $O(1)$ factors of $2|V(\al)|q_\al$.
  \end{remark}
		Combing the above, for each floating component, it suffices for us to take \begin{align*}
		    \float(C_i) &\leq \max(\sqrt{n}\cdot (\singdecay)^{V(C_i)-1}\cdot 1_{C_i\cap S =\emptyset } , (2|V(\al)|q_\al)^{O(1)} ) \\&\leq \max(\sqrt{n}\cdot \exp(-d)\cdot 1_{C_i\cap S =\emptyset } , (2|V(\al)|q_\al)^{O(1)})
		\end{align*}
		\end{proof}

\subsubsection{Bounding the block value for any separator of the shape}
From the above, we have established any labeling corresponds to a separator of the shape with higher block value if the labeling does not produce a separator of the shape itself due to singleton edges. That said, to bound the value of the maximal value-labeling, it suffices for us to consider separators of the shape. 

Let's now try to bound the block value $\edgeval(\calL)\cdot \vtxcost(\calL)$ for any locally optimal $\calL$, let $S$ be the corresponding separator, let $S_L$ be the vertices in $S$ that are reachable from $U_\al$ without passing through any vertex in $S$,
  \begin{enumerate}
  	\item Each vertex outside $S$ contributes a factor of $\sqrt{n}$;
  	\item Each vertex in $S_L$ is reachable via an $R$ edge, hence there is no vertex factor; 
	\item Each vertex in $S$ that is not doubly constrained in the encoding process is arrived along an $H$ edge, and hence requires at most a $\sqrt{d}$ vertex factor if non-dangling, and a factor of $d$ if dangling;
	\item Each edge receiving an $H$ label is in $E(S)$ and contributes a value of at most $\sqrt{\frac{n}{d}}$;
    \item Each dangling branch contributes a factor of $\dang$, that is at most a factor of $2|V(\al)|q_\al$ per dangling branch via either $\pur$ or $H$-visit;
    \item Each floating component, depending on whether it's locally tree-like, either contributes a factor of at most $\sqrt{n}\exp(-d) $, or a factor of $(2|V(\al)|q_\al)^{O(1)}$ via $\pur$ but not both.
\end{enumerate} 

For intuition, we recommend the reader to ignore factors from excess $F$ edges, and dangling vertices. However as stated, the above bound is not easy to apply as it is not immediately clear what vertex factors are needed for encoding vertices inside $S$; fortunately, this can be simplified once we restrict our attention to a class of separators, \emph{irreducible} separators, without any loss in the SMVS value.
 
 We now unpack the factors from the above in a more illustrative manner, and show that it suffices for us to assume each vertex in $S$ is doubly constrained by $S_L$, and moreover, there is no dangling vertex inside $S$.
%\begin{observation}[$\sqrt{n}$ for each vertex outside $S$, and an extra $\sqrt{d}$ if dangling]
%	Charging the edge-value of the $R$ edge leading to that vertex in the BFS process from $S_L$ if the vertex is in $L(U_\al\setminus S_L, S_L)$, or the edge value of the non-excess $F$ edge that leads to the vertex in the BFS from $S$, and the vertex cost, each (non-dangling) vertex in $V(\al)\setminus S$ contributes a value of $\sqrt{n}$. Each vertex dangling picks up an extra $\sqrt{d}$ factor and contribute a total of $\sqrt{nd}$.
%\end{observation}
%\begin{observation}[$\sqrt{\frac{n}{d}}$ for each separator edge] Each edge inside the separator receives an $H$ label, and contributes a value of $\sqrt{\frac{n}{d}}$.
%\end{observation}

\subsubsection{Not too many surprise visits leading to vertices outside $S$ }\label{sec:not-too-many-exc}
We bound the contribution to the block-value via $\pur$ factors, and notice unless floating components and dangling components are involved, each such factor corresponds to a surprise visit (i.e. an excess edge in the case it is inevitably surprised due to shapes containing cycles). We remind the reader that $\sqrt{d}^{|\exc(F)|}$ is an artifact of our encoding since we encode a potential direction for future unforced return that may be incurred due to the excess edge, however, we point out that this is not too terrible a cost to pay as in any optimal labeling, there cannot be too many $\exc(F)$ edges, and we can handle the blow-up via edge-decay in the optimal labeling. \begin{proposition}
	Let $\calL$ be a locally optimal labeling, and $\exc(F)$ be the number of $F$ edges that lead to visited vertices when we consider a BFS from $S$,  \[ 
	\exc(F)< \lceil \frac{|F\setminus \exc(F)|}{\log_d n}\rceil
	\]
\end{proposition}
\begin{proof}
This follows from the optimality of SMVS, consider the new separator formed by including all vertices reached by $F$ edges, which gives us a change of $
\sqrt{n}^{-|F\setminus \exc(F)| } \cdot \sqrt{\frac{n}{d}}^{|F|}<1
$.
\end{proof}

\subsubsection{Surprise visits arriving at $S$ can be locally improved}
Fixing the traversal direction from $U_\al$ to $V_\al$, we case on whether the edge comes from outside the separator. Suppose so, an edge from $V(\al)\setminus S$ to a vertex $v\in S_L$ may be a surprise visit if it is an $F$ edge, that said, we observe that we can locally improve the value by traversing the path from $v\in S_L$ to $a \in U_\al$ (recall that $v$ is reachable from $U_\al$ without passing through $v$, hence such a path is well-defined), and turn each $F$ edge into an $R$ along the path. Noice we have the following change of value, \begin{enumerate}
	\item Let $\calL$ be the edge-labeling before the flip, and $\calL'$ after the flip;	\item By our previous clean-up, we can assume there is no singleton edge involved;
	\item Any vertex not in $U_\al$ along the path outside the separator contributes a factor of $\sqrt{n}$ in both labelings;
	\item The edge leading to $v$ is a surprise visit, and contributes a factor of $(2|V(\al)|q_\al)^{c_S}$ via $\pur$ factor in $\calL$ but not in $\calL'$;
	\item The vertex $a \in U_\al$ (as well as any other vertex in $U_\al$ reachable from $v$) contributes a factor of $\sqrt{n}$ to $\calL'$ but not $\calL$
\end{enumerate}

Therefore, we pick up a change of at least \[ 
\frac{\edgeval(\calL') \cdot \vtxcost(\calL')}{\edgeval(\calL) \cdot \vtxcost(\calL)} \geq \frac{\sqrt{n}}{(2|V(\al)|q_\al)^{c_S}} \gg 1\,.
\]

On the other hand, suppose the edge is contained in $E(S)$, observe that flipping the edge from $F$ to $H$ does not affect the vertex factor of its origin endpoints, while its destination may originally contribute a cost of $2|V(\al)|q_\al$, and it may now be doubly constrained and does not contribute any vertex cost in $\calL'$. That said, the $H$ edge is now contributing a factor of $\sqrt{\frac{n}{d}}$ through its edge-value, and hence, we pick up a change of \[ 
\frac{\edgeval(\calL') \cdot \vtxcost(\calL')}{\edgeval(\calL) \cdot \vtxcost(\calL)} \geq \frac{\sqrt{\frac{n}{d} } }{(2|V(\al)|q_\al)^{c_S}} \gg 1\,. \]
\subsubsection{Pruning superfluous $H$ edges}\label{sec:pruning-superfluous-H}

\begin{definition}[Irreducible SMVS]
	We call a separator $S$ an irreducible SMVS for a shape $\al$ if for any vertex $v\in S$ of degree $1$ in $E(S)$,  \begin{enumerate}
		\item There are two vertex-disjoint paths $P_U, P_V$ from $v$ to $U_\al$ and $v$ to $V_\al$
		\item Both paths do not pass through any other vertex in $S$
	\end{enumerate}
\end{definition}

\begin{claim}
	For each locally optimal $\calL$ labeling that produces separator vertices in $S(\calL)$ not doubly constrained by $S_L(\calL)$, there is a labeling $\calL'$ with at least the same SMVS value while each vertex in $S(\calL')$ is doubly constrained by $S_L(\calL')$. Moreover, there is no dangling vertex in $S(\calL')$.
\end{claim}
\begin{proof} We start by considering non-dangling vertices in $S$.
	We first notice for any vertex in $S_L$ of degree $1$ in $S$, unless it is incident to an $F$ edge, we can pop it out of the separator and obtain the same value by turning the incident $H$ edge into an $R$ edge: the $H$ edge contributes at most a $\sqrt{d}$ vertex factor to indicate the destination of the edge and it contributes a factor of $\sqrt{\frac{n}{d}}$ through the edge-value, while we obtain the same value by picking the $\sqrt{n}$ value from the vertex being outside the separator. 
	
	Let $\calL'$ be the labeling obtained by applying the above procedure on $\calL$, and let $S = S(\calL')$ be its corresponding separator. Consider a BFS from $S_L$ to traverse vertices in $S\setminus S_L$, and note that each vertex in the separator not doubly constrained by $S_L$ has only one simple path to $S_L$, and it is not on any path from $S_L$ to a cycle. In other words, given $S_L$, each non-doubly constrained vertex forms a tree rooted at $S_L$. For each branch, starting from the leaf vertex, notice each vertex (if non-dangling) contributes to the block value via the edge that leads to that vertex from $S_L$ a factor of $\sqrt{\frac{n}{d}}$ and a vertex factor $\sqrt{d}$ for indicating the destination of the $H$ edge, giving a total of $\sqrt{\frac{n}{d}}\cdot \sqrt{d}=\sqrt{n}$. However, the same value can be attained by flipping the $H$ edge to an $F$ edge and pop this vertex out of the separator. 
	
	For each dangling vertex, the pop-out procedure is rather identical except we now need to take into account the dangling factor. Suppose it is indicated using a label in $d$, it contributes a value of $\sqrt{\frac{n}{d}}\cdot d=\sqrt{nd}$ while we can again pop the vertex out of the separator and obtain a value of $\sqrt{nd}$ with the extra $\sqrt{d}$ coming from the dangling factor $\dang$; analogously, we can generalzie this argument for long dangling branch inside the separator, we can again use the idea in Proposition~\ref{prop:long-high-mul-path}, and obtain a $\sqrt{2|V(\al)|q}$ bound as we would have 
%		For each dangling branch, notice each edge comes with a value of $\sqrt{\frac{n}{d}}^{\branch(i)}$, while each time a label in $2|V(\al)|q$ or a sequence of labels in $d$ is needed, giving us a total of 
\begin{align*}
		\sqrt{\frac{n}{d}}^{\branch(i)}\min(2|V(\al)|q, d^{\branch(i)}) &\leq 
	\sqrt{n}^{\branch(i)} \min(\frac{2|V(\al)q}{\sqrt{d}^{\branch(i)} }, \sqrt{d}^{\branch(i)}) \\&\leq \sqrt{n}^{\branch(i)}\min(\sqrt{2|V(\al)|q} ,\sqrt{d}^{\branch(i)})
	\end{align*}
	Note that this is the same value we would pick up from vertex factors if we pop out the entire dangling branch, and this completes our proof to the claim.
%%	Recall that the factors of $\sqrt{n}^{\branch(i)}$ is already taken into account by the vertex factors of $\sqrt{n}$ for each vertex outside the separator, hence, each block's dangling blow-up is still at most $\min(\sqrt{2|V(\al)|q} ,\sqrt{d}^{\branch(i)})$ for each branch. Combining the argument for all dangling branches gives us the desired proposition.
%%
%%	
\end{proof}
\subsubsection{Wrapping up block value bound for a single shape} 
We are now ready to wrap up our bound for the block value. Recall that under our setup via edge-labeling, for a particular block of shape $\al$, our bound proceeds as \[
B(\al) = \sum_{\calL} B(\calL) = \sum_{S:\text{separator}} B(\calL_S)
\]
and we first determine the dependence of a fixed labeling.
\begin{proposition} \label{prop:norm-for-fixed-shape}
	For any irreducible locally optimal labeling $\calL$, let $S$ be its corresponding separator, \[		\vtxcost(\calL)\cdot \edgeval(\calL) \leq  \cnorm^{|V(\al)\setminus S| }\cdot \sqrt{n}^{|V(\al)\setminus S|}\cdot \left(\sqrt{\frac{n}{d}}\right)^{|E(S)|}\cdot \dang(\al\setminus S)\cdot \float(\al )\cdot \sqrt{n}^{I(\al)}
	\]
	with \[ 
	\dang(\al\setminus S) \leq \prod_{b_i\in \branch(\al\setminus S)}\min \left(\sqrt{d}^{|b_i|}, 2|V(\al)|q_\al\right)
	\]
	and 
	\[
	\float(\al) =  \prod_{C_i \in F(\al)}\left(2|V(\al)|q_\al, \sqrt{n}\cdot \singdecay^{|E(C_i)|} \cdot 1_{V(C_i)\cap S = \emptyset } \right) 
	 \]
	for $F(\al)$ the collection of floating components in $\al$, and  for $I(\al)$ the number of isolated vertices in $\al$.
\end{proposition}
\begin{proof}
This follows by recalling the factors accounted from above, 
	\begin{enumerate}
  	\item Each vertex outside $S$ contributes a factor of $\sqrt{n}$  (via \cref{sec:root-n-outside});
	\item Each vertex in $S$ is doubly constrained by the boundary of $S$ (via \cref{sec:pruning-superfluous-H}), unless floating which is captured by the factor in $|\float(\al)|$ ;
	\item Each edge receiving an $H$ label is in $E(S)$ and contributes a value of at most $\sqrt{\frac{n}{d}}$;
	\item For dangling vertices, we pick up an additional dangling factor $\dang(\al\setminus S)$; each floating component comes an extra factor captured by $\float(\al)$ via Proposition~\ref{prop:improved-dangling} and Claim~\ref{claim:floating-factor};
 \item  Each excess $F$ edge that leads to visited vertices incur an extra cost of $|2V(\al)q_\al|^{O(1)}$ via $O(1)$ of the $\pur$ factor it may incur, however, this is subsumed by vertex decay outside the separator (via \cref{sec:not-too-many-exc}).
	\end{enumerate} 
\end{proof}
Finally, we obtain the bounds by summing over the choice of separators, and notice that each vertex in $U_\al\cap V_\al$ is in the mandatory separator, 
\begin{corollary}\label{cor:block-value-single-shape}
	For any shape $\al$, we can bound its block-value function by \begin{align*}
B(\al)  &= \sum_{S:\text{separator}} \vtxcost(\calL_S)\cdot \edgeval(\calL_S) \\ &\leq \sum_{S:\text{separator}} \cnorm^{|V(\al)\setminus S |} \cdot \sqrt{n}^{|V(\al)\setminus S|}\cdot \left(\sqrt{\frac{n}{d}}\right)^{|E(S)|}\cdot \dang(\al\setminus S)\cdot \float(\al )\cdot \sqrt{n}^{|I(\al)|}		\\
&\leq (\cnorm')^{|V(\al)\setminus U_\al\cap V_\al| }\sqrt{n}^{|V(\al)\setminus S|}\cdot \left(\sqrt{\frac{n}{d}}\right)^{|E(S)|}\cdot \dang(\al\setminus S)\cdot \float(\al )\cdot \sqrt{n}^{|I(\al)|}	
	\end{align*}
	where we define 
\[ 
	\dang(\al\setminus S) \leq \prod_{b_i\in \branch(\al\setminus S)}\min \left(\sqrt{d}^{|b_i|}, 2|V(\al)|q_\al\right)
	\]
	and 
	\[
	\float(\al) = \prod_{C_i \in F(\al)}  \left(2|V(\al)|q_\al, \sqrt{n}\cdot \singdecay^{|E(C_i)|} \cdot 1_{V(C_i)\cap S = \emptyset } \right) \]
 where $F(\al)$ is the collection of floating components in $\al$, and $I(\al)$ the set of isolated vertices.
\end{corollary}
\subsection{Bounding the drawing costs for unknown shapes}

\begin{definition}[Wedge bundling]
	Given a set of vertices $S\subseteq [n]$ , and a wedge $W$, we call it bundled if $U_W$ is labeled as a subset of $S$. Moreover, edges that are part of the wedge also receive labels in $\binom{n}{2}$ as a subset.
\end{definition}
To bound the encoding cost for bundling wedges, we give a query algorithm for bundling vertices in $U_\al$ (and respectively, in $V_\al$).
\paragraph{Set-up for bundling}
Let's start with a set-up that allows us to turn the permutation question into a query problem.
\begin{definition}[Wedge and its legs]
Wedge $W$ is a shape with only one labeled set $U_{\calW}$(as opposed to $U_\al$ and $V_\al$), and each vertex in $i\in U_{W}$ has a path to some $j\neq i \in U_{W}$ in $\calW$, and each vertex in $V(\calW)\setminus U_{\calW}$ is connected to the labeled set $U_{W}$.\end{definition}
\begin{definition}
	We additionally call a wedge an \emph{irreducible} wedge if every leg is of degree $1$, i.e., for any $i\in U_{W}$, $\deg_W(i) = 1$.
\end{definition}
\begin{remark}
	Given a wedge $W$ that is not irreducible, it can be reduced to some irreducible wedge by removing non-degree-1 vertex from $U_W$.
\end{remark}
We are now ready to introduce the central question for bounding the combinatorial factor of wedge-bundling:
\begin{question}
	Given $G$ a random graph sample, particularly it is guaranteed to be $2$-cycle free at $\kappa\approx 0.4\log_d n$ neighborhood of any vertex, and given a list of special vertices $S$ that we want to match, and a list of vertex-disjoint wedges $\{\calW_i\}$ that are subgraphs of $G$, can we ask each edge in the wedge $O(1)$ Yes/No questions and identify each wedge $\{\calW_i\}_{i\in t}$ as an edge set? 
\end{question}
\subsubsection{A query algorithm for wedge bundling}
We introduce our bundling procedure via a query algorithm.

\begin{lemma} \label{lem: bundling_cost}  Given a sequence of wedges  $\calW = \{\calW_i\}$, and their legs given as a set of vertices $\bigcup_i U_{\calW_i}$,  the wedges can be bundled (identified as an edge set) at a cost \[\cost \calB \leq   (c_{\matched})^{\sum_{i} |E(\calW_i)| }\]
%\] \cdot \prod_i|\aut(\calW_i)|
%\] 
for some absolute constant $c_\matched$ independent of $d$.
\end{lemma}
%
%On a high level, we note that the cost of $(c_{\matched})^{E(\sum_{i} E(\calW_i)) } $ is incurred by the bundling stage, and the automorphism factor is only used for matching the vertices 
%
For starters, we first consider a \emph{merging} process for a irreducible wedge as the following,
%\begin{enumerate}
%	\item For each wedge with $U_{\calW}$ the set of labeled leg vertices, pick $a,b\in U_{\calW} = \arg\min_{i,j\in U_{\calW}} \dist_{\calW}(i,j)$ that minimizes the pairwise distance on the wedge;
%	\item For any $c\in U_{\calW}$ s.t. $\dist_{\calW}(a,c) = \dist_{\calW}(a,b)$ and $P_{a,c}$ shares edges with $P_{a,b}$, we define the path $P_{a,c}$ ;
%	\item Contract all the vertices on the path $P_{a,b}\cup P_{a,c}\cup P_{b,c}$ into a single node $S$, and put $S$ into the labeled set $U_{\calW}$
%	\item Let $t$ be the midpoint of the wedge path connecting $a,b$, 
%	\item Repeat this process until there is one single vertex in the labeled set of the wedge.
%\end{enumerate}
\begin{mdframed}
	\begin{enumerate}
	\item For each wedge, pick an arbitrary vertex $t $ as a root, and traverse the wedge from the root; notice by removing cycle edges (edges leading to visited vertices in the traversal) from the wedge, this is a tree with leaves at the labeled set $U_\calW$, and the traversal from the root allows us to assign a direction for each tree edge;
	\item For each non-leaf vertex traversed that has degree more than $2$, we call this a \emph{branch} vertex;
	\item For each leaf node $v$, consider its closest branch node on its path to $t$, and let the branch node be $B$; let $\{b_i\}_B\subseteq U_{\calW}$ be the labeled leaf nodes of $B$, let $b_1,b_2$ be the two closest leaf nodes to $B$ among its children $\{b_i\}_B$ (with ties broken arbitrarily).
	\item For each branch node $b\in V(\calW)$, let $c_b$ be its number of downward (children) branches, we associate it with \begin{itemize}
		\item $\mathsf{Up}[b] \in \N$: the upward-bundle query distance of $b$;
		\item $\mathsf{Down}[b]\in \N^{c_b-2}$: a list of downward-bundle query distance of $b$
	\end{itemize}
	\item For each leaf node $u\in U_{\calW}$, we associate it with $\mathsf{Up}[u]$ its upward-bundle query distance.
\end{enumerate}
\end{mdframed}

\paragraph{Bundling process using upward/downward-bundle distance}
Before we proceed to bound the auxiliary data required and the query number, let's first see how these auxiliary data ($\mathsf{Up}$ and $\mathsf{Down}$) can be helpful for bundling vertices from the given labeled set into wedges. In the beginning of the bundling process, no branch vertex is revealed to the decoder, while each leg vertex in $U_{\calW}$ comes with an upward-bundle distance, which we claim is sufficient to for the decoder to identify the branch nodes throughout the process.
 
   For the sake of illustration, we restrict to a single wedge case while its generalization to multiple wedges is straightforward. Throughout the process, we maintain $W$ as a set of to-be-bundled vertices and start with $W= U_{\calW}$ (for the general case, we have $W=\cup_{\calW} U_{\calW} $).  With the sequence of upward bundle distance given to the decoder, the decoder may attempt to bundle revealed vertices from $W$ by querying the decoder each pair of possible bundling (i.e. two vertices in $W$ that have the same upward \footnote{or upward-downward if matching a branch node with a leaf node}-bundle query distance, and there is a path in the revealed walk of the same length connecting these two vertices). Any bundled path may give rise to at most one branch node, and it can be specified by traversing the matched path, and we include the branch node into the active set $W$ while removing vertices that are bundled using upward match query distance. Once a branch node is identified along the process, the decoder may encode the auxiliary data for the branch node: the data may be either the next unbundled downward-match query distance (if any), or the upward-match query distance of the current branch node. 
   
 \begin{observation}
 	The sequence of upwards matched distance for each leg vertex in $U_{\calW}\setminus V_\alpha $ can be specified to the decoder when each its incident edge is specified from the last block (which exits guaranteed by permissible shape).
% 	 Moreover, notice $U_{\calW}\cap (U_\al\cap V_\al) =\emptyset$ as any such vertex would be isolated after removing $F$ edges, or there would be short wedges around $U_\al$ and $U_\al\cap V_\al$ that can be factorized.
 \end{observation}
 
	\subsubsection{Bounding auxiliary data cost and bundle-attempts} 
	 We now proceed to bound the cost of the auxiliary data and the number of bundle-attempts by the decoder. Let's start with the auxiliary data used for specifying upwards/downwards matched query distance for each leg/branch vertex. For intuition, recall that we pick up a constant edge decay for each edge in $E(\al)\setminus E(U_\al\cap V_\al)$, and our goal is to show these auxiliary data can be offset by the edge decays from each block.
\begin{claim}
	For any sequence of wedges $\{\calW_i\}$, \[ 
	\sum_{\calW_i} \sum_{i\in V(\calW_i)} \mathsf{Up}[i] + \mathsf{Down}[i]\leq 2 \sum_{\calW_i}|E(\calW_i)|
	\]
\end{claim}

\begin{proof}
It suffices for us to prove this for a fixed wedge, and notice $\mathsf{Up}$ is only defined for branch nodes and leaf nodes, while $\mathsf{Down}$ only for branch nodes. First consider $E'(\calW)\subseteq E(\calW)$ that is obtained by removing cycle edges, and notice it suffices for us to consider edges in $E'(\calW)$ as edges in $E(\calW)\setminus E'(\calW)$ only contribute to the RHS. 
	 
	 The claim then follows by noticing $E'(\calW)$ is a tree with vertices being branch nodes and leaf nodes, and each edge contributes at most $2$ to the LHS.
\end{proof}

\begin{claim} \label{claim:bundle-bound} For a sequence of wedges $\{\calW_i\}$,
	the number of match queries is at most $O(\sum_{i} E(\calW_i))$. 
\end{claim}
%
%\begin{lemma} For any wedge, edge decay is sufficient to encode $D_{\matched}$ for each vertex $i\in U_{\calW}$. Formally,	\[
%	\sum_{i\in \calW} D_{\matched}(i) \leq 2 E(\calW)
%	 \]
%\end{lemma}
It suffices for us to bound the number of ''No'' queries throughout the bundling process. Towards this end, notice we can bound the number of ''No'' queries at each search level $D$ (either ''Upwards'' or ''Downwards''-bundle distance) and charge the ''No'' queries to the number of edges matched at the particular search level.

	\begin{definition}[Wedge-search graph for level $D$]
	Given a number $D$, we define its corresponding search graph as, \begin{enumerate}
		\item Each vertex in the graph is a labeled (in $[n]$) vertex (either leaf or branch node) in the execution of the bundling process that is to-be-matched with some other vertex at distance $D$ in the graph;
		\item Each edge corresponds to a query of the decoder asking whether the edge is part of a wedge connecting the labeled vertices.
	\end{enumerate}
	\end{definition}
%	With the auxiliary graph in hand, we further notice we can contract each wedge in the graph into a wedge node s.t. each wedge node contains at least $2$ labeled vertices of the corresponding $D_\matched$ level, and each labeled node of a wedge may have a path of length $D_\matched$ to some labeled node of another wedge corresponding to a ''no'' query. We first notice each labeled vertex of a wedge is connected with each other inside the wedge using a path of length $D_{\matched}$ (the matched wedge); furthermore, if we contract each''no'' query corresponds to a edge in the auxiliary wedge graph (allowing self-loop), the task of bounding ''no'' queries reduces to bounding the number of edges in the auxiliary graph.
	To bound the number of queries at each level $D$, we notice any subgraph (even of size polynomial in $n$) of a random graph is sparse despite potentially denser than $G_{n,d/n}$.
%	
%	crucially rely upon the underlying random graph sample $G$ is $2$-cycle free within $\kappa$ radius. We first observe that any cycle of length $C$ in the auxiliary graph  $G_{D}$ 
%	 corresponds to a cycle of length at most $C\cdot D$ in the underlying graph. Since the random graph is $2$-cycle free at radius $\kappa$, the auxiliary graph is also $2$-cycle free at radius $\frac{\kappa}{D}$. We then appeal to the Moore's bound for bounding the number of excess edges (in its $2$-cycle free version).
%	 
%	 
%	 
	 
%	 \begin{proposition}
%	 	Let $H$ be a $v$-vertex graph that is bicycle free at radius $r$, and assume $r\geq 10\log v$, then\[ 
%	 	\exc(H)\leq \frac{\log(ev)}{r}\cdot v
%	 	\] 
%	 \end{proposition}
	 \begin{proposition}(Sparsity of small subgraph in $G_{n,p}$) For $G\sim G_{n,p}$ for $p<\frac{1}{2}$, with probability at least $1-O(\frac{1}{n^6})$, every subgraph $S$ of $G$ such that $|V(S)| \leq (\frac{1}{p} )^{1/2} $, \[
	|E(S)| \leq 3|V(S)|\log_{1/p}(n) 
	 \]	
	\end{proposition}
	\begin{proof}
		This follows by first moment. Let $e^*(v) = 3v\log_{1/p}(n) 
$,\begin{align*}
			\Pr[\exists \text{a dense subgraph}] &\leq \sum_{v=2}^{(1/p)^{1/2} } \sum_{e=e^*(v)}^{\binom{v}{2}}\binom{n}{v}\cdot (2v)^{e}p^e\\
			&\leq \sum_{v>2} \sum_{e=e^*(v)}^{\binom{v}{2}} 2^{v\log n-v\log v+e\log p}\\
			&\leq \frac{1}{n^6}
		\end{align*}
	\end{proof}
	\begin{corollary}
		W.h.p., for $D<0.2\log_d n$ and $p=\frac{d^D}{n}$, and $|V(S)|\leq \sqrt{\frac{n}{d^D}} \leq n^{0.4} $, we have $
		|E(S)|\leq 	5|V(S)|	 $.
	\end{corollary}
%	\begin{corollary}
%		Plug in $p=\frac{d}{n}$, for any $S$ of size at most $(\frac{n}{d})$, we have \[ 
%		|E(S)|\leq 
%		\]
%	\end{corollary}
%	 \jnote{trouble arises for $r\geq 10\log v$ condition....}
	  \begin{proof}[Proof to claim \ref{claim:bundle-bound}]	  
	 For starters, observe for $D\geq 0.2\log_d n$, edge decay along the wedge is sufficient for us to assign two label in $T_D$ to hardwire the match, and it suffices for us to focus on short wedges. For each bundle query distance $D$, let $T_D$ be the number of vertices querying at distance $D$ at any stage in the algorithm, notice the auxiliary graph is a random graph on $T_D$ vertices with each edge present with probability at most $(\frac{d}{n})^D\cdot \binom{n}{D-1} \leq \frac{d^{D}}{n} $ as this is a path of length $D$ with both endpoints fixed in a $G_{n,d/n}$. By the edge bound from above, we have at most $5T_D$ queries (with each query requires a yes/no answer) , while we have edge decay from $T_D\cdot D/2$ many edges from the edges of the wedges matched. To see this, notice each vertex at this level matches some other vertex via a path of length $D$, and charging the edges on the path to both endpoint gives the desired bound, \[ 
	 2^{5\cdot T_D}\ll c^{O(T_D\cdot D)}	 \]
for some $c>1$
%	 
%	 we case on whether $D>\log T_D$. If so, notice if suffices for us encode the matched vertex's label in $T_D$ for each vertex query, as we have \[
%	 (T_{D})^{T_D}  = 2^{T_D\log_2 T_D} \leq c^{T_D\cdot D}
%	  \]
%	  for some vertex $c>0$. In other words, the edge decay of the wedge edges are sufficient for encoding labels in $T_D$.
%%	 $\frac{\kappa}{D}>10\log T_D$. If so, notice we pick up at least 
%%	 \[ 
%%	  c_1^{T_D\cdot D}	 \]
%%from the edges bundled at this level, while we need a cost of $T_D!\leq 2^{T_D\log T_D}$	 
%%	  to encode a permutation on $T_D$ vertices.
%%	  
%%	   Setting $D_{\max} = \sqrt{\kappa}$, for any $D<D_\max$ and  $\frac{\kappa}{D}>10\log T_D$,	 we have 
%%	  \[ 
%%	  T_D\log T_D \leq T_D \cdot \frac{\kappa}{D}\leq 
%%	  \]
%Otherwise, suppose $D<\log T_D$ implies $\frac{\kappa}{D}\geq 10\log (T_D) $ (the assumption for the above proposition), i.e. $T_D \leq 2^{\frac{\sqrt{\kappa} }{10}} $ , as we would then have \[ 
%D\cdot 10\log (T_D)\leq 10\log^2 (T_D)\leq \kappa
%\]
%	 Applying the proposition with $T_D$ the number of vertices querying at distance $D$ at any stage in the algorithm gives us the excess bound \[ 
%	 \exc(G_{D}) \leq \frac{\log (e\cdot T_D)}{\kappa/D }\cdot T_D\leq O(1)\cdot T_D\cdot D
%	 \]
%	 where the last inequality follows from $D_{\matched}\leq \frac{1}{10}\kappa$ and $T\leq n^\epsilon$ for some $\eps\leq O(\frac{1}{\log d}$).
%	 Finally, we notice we pick up at least $\frac{T_D}{2}\cdot D$ edges from the edges on the bundled paths in the query at level $D$.
		 
	 Summing over all levels of $D$  give us the desired bound.
	 \end{proof}
%	 \begin{remark}
%	 	Since $\kappa=0.4\log_d n$ for us, this holds $U_\al, V_\al$ of size at most $2^{c\cdot \sqrt{\log n}}$ as their size is a trivial upper bound for $T_D$ for any $D$.
%	 \end{remark}
%
\subsection{Matching bundled wedges}
In this subsection, we show the cost for assigning each vertex a label in $\calW$, and matching wedge's legs from the bundled edge set can be handled by vertex decay. As a thought process, let us momentarily forget the shape $\al$ being given to us before the current block. Instead, we are going to encode (in the current block) the shape $\al$ whose norm we are computing.
% moreover, the shape $\al$ will in addition have vertices in $U_\al$  (or in $V_\al$) labeled in $[n]$. Concretely, we follow a $2$-step strategy, \begin{enumerate}
%	\item Drawing the shape: from the sequence of bundled wedges, draw out the shape $\al$ that is being encoded with vertices in $U_\al$ labeled in $[n]$;
%	\item Matching using automorphism: use automorphism factor to embed the wedges (concretely their legs) into $U_\al$.
%\end{enumerate}

 Let $\cost \calW$ be the cost of matching shape $\al$ from the sequence of bundled wedges, 
\begin{lemma} \label{lem: matching_cost}
	For some absolute constant $c_\matched>0$, we have\[ 
	\cost \calW \leq 2^{|V(\al)\setminus S|}\cdot  (c_\matched)^{|E(S)|}	\]
\end{lemma}

\paragraph{Bounding the drawing cost}
 As a reminder, throughout this section, we restrict our attention to maximal value labelings, as if not, we can hard-code its corresponding permutation-jump factor via a gap in the adjacent block-value. Given that we are looking at an optimal labeling,
  we first show that the first idea of identifying via vertices is sufficient when the wedges are far apart as the label for identifying a vertex for each wedge can be handled via edge-decay.
%   For the other case, the isomorphic wedges must be disconnected in $\al$, and automorphism factor from $|\aut(\al)|$ comes into our rescue.
\paragraph{Wedges are either disconnected or far-apart}
Let $\calL$ be an edge-labeling for a shape $\al$, let $\al_H(\calL)$ be the subshape obtained by removing $F$ edges (and isolated vertices connected by $F$ edges) in $\calL$, and recall that $\calW$ is a list of wedges (i.e. connected components connected to $U_\al$) in $\al_H$.

We first observe that wedges not connected in $\al_H$ may still be connected in $\al$, however, we can observe that they must be ''far-apart'' with each other in $\al$ to appear disconnected in $\al_H$.
\begin{proposition}
	For any wedge $W_i, W_j \in \calW$ for $\calW$ the list of wedges obtained from $\al_H(\calL)$ with some maximal-value labeling $\calL$, we have \[\dist_\al (W_i, W_j) > \Omega(\log_d n) \] 
\end{proposition}
 \begin{proof}
 	Suppose not and $W_i, W_j$ are connected in $\al$, by definition of $\al_H$, there is at least an all-$F$ path connecting $W_i$ and $W_j$. Considering the labeling obtained by flipping all $F$ path to all $H$ gives the desired logarithmic tradeoff.
 \end{proof}

\begin{proof}[Proof to Lemma ~\ref{lem: matching_cost}]
We can now prove our main lemma. Using the standard BFS encoding for $\al$, we can start from the sequence of bundled wedges, consider the BFS traversal on the wedges, and add edges if needed: in particular, \begin{enumerate}
	\item An edge leading to a new vertex can be succinctly encoded by a label in $[2]$, and another label in $[2]$ to specify the vertex-status of the new vertex (whether in $V_\al$);
	\item An edge leading to a visited vertex is more delicate as a label in $V(\al)|$ may be needed, however, via the above proposition, any such label can be charged to the long path;
	\item The bundling data requires $c_{\matched}$ factor from each edge that participates in the process, and this is only needed for $H$ edges, hence a factor of $(c_{\matched})^{|E(S)|}$.
\end{enumerate}
Therefore, this bounds the cost of drawing out $\al$ by $2^{|V(\al)\setminus S|}\cdot  (c_\matched)^{|E(S)|}$.\end{proof}

%\paragraph{Handling disconnected wedges}
	
%		More generally, the above argument applies to any wedge in $\calL$ once we consider an isolated vertex as a special wedge,
%	\begin{claim}
%	    Let $L(\calW)$ be the collection of isomorphic wedges of shape $\calW$ in $L$, and let $\aut(\al_W)$ be the automorphism of $\al$ restricted to vertices in $L(\calW)$, let $P_\calW$ be the collection of edges in $E(\al)\setminus E(L(W))$  that can be reached 
%	    from $L(\calW)$   without passing through any vertex in $U_\al\setminus L(\calW)$, then
%	    \[ 
%	    \frac{|\aut(L(\calW) |}{|\aut(\al_\calW) | } \leq c^{|P_\calW|}
%	    \] 
%	\end{claim}
%	Finally, to identify each leg of the wedge set from the edge set, we need to bound the gap between automorphism of each wedge $\calW$ in $L$ and the true autormorphism in $\al$. This is again done via edge decay of edges in $\calW$.
%	\begin{claim}
%		For a wedge $\calW$, let $\aut(\al_W)$ be the automorphism of$\al$ restricted to vertices in $\calW$ (i.e. holding vertices in $V_\al$ and $U_\al\setminus U_W$ fixed), then  \[
%		\frac{|\aut(W)| }{|\aut(\al_W)|} \leq c^{|E(\calW)|}
%		 \] 
%	\end{claim}
%\end{proof}

Combining the above gives us the desired Lemma~\ref{lem: matching_cost}.

\subsection{Wrapping up block-value for grouped graph matrix}
We are now ready to wrap up this section by our main lemma,
\begin{lemma} \label{lem:block-value-group-main-lemma}
For a fixed active profile $\calP$, for any sequence of shapes $\tau_{\calP} = \{\tau\}$ with active profiles prescribed by $\tau_P$, let $B$ be the block-value bound, we have \begin{align*}
    B(\sum_{\tau\in \tau_{\calP} } M_{\tau} ) \leq \max_{\substack{\tau\in \tau_{\calP} \\ S_{\tau}: \text{separator for }\tau  }} \cnorm^{|V(\tau)\setminus (U_\tau\cap V_\tau) |} \cdot \cnorm^{|E(S)|} \cdot \sqrt{\frac{n}{d}}^{|E(S_\tau)|} \cdot \sqrt{n}^{|V(\tau)\setminus S_\tau |} \cdot \dang(\tau\setminus S) \cdot \float(\tau) 
\end{align*}
for  some absolute constant $\cnorm$.
\end{lemma}
\begin{proof}
This follows by our drawing cost that for any permissible shape, the wedges can be bundled at a cost of $c_{matched}^{|E(S)| } $ since for wedge bundling process, it suffices for us to identify for each active profile its distance to the vertex in $S_R$ (the right boundary of the separator) that it can reach, and notice each vertex passed through the path before arriving at the destination at $S_R$ contributes $1$, hence the cost can be loosely bounded by $c_{matched}^{|V(\al)\setminus S_R| } $.

Once the shape is identified, the norm bound is then given by our block-value bound for a single shape from Corollary~\ref{cor:block-value-single-shape} , which is \[
\cnorm^{|V(\tau)\setminus (U_\tau\cap V_\al)|} \cdot \cnorm^{|E(S)|} \cdot  \sqrt{\frac{n}{d}}^{|E(S_\tau)|} \cdot \sqrt{n}^{|V(\tau)\setminus S_\tau |} \cdot \dang(\tau\setminus S) \cdot \float(\tau) \,.
 \]
 Combining the above yields the desired.
 \end{proof}
 
 \begin{remark}
 	This bound is potentially less meaningful if the sequence of grouped shapes contains shapes of different sizes,  as the max will simply be dominated by the largest one due to different scaling of $n$. However, it will be more transparent in our application to bounding  middle shapes and intersection shapes soon as we will account for the different sizes by combining the shapes with their pseudo-calibrated coefficients.
 \end{remark}

	\section{Moment matrix construction}
\label{sec:moment-matrix-construction}

\subsection{Trimming the graph}
Let $G$ be a random graph on $n$ vertices sampled from $G_{n,d/n}$, instead of working directly with the graph sample, we work with a nice subgraph of $G$ of $n'$ vertices that is obtained by removing ''not-too-many'' $o(n)$vertices, and show that SoS continues to think there is an independent set of size $o(\frac{n'}{\sqrt{d}})$, we can translate the lower bound for $\widetilde{G}$ to a lower bound for the untruncated graph $G$. Let's start with the definition of our ''trimmed'' subgraph.
	\begin{definition}[Trimmed subgraph] We call $\widetilde{G}$ the \emph{trimmed} subgraph of $G$ obtained by removing any ''ultra-high degree'' vertex, i.e. a vertex that  has degree at least $c_{deg}\cdot d$ in $G$. For this work, we set $c_{deg} = 10$.
%	\item (cycle-bad) it is part of a cycle of length at most $\cyclen= \log \log n$.
%\end{enumerate}
\end{definition}

% \paragraph{Coloring high-degree vertices}
% We now show that the subgraph restricted to high-degree vertices $G_H$ admit an actual $2$-coloring.
% %\begin{lemma}
% %	Let $G_H$ be the graph induced by the high-degree vertices, with probability at least $1-O(\frac{1}{n})$, it can be colored using $2$ colors.
% %\end{lemma}
% %\begin{proof}
% %We show in the following claims that both $G_H$ and $G_C$ can be $2$ colored, which yields our desired lemma.	
% %\end{proof}

% 	\begin{claim}
% 	[High-degree vertices are far apart] With high probability, let $G_H$ be the subgraph induced by degree-bad vertices in $G$, $G_H$ is 2-cycle free.\end{claim}
% 	We note that $G_H$ is essentially a random graph with average degree bounded away from $1$, hence a usual concentration argument implies that this is a cycle-free graph with high probability.
 
 By standard Poisson concentration, we have the following bound on the probability of a vertex begin degree-bad.
	\begin{proposition}
		Each vertex is degree-bad with probability at most $\exp(-c_{\deg}^2\cdot d)$.
	\end{proposition}

\paragraph{Gluing pseudo-distributions for the whole graph}	
Let $\pE_{\textsf{trimmed}} $ be a degree-$\dsos$ pseudo-expectation for independent set of size $k$ for $G_{\textsf{trimmed}}$ on $n'$ vertices, and let $v_B$ be an (integral) coloring assignment for $G_B$, we note that they can be merged into a degree-$\dsos$ pseudo-expectation $\pE_G$ for independent set of the whole graph $G$ on $n$ vertices.
	
\begin{definition}
	For $S=S_{G}\cup S_H\subseteq[n]$ where $S_H$ is the set of high degree vertices and $S_G$ the vertices in the trimmed graph,  we define 
  \[
	\pE_G[x_{S}] \coloneqq\pE_{\text{trimmed}}[x_S]
	 \]
  for $S\cap S_H = \emptyset$, and $0$ otherwise.
	 % where we extend $\pE$ from a pseudo-distribution on $[k]$ colors to $\pE'$ on $[k+2]$ colors by defining \begin{align*}
	 % 		 \pE'_{\text{trimmed}}[x_{S_G,\beta(S_G)}] \coloneqq \begin{cases}
	 % 	\pE_{\text{trimmed}}[x_{S_G,\beta(S_G)}] &(\beta(S_G)\subseteq [k])\\
	 % 	0 &\text{otherwise}
	 % \end{cases}
	 % 	 \end{align*} 
\end{definition}	
\begin{claim}
	$\pE_G$ is a valid independent-set pseudo-expectation for $G$ at degree $\dsos$ of value $k$. Moreover, $\pE_G$ gives an independence number of $G$ at least $ \frac{k}{n } \geq (1-o(1))\cdot \frac{k}{n'} $.
\end{claim}
\begin{proof}
The constraints for independent-set pseudo-expectation are immediate, and the value of the independent set size follows from $n-n'=o(n)$.
\end{proof}
With this reduction, the bulk of our work boils down into constructing $\pE_{\text{trimmed}}$ for the graph restricted to the nice vertices. From this point on, we will drop the subscript ''trimmed'' and refer to it simply by $\pE$. Let's now proceed to its construction.
\subsection{Constructing pseudo-expectation for the truncated subgraph}

\paragraph{Pseudocalibration for Independent Set}
Similar to previous works, we consider one of the most straightforward planted distribution $D_{pl}$ for any independent set of size $k$ in $G_{n,d/n}$: \begin{enumerate}[(1)]
    \item Sample a random graph $G\sim G_{n,d/n}$;
    \item Sample a subset $S \subseteq [n]$ by picking each vertex with probability $\frac{k}{n}$;
    \item Let $\Tilde{G}$ be $G$ with edges inside $S$ removed, and output $(S,\Tilde{G})$.
\end{enumerate}
It can then be computed that we have \begin{lemma}[Lem 4.4 in \cite{JPRTX}]
    Let $x_T(S)$ be the indicator function for $T$ being in the planted solution i.e. $T \subseteq S$. Then, for all $T \subseteq [n]$ and $\al \subseteq \binom{[n]}{2}$,
    \begin{align*}
        \E_{(S,\tilde{G}) \sim \calD_{pl}}[x_T(S)\cdot \chi_\alpha(\tilde{G})] = 
    \left(\frac{k}{n}\right)^{|V(\alpha)\cup T|}\left(-\frac{p}{1-p}\right)^{\frac{|E(\alpha)|}{2}}
\end{align*}
\end{lemma}
\paragraph{Non-dangling truncation: an upgrade from \cite{JPRTX}}
A major insight from \cite{JPRTX} in designing moments for sparse random graph is that one may employ the rule "connected truncation", and ignore any graph (or more precisely, "shape") that is not connected to $S$.

 On a high level, the connected truncation removes polynomial distinguisher based on subgraph-statistics, and a natural question is whether their that alone is sufficient here in the ultra-sparse regime should one ask for an independent set lower bound. Surprisingly, it is not! At least not within the current framework of analysis.

 We now recall the definition of dangling vertex from our prior discussion for tight matrix norm bounds. 
\begin{definition}[Dangling and non-dangling vertex]
	For a given shape $\al$, we call a vertex $v\in V(\al)$ a ''dangling'' vertex if it is not on any simple path from $U_\al$ to $V_\al$; otherwise, we call the vertex ''non-dangling''. Equivalently, any dangling vertex is a vertex of degree at most $1$ outside $S = U_\al\cup V_\al$.
\end{definition}

  It can be verified that without removing dangling shapes, some shape indeed has large norm, and break apart the charging strategy in their analysis. Therefore, removing such shapes turns out to be surprisingly necessary within our current framework.

		 \begin{definition}[Truncation of $\pE$] Let $T(S)$ be the set of graphs with labeled vertices $S\subseteq [n]$  s.t. for any graph $\al\in T(S)$,
		 \begin{enumerate}
		 	\item (Size-constraint) $|V(\al)\cup S|\leq D_V$;
		 	\item (Connectivity) Any vertex in $V(\al)$ is reachable from $S$;
		 	\item (Non-dangling) Any vertex in $V(\al)\setminus S$ has degree at least $2$;
		 \end{enumerate}
		 for $D_V =c_{trunc}\cdot  \dsos\log n $ for some absolute constant $c_{trunc}>0$.
		 \end{definition}
	% \begin{remark}
	% 	Though it is not explicitly removed from the truncation rule, we note non-dangling  that any shape that has degree-$1$ vertex dangling from $S$ automatically comes with coefficient $0$ via $k$-wise symmetry of the coloring coefficient.
	% \end{remark}
 \begin{remark}
     When it is clear, we also drop the dependence on $x_S$ and let $\calS$ be the set of shapes (or ribbons) that appear in our moment matrix $\pE_{trimmed}$.
 \end{remark}
%		And we define our moment matrix as 
%		\begin{align*}
%			\Lambda &= \sum_{\al_c:\al\in S} \left(\frac{1}{k}\right)^{|U_\al \cup V_\al|}\left(\eta\sqrt{\frac{d}{n}} \right)^{|E(\al)|} \calM_{\al_c}\\
%			&= \sum_{\al\in S} \lambda_\al M_{\al} \otimes H_{\al}
%		\end{align*}
%		where $M_\al$ is the usual graphical matrix with rows/columns indexed by subset of vertices (without color assignment)), and $H_\al$ as the matrix that fills in thee corresponding matrix with rows/columns indexed by both vertices and color assignment given the shape.

This lands us at the following definition of the moment matrix in the language of linear operator, as we will defer the matrix's language to the latter section after formally introducing ribbons,
	\begin{definition}[Pseudo-expectation for Independent-Set as a linear operator]
	% 	For any $S\subseteq [n]$ with $|S|\leq \dsos$,, we define its pseudo-expectation value	\begin{align*}
	% \pE[x_{S,\beta}](G) &= \sum_{\al\in T(S)} \E_{\sbm}\left[x_{S,\beta}\cdot \chi_\al\right]\cdot \chi_\al(G)\\
	% &=\sum_{\al \in T(S)} \sum_{\beta':V(\al)\rightarrow [k], \beta'_S=\beta} 
	% \left(\frac{1}{k}\right)^{|V(\al)|} \cval{E(\al)}_{\beta'}  \cdot \left(\eta\cdot \sqrt{\frac{p}{1-p}}\right)^{|E(\al)|} \cdot \chi_\al(G) 
	% \end{align*}	
 \begin{align*}
	\pE[x_{S}](G) &= \sum_{\al\in T(S)} \E_{D_{pl}}\left[x_{S}\cdot \chi_\al\right]\cdot \chi_\al(G)\\
	&=\sum_{\al \in T(S)} \left(\frac{k}{n}\right)^{|V(\al)\cup S| } \cdot \left(-\sqrt{\frac{p}{1-p}}\right)^{|E(\al)|}\cdot \chi_\al(G)
	\end{align*}	
	\end{definition}
% \paragraph{A strengthening of "connected truncation" from \cite{JPRTX}}
% In a prior work, \cite{JPRTX} employs "connected truncation" to design pseudo-moment matrix for independent-set problem in the sparse regime. On a high level, the connected truncation removes polynomial distinguisher based on subgraph-statistics, and a natural question is whether their that alone is sufficient here in the ultra-sparse regime should one ask for an independent set lower bound. Surprisingly, it is not! At least not within the current framework of analysis. It can be verified that without removing dangling shapes, some shape indeed has large norm, and break apart the charging strategy in their analysis. Therefore, even though when we work with $k$-coloring via the SBM model, these shapes get naturally truncated away due to their cancellation, this turns out to be surprisingly necessary as well with our current techniques. 

\begin{definition}[Setting objective value $k$]
    Throughout this work, we will pick $k = \frac{n}{\ceta \sqrt{d } \cdot \dsos^{4}}$ for some absolute constant $\ceta>1$.
\end{definition}

Formally, we define our candidate moment matrix for $\pE$ as the following.
\begin{definition}[Moment matrix $\widetilde{\Lambda}$]
    We define the moment matrix as \[
    \widetilde{\Lambda } \coloneqq  \sum_{\al \in \calS} \left(\frac{k}{n}  \right)^{|V(\al)|} \cdot \left(-\sqrt{\frac{p}{1-p}}\right)^{|E(\al)|} \cdot \frac{M_\al}{|\aut(\al)|} 
    \]
    where we remind the reader that $|\aut(\al)|$ is the size of automorphisnm group used to ensure each edge-set contributes once in the decomposition, and $\calS$ is the set of shapes such that \begin{enumerate}
        \item (Size-constraint) $|V(\al)|\leq D_V$;
        \item (Degree-constraint) $|U_\al\cup V_\al| \leq \dsos$;
		 	\item (Connectivity) Any vertex in $V(\al)$ is reachable from $S$;
		 	\item (Non-dangling) Any vertex in $V(\al)\setminus S$ has degree at least $2$.
    \end{enumerate}
\end{definition}
 
\subsection{Everything but PSDness for Independent Set }
Since our truncation is largely inspired by the "connected truncation", it continues to inherit several nice properties from connected truncation that allows us to readily verify the non-PSDness properties of $\pE$.
\begin{claim}[Independent-set constraint]
	For any $S\subseteq [n]$,, \[ 
	\pE[x_{S}] = 0
	\]
	if $S$ is not an independent set in $G$.
	\end{claim} 
	\begin{proof} This follows by grouping the subgraphs contributing to $\pE[x_S]$ according to the independent-set indicator in $S$ as \begin{align*}
	    \pE[x_S] &= \sum_{\al \in T(S) } \left(\frac{k}{n}\right)^{|V(\al)\cup S| } \cdot \left(-\sqrt{\frac{p}{1-p}}\right)^{|E(\al)|} \cdot \chi_\al(G)\\
     &=\sum_{\al \in T(S): E(\al)=\emptyset}\left(\frac{k}{n}\right)^{|V(\al)\cup S| } \cdot \left(-\sqrt{\frac{p}{1-p}}\right)^{|E(\al)|} \cdot \chi_\al(G) \cdot \left(\sum_{H \subseteq \binom{V(S)}{2}}   \sqrt{\frac{p}{1-p}} \cdot \chi_H(G)  \right)
	\end{align*}
		\begin{proposition}[Independent-set indicator]
			For $e\in \{0,1\}$, $1- \sqrt{\frac{p}{1-p}} \cdot \chi(e) = \frac{1}{1-p}(1-e)$; furthermore, for any set of edges $H\subseteq \binom{n}{2}$,\[ 
			\sum_{T\subseteq H}\left(- \sqrt{\frac{p}{1-p}} \right)^{|T|} \chi_T(G) = \frac{1}{(1-p)^{|H|}}\cdot \prod_{e\in H} (1-e)			\]
		\end{proposition}
Combining the above proves our claim.
 \end{proof}

\begin{claim}[Normalization]
    For any graph $G$, $\pE[1 ] =1$.
\end{claim}
\begin{proof}
    This follows immediately by connected truncation, that the only shape contributing to $S=\emptyset$, i.e., $\pE[1]$, is the emptyset, which comes with coefficient $1$.
\end{proof}
\begin{claim}
    The candidate moment matrix of $\widetilde{\Lambda}$ is SoS-symmetric.
\end{claim}
\begin{proof}
    This follows by observing for a fixed $S = U_\al \cup V_\al $, the coefficient does not depend on the partition of $S$ into $U_\al$ and $V_\al$. 
\end{proof}
	
\begin{claim}[Large independent set (Lemma 4.11 in \cite{JPRTX})] 
    With probability at least $1-o_n(1)$, we have \begin{align*}
        \pE[\sum_i x_i] \geq (1-o(1)) k
    \end{align*}
\end{claim}

	\section{PSDness Analysis from Norm Bounds}
As foreshadowed in our overview, this section follows the PSDness analysis strategy as \cite{JPRTX}, and we refer interested reader to it and \cite{potechin2023machinery} for a more thorough exposure to the approximate factorization machinery for proving high-degree Sum-of-Squares lower bounds.  In this work, we build upon the technical analysis from  \cite{JPRTX} to give a  more delicate combinatorial analysis to control the lower-order dependence of norm bounds unique to us. 

%\subsection{Factorization overview}
\subsection{Decomposition of the SoS moment matrix into the graph matrix basis}
With these technical definitions ready, we can now observe that the SoS moment matrix we define earlier can be rewritten as \begin{align*}
	\widetilde{\Lambda}
% \sum_{(R,\beta):\text{color-ribbon}: R\in \calS } \edgeco^{E(R)}\cdot \cval{E(R)}_\beta\cdot \left(\frac{1}{k}\right)^{V(R)}  \chi_R\\
	% &=\sum_{\text{ribbon } R\in \calS }   \edgeco^{E(R)}\cdot \sum_{\beta:V(R)\rightarrow [k]}\cval{E(R)}_\beta \cdot \left(\frac{1}{k}\right)^{V(R)} \cdot \cm_R\\
	&=\sum_{\text{shape } \al\in\calS}\frac{ \widetilde{\lambda}_\al \cdot \cm_\al }{|\aut(\al)|}
\end{align*}
where we define \[ 
\widetilde{\lambda}_\al \coloneqq \left(\frac{k}{n}\right)^{|V(\al )|}  \cdot  (-\sqrt{\frac{p}{1-p}})^{|E(\al)|}  
\]
and recall that $\calS$ is the set of ribbons/shapes under our truncation rule.
\begin{definition}[(Scaled) Shape coefficient]
	Given a shape $\al$, we define its shape-coefficient \[ 
	\lambda_\al = \left(\frac{k}{n} \right)^{|V(\al)| - \frac{|U_\al|+|V_\al|}{2} }\cdot \left(-\sqrt{\frac{p}{1-p}}\right)^{|E(\al)|}
	\]
\end{definition}
For convenience, it is convenient to work with a rescaling of the above matrix $\Lambda$ that can be obtained by left and right multiplying a diagonal matrix with diagonal being $\sqrt{k}^{|U_\al|}$, giving us the matrix \[ 
\Lambda = \sum_{\al\in\calS:\text{shape}} \frac{ \lambda_\al \cdot M_\alpha}{|\aut(\al)|} \,.
\]

\subsection{Recursive factorization overview: decomposition}
Recall that our goal is to show \[\Lambda = \sum_{\al \in \mathcal{S}} \lambda_{\al} M_\al \succeq 0   \]
Towards this goal, we will decompose each shape $\al$ into $\sigma\circ\tau\circ \sigma'^T$, for $\sigma, \sigma'^T \in \leftshape $ some left shape, and $\tau\in \midshape$ some middle shape based on \emph{Minimum Vertex Separator} (MVS), which corresponds to writing the matrix (modulo intersection terms) 
\[\cm_\al \approx \cm_\sigma  \cm_\tau \cm_{\sigma'^T} \]
Analogously, we will also factorize the coefficients of $\lambda_\al$, and with the decomposition in hand, we can simply write \[
\Lambda \approx (\sum_{\sigma\in \leftshape} \lambda_{\sigma} M_\sigma  ) (\sum_{\tau\in \midshape}  \lambda_{\tau} M_\tau ) (\sum_{\sigma'^T\in \leftshape} \lambda_{\sigma'^T} M_{\sigma'^T}  )  = LQL^T
\]
Our goal now reduces to showing $Q\succeq 0$, and this would allow us to restrict our attention to the middle shapes. Unfortunately, the above is an over-simplification as the equality only holds in an approximate sense, in which we hide out terms arising from ribbon intersection that warrants some extra care.

Consider ribbons of shape $\sigma, \tau, \sigma'$ intersect into some shape $\zeta$  we follow the usual strategy of recursive factorization and factor out $\gamma, \gamma'$ s.t. $\gamma$ is the leftmost MVS separating $U_\sigma$ and $V_\sigma \cup \int(P)$, and $\gamma'$ the rightmost MVS separating $V_{\sigma'^T}$ and $U_{\sigma'^T} \cup \int(P)$, i.e., we decompose \[ 
M_{\zeta} \approx M_{(\sigma-\gamma)} M_{\tau_P} M_{(\sigma'^T-\gamma'^T)}
\]
to further handle the intersection terms. The recursive factorization will create more intersection terms, however, the key is to notice the leftmost/rightmost MVS will move closer towards $U_\sigma$ and $V_{\sigma'^T}$ after each intersection, and hence this process would terminate. Another technical detail that we ignore so far is the truncation error, and we defer this to a later section. To summarize the above, our decomposition strategy can be captured by \begin{align*}
	\Lambda &= L(Q_0-Q_1+Q_2+\dots+Q_{\dsos})L^T \\
	&=(\sum_{\sigma\in \leftshape } \lambda_{\sigma} M_\sigma)(\sum_{\tau\in \midshape} \lambda_{\tau} M_\tau+ \sum_{\tau_P\in \int} \lambda_{\tau_P} M_{\tau_P}) (\sum_{\sigma'\in \leftshape} \lambda_{\sigma'^T} M_{\sigma'^T})\\ & + \text{Truncation} 
\end{align*} 

At this point, a natural next step is to show $Q = Q_0 - Q_1 +\dots + Q_{\dsos} \succeq 0$. To this end, we observe that there is a trivial identity component in $Q$ which may be pulled out and serve as our PSD mass, and the subsequent goal is to show \[ 
\sum_{\tau\in \midshape} \lambda_{\tau} M_\tau + \sum_{\tau_P\in \int} \lambda_{\tau_P}M_{\tau_P}
\] 
to be of small norm (after appropriate factoring out the kernel due to indicators).

Concretely, we follow the previous works and adopt a two step strategy, we first identify the matrix norm of each shape that arises in the above, and then combine it with the given coefficient and show that it can be charged to the appropriate diagonal with a charging argument. 

%
% \jnote{continue to change things into ind-set from here...add a main theorem for psdness and proof to that usinf the following lemmas}
\subsection{Main lemmas for PSDness}
\begin{restatable}[Independent-set indicator is PSD]{lemma}{colorindicatorpsd} \label{lem:color-indicator-psd} Consider $\Pi$ the diagonal matrix indexed by any subsets $S$ such that $|S| \leq \dsos$ with $\Pi[S,S] = 1[S \text{ is an Independent Set in } G]$, 
\[	\Pi \succeq 0 \]
\end{restatable}
\begin{proof}
    This is a diagonal matrix with non-negative entries.
\end{proof}
\begin{restatable}[Middle shape is bounded]{lemma}{midshapemain}
\label{lemma:midmain}
	\[ 
	Q_0 = \sum_{\tau \in \midshape} \lambda_\tau \cdot  M_\tau \preceq \frac{1}{10}\Pi
	\]
\end{restatable}
\begin{restatable}[Intersection is bounded]{lemma}{intersectionmain}
\label{lemma:intmain}
	\[ 
	 \sum_{0< i\leq \dsos} Q_i = \sum_i \sum_{\tau_P \in \int_i } \lambda_{\tau_P} \cdot  M_{\tau_P} \preceq \frac{1}{10}\Pi
	\]
	where $\int_i$ is the set of intersection shapes of level $i$.
\end{restatable}

\begin{restatable}[Truncation error]{lemma}{truncationmain}
\label{lemma:truncation-main}
\[
	\text{truncation error} \preceq n^{-\Omega(\dsos\log n )} \P\,.
\]
\end{restatable}
\begin{restatable}[Left-shape is well conditioned]{lemma}{singvalmain}
\label{lemma:singval-main}
\[
\left(\sum_{\sigma \in \leftshape} \lambda_\sigma M_\sigma\right)\left(\sum_{\sigma \in \leftshape} \lambda_\sigma M_\sigma\right)^T \succeq n^{-O(\dsos)} \Pi \,.
	 \]
\end{restatable}

Assuming the above lemmas, we are able to complete the proof to our main theorem by filling in the missing piece about PSDness from previous section,
\begin{proof}[Proof of main theorem]
Following our recursive factorization, we have 
\begin{align*}
    \Lambda &= L(Q_0-Q_1+Q_2+\dots+Q_{\dsos})L^T \\
    &= L\cdot  \Pi^{1/2}(Id + Q_0 + \sum_i (-1)^i Q_i) \Pi^{1/2} \cdot L^T +\text{truncation}\\
    &\succeq \Omega(1) \left(\sum_{\sigma \in \leftshape} \lambda_\sigma M_\sigma\right)\left(\sum_{\sigma \in \leftshape} \lambda_\sigma M_\sigma\right)^T  +\text{truncation}\\
    &\succeq (\Omega(1) - n^{-O(\dsos)} )\Pi\\
    &\succeq 0
\end{align*}
where the first psd inequality we plug in the norm bound for the middle shape and intersection term, and then apply the bounds from singular value lower bounds and truncation error bound.
 \end{proof}

%\section{Alternate factorization strategy}
% Let's start by considering $\eta<0$ regime (where we have $|\eta\cdot k|\leq 1$), notice in this regime, DMVS works well because of the magnitude condition (also see ind set). The switch edge case can be handled by using the planted clique strategy that ignoring edges inside DMVS ($\tau$) does not affect $\sigma, \sigma'$ gluing with $\tau$ (i.e. regardless of edges in $E(\tau)$, $\tau$ is always a DMVS w.r.t. $\sigma,\sigma'$, just the final graph might be different). The worrisome case coming from switch edges can be handled by charging to the shape $\tau$ with one fewer switch edges (or just throw all those out). The crucial observation there is any edge that might impose a $k$-blow-up because we are not using the SMVS decomposition must be a separator edge, in which case we have a $\eta$ left-over to cancel out with the blow-up.
% 
% 
% In the $\eta>0$ regime, the above strategy no longer works as we don't have magnitude control across $\eta$ and $k$. This is where SMVS comes into play. For a given color-shape $(\al,c)$, we are going to case on whether its SMVS contains switch edges; if the SMVS contains an odd number of switch edges, we decompose the color-shape using DMVS; otherwise, we decompose it using SMVS.
%
%Couple issues to make clear of (that i'm aware of)\begin{enumerate}
%	\item The $\eta>0$ regime, need to make sure there is no single shape being produced in both factorization (need to impose extra condition on $\sigma, \sigma'$ there)
%\end{enumerate}

\subsection{Bounding the middle shapes}
\label{subsec:midshape}
In this subsection, we show we can charge the middle shape $\midshape $ the set of canonical middle shapes with the corresponding diagonal. We restate the main lemma of the section below,
%\begin{restatable}[$Q\succeq 0$]{lemma}{midshapemain}
%\label{lemma:midshape-main}
%	\[ 
%	Q = \sum_{\tau \in \midshape} \lambda_\tau \cdot  \cm_\tau \succeq \biggl(1-o(1)\biggr)\Pi_{\mathsf{color}}
%	\]
%\end{restatable}

We first identify what a middle shape is that arises in the approximate factorization machinery based on \emph{minimum vertex separator}.
\begin{definition}[Middle shape]
    A shape $\al$ is a middle shape if $U_\al$ and $V_\al$ are both minimum vertex separators for $\al$.
\end{definition}

\paragraph{Overall strategy for middle shape}
Recall that the ultimate goal of this section is to show \[ 	Q_0 = \sum_{\tau \in \midshape} \lambda_\tau \cdot  M_\tau \preceq \frac{1}{10}
\] (ignoring the kernel from independent set indicator),
we  factorize out edges in $E(U\tau \cap V_\tau )$ as they form a corresponding independent-set indicator in our context, captured by a diagonal matrix $\Pi$, landing us at 
\begin{align*}
		\sum_{\tau\in\midshape} \lambda_\tau M_\tau  &= \Pi^{1/2} \cdot \bigg(\sum_{\substack{\tau:\text{middle shape}\\ E(U_{\tau}\cap V_\tau ) = \emptyset }} \lambda_\tau \cdot M_\tau \bigg) \Pi^{1/2} \\&=
		\Pi^{1/2} \cdot \bigg(\Id + \sum_{\substack{\tau:\text{ middle shape}\\ E(U_{\tau} \cap V_\tau) = \emptyset\\\text{non-trivial} }} \lambda_\tau \cdot M_\tau \bigg) \Pi^{1/2}
\end{align*}

%
%\begin{align*}
%	\sum_{\tau\in\midshape} \lambda_\tau \cm_\tau  &=  \Pi_{\text{color}} \cdot \biggl( F \cdot \biggl(\sum_{\substack{\tau_c\\ \text{cushioned middle shape} \\ E(U_\tau)  = \emptyset } } \lambda_{\tau_c} \cm_{\tau_c} + \sum_{j>0} (-1)^j \sum_{\substack{\tau_P\\ \tau_c\circ\tau_P \in \int(\tau_f,\tau_c) } } \lambda_{\tau_P}\cm_{\tau_P}  \biggr)  + \mathsf{Truncation}\biggr) \\
%	&= \Pi_\col \cdot \biggl( F \cdot  \biggl(\Id + \sum_{\substack{\tau_c\\ \text{cushioned middle shape}\\ \text{non-trivial}\\ E(U_\tau) = \emptyset }} \lambda_{\tau_c}\cdot \cm_{\tau_c} + \sum_{j>0} (-1)^j \sum_{\substack{\tau_P\\ \tau_P \in \int(\tau_f,\tau_c)  \\ E(U_\tau)=\emptyset} } \lambda_{\tau_P}\cdot \cm_{\tau_P}  \biggr) + \mathsf{Truncation}\biggr)  \end{align*}
At this point, the PSDness of the above matrix boils down to the following lemmas,
%\begin{restatable}[Well-conditioned floating matrix]{lemma}{wellconditionedfloating} \label{lem:well-conditioned-floating}
%\[	F \succeq \Id \]
%\end{restatable}

%\colorindicatorpsd*
\begin{restatable}[Charging non-trivial middle shape]{lemma}{midcushionedcharging}\label{lem:charging-cushioned-mid}
	\[ 
	\biggl\|\sum_{\substack{\tau\\ \text{middle shape}\\ \text{non-trivial}\\ E(U_\tau\cap V_\tau) = \emptyset }} \lambda_{\tau_c}\cdot M_{\tau_c} \biggr\| \leq \frac{1}{10}
	\]
\end{restatable}
%
%\begin{restatable}[Charging floating intersection]{lemma}{floatintersection} \label{lem: float-intersection}
%\[
%  \biggl\|\sum_{j>0} (-1)^j \sum_{\substack{\tau_P\\ \tau_P \in \int(\tau_f,\tau_c)  \\ E(U_\tau)=\emptyset} } \lambda_{\tau_P}\cdot \cm_{\tau_P}  \biggr\| \leq o(1)
% \]
%\end{restatable}
%
%

 % \midcushionedcharging*
 
\paragraph{Overview of middle shape charging}
\begin{enumerate}
	\item For each middle shape, it can be charged to the appropriate diagonal (via \cref{sec:midshape-norm-bound});
	\item For each diagonal, there are not ''too-many'' middle shapes charged to it. A naive triangle inequality and union bound does not work well here, and we appeal to the grouping via active-profile of shapes(via \cref{subsec:grouping} ). 
\end{enumerate}

We first give an argument for charging a particular middle shape with certain gap, and then combine it with our counting of middle shapes.
\subsection{Charging a single middle shape} \label{sec:midshape-norm-bound}
 The heart of this section is to show 
 \begin{lemma}
 	For any middle shape $\tau$, \[ 
 	\lambda_\tau \cdot \|M_\tau \| \leq (\ceta)^{|V(\tau)- \frac{|U_\tau|+|V_\tau|}{2}  } \cdot (\frac{1}{\dsos^2})^{|V(\tau) - \frac{|U_\tau|+|V_\tau|}{2} |} 	\]
 \end{lemma}
 where we defer the specific $o(1)$ slack for counting at the moment. Recall 
  from our graph matrix norm bounds, 
 \begin{proposition}
 	For each shape $\al$, it has norm bound \[ 
 	\| M_\al \| \leq \cnorm^{|V(\tau)|}\cdot  \max_{S: \text{separator} } \sqrt{n}^{|V(\tau)\setminus S|} \cdot \left( \sqrt{\frac{1-p}{p}} \right)^{|E(S)|} \cdot \float(\al) \cdot \dang(\al) \]
 \end{proposition}
 Furthermore, recall that our middle shape coefficient is given by \[
 \lambda_\tau = \left(\frac{k}{n}\right)^{|V(\tau)| - \frac{|U_\tau|+|V_\tau|}{2} }\cdot \edgeco^{E(\tau)}
  \]
  Combining the above gives us \begin{align*}
  	\lambda_\tau \cdot \|M_\tau\|  &\leq \cnorm^{|V(\tau)|}\cdot (\frac{k}{n})^{|V(\tau)| -|U_\tau|} \max_{S: \text{separator} } \cdot \sqrt{n}^{|V(\al)\setminus S|} \cdot \left( \sqrt{\frac{1-p}{p}} \right)^{|E(S)|} \cdot \edgeco^{E(\tau)}\\&\cdot  \float(\al) \cdot \dang(\al)
  \end{align*}
 \begin{proof}
 	Momentarily, let's ignore the dependence of $\dsos$ factor: the core of the argument relies upon the observation that we are able to assign a $\frac{1}{k}$ factor for each vertex outside the separator from the shape coefficient via a decomposition into left/middle/right part according to MVS,
 	\begin{enumerate}
 		\item We have a total vertex coefficient $\left(\frac{k}{n}\right)^{|V(\tau)|-|U_\tau|}$, and $U_\tau$ is by construction the MVS of $\tau$ while $S$ is a separator, so $|S|\geq |U_\tau|$ and hence we have $\left(\frac{k}{n}\right)^{|V(\tau)\setminus S|}\leq  \left(\frac{k}{n}\right)^{|V(\tau)\setminus U_\tau|} $;
 		\item By connectivity of middle shape, we can consider a BFS from $S$, and we note the $\frac{k}{n}$ vertex coefficient shown from above is sufficient for charging the factors on vertices outside $S$;
 		\item For each vertex outside $S$, it is reachable from $S$, and we can consider the edge that explores the vertex in the BFS process, combining the factor of that particular edge with the vertex factor of the vertex gives, \[
 		\left| \sqrt{n} \cdot \edgeco\cdot \frac{k}{n} \right| \leq \ceta
 		 \]
 		 where we recall our hardness assumption  and $\frac{\sqrt{d}}{k}\leq \ceta$ (ignoring the extra decay from $\frac{1}{\dsos}$ factors which are reserved for counting shapes) ;
    \item Since each vertex outside the separator gives us a decay of $\ceta$, and observe that by middle shape property, $|V(\tau)\setminus U_\tau\cap V_\tau | \geq \frac{1}{2} |V(\tau)\setminus S|  $ as both $U_\tau$ and $V_\tau$ are the MVS, thus setting adjusting $\ceta$ to appropriate constant gives us \[ 
    \cnorm^{|V(\tau)|  } \cdot \frac{1}{\ceta^{|V(\al)\setminus S|}} \leq (\frac{1}{10})^{V(\tau)\setminus U_\tau\cap V_\tau |}
    \]
 		 \item For each edge leading to an explored vertex outside $S$, we pick up an edge factor of \[
 		\left|  \edgeco \right|= O\left(\frac{1}{\sqrt{n}}\right)
 		  \]
 		  \item For floating factor, we note that there is no floating component due to connected truncation;
 		  \item For dangling factor, notice we pick up potentially one such factor for any dangling branch outside the separator,  however, since each vertex outside $U$ and $V$ is of degree at least two, for each dangling branch, there is at least one edge unused for vertex connectivity to the separator in the prior charging, and hence we have a gap of at most \[ 
 		 \left| \sqrt{\frac{p}{1-p}} \cdot (2|V(\al)|q_\al)^{O(1)}\right| \ll 1\,;
 		  \]	  \item Finally, for edges inside the SMVS $S$, we have \[ 
 		  \left|(k-1)\cdot \edgeco \cdot \sqrt{\frac{1-p}{p}} \right| = 1 
 		  \]
 	\end{enumerate}
 \end{proof}
 \subsection{Grouping middle shapes}\label{subsec:grouping}
 To complete a bound for middle shape, we still need a count of them. Departing from the prior works that apply a union bounds over various shapes which leads to inevitable loss of $\poly\log$ factors, we circumvent this issue by grouping middle shapes and apply trace-method on the gruped martrix.

 \paragraph{Grouping shapes into permissible shapes}
 % \jnote{take care of this tmr morning}
 \begin{definition}[Permissible shapes]
		We call a shape permissible if it can be obtained from the following process,\begin{enumerate}
			\item Start from any shape $\al$ such that any vertex not in $U_\al\cup V_\al$ has degree at least $2$, and remove any edge between $V(\al)\setminus U_\al\cup V_\al$ and $U_\al\cap V_\al$ ;
			\item For any dangling vertex, put in an edge from the dangling vertex to the lowest-ordered  (in $\dsos$) vertex in $U_\al\cap V_\al$;
			\item For any floating component, it must have at least two edges that get removed, and put back the the two edges connecting to the lowest ordered vertices in $U_\al\cap V_\al$.
		\end{enumerate}
	\end{definition}
\begin{lemma}
	For any set of potential edges $E$ on vertices $V(E)$, \[ 
 \left| \sum_{E'\subseteq E} \prod_{e\in E'}\left(-\sqrt{\frac{p}{1-p}} \chi_E \right) \right| \leq (1+ \frac{p}{1-p} )^{|E|}\,.
%  \cdot (1+\frac{p}{1-p})^{|E|-|E_\yes| }
	\]  
%	where $E_\yes$ is the set of edges in $E$ that appear in the random graph.
\end{lemma}
\begin{proof}
%We first show a weaker bound of $(1+\frac{1}{k-1})^{|E|}$,
	\begin{align*}
		\left| \sum_{E'\subseteq E} \prod_{e\in E'}\left(-\sqrt{\frac{p}{1-p}} \chi_e \right) \right|  &= \prod_{e\in E'} \left| 1- \sqrt{\frac{p}{1-p}} \chi_e   \right|\\
		&\leq (1+\frac{1}{k-1})^{|E|}
	\end{align*}
	where the second inequality follows by casing on the edge present-status ,\begin{enumerate}
		\item If the edge is present, \[ 
		1- \sqrt{\frac{p}{1-p}} \chi_e   = 1-1  = 0 \,;
		\]
		\item If the edge is missing,\[ 
		1- \sqrt{\frac{p}{1-p}} \chi_e    = 1 + \frac{p}{1-p}\,;
		\]
		% \item If this is a cross-color edge, and edge is present, \[ 
		% 1+ \sqrt{\frac{p}{1-p}} \cdot  \eta \cdot \cval{e}_\beta\cdot \chi_e  = 1+  (-\frac{1}{k-1}) \cdot (-1) = 1+\frac{1}{k-1} \,;\]
		% \item If this is a cross-color edge, and edge is missing, \[
		% 1+ \sqrt{\frac{p}{1-p}} \cdot  \eta \cdot \cval{e}_\beta\cdot \chi_e  = 1+  (-\frac{1}{k-1}) \cdot (k-1) = 1-\frac{p}{(1-p)(k-1)} \,; \]
		% \item Taking the max among above cases gives us a factor of at most $(1+\frac{1}{k-1})$ from each potential edge.
	\end{enumerate}
%	To improve from this bound, notice that we can case on the presence of each potential edge, and for each yes-edge, we pick up a factor of at most $(1+\frac{1}{k-1})$; on the other hand, for the no-edge, we pick up a factor of at most $1+\frac{p}{1-p}$, thus the above is also bounded by \[ 
%	(1+ \frac{1}{k-1})^{|E_\yes|} \cdot (1+\frac{p}{1-p})^{|E|-|E_\yes| }
%	\]
\end{proof}

 \begin{lemma}[Bounding the blow-up from hidden potential-edges]
 	For each permissible shape $\tau$, let $E$ be the set of edges that can be added to $\tau$ such that $\tau$ remains a valid middle shape, \[ 
 	\left| \sum_{E'\subseteq E} \prod_{e\in E'}\left(-\sqrt{\frac{p}{1-p}} \chi_e  \right)\right| \leq 1+o_n(1)
 	\] \end{lemma}
 \begin{proof}
 	Notice each such edge has to go from a vertex outside of $U_\tau\cap V_\tau$ to $U_\tau\cap V_\tau$, and this is at most $\dsos$ choices for each edge outside, therefore $|E|\leq \dsos \cdot |V(\tau)\setminus U_\tau\cap V_\tau|$. Combining with the bound from previous, we have \begin{align*}
 		\left| \sum_{E'\subseteq E} \prod_{e\in E'}\left(-\sqrt{\frac{p}{1-p}} \chi_e  \right)\right|&\leq (1+ \frac{p}{1-p})^{|E|}\\&\leq   1+ \frac{  p}{1-p} \dsos |V(\tau)\setminus U_\tau\cap V_\tau|\\&\leq 1+o_n(1)
 		 	\end{align*} 
     where the last bound follows from our choice of $\dsos = o(n)$.
 \end{proof}
 
 % \begin{remark}
 % 	We will offset the blow-up here by the $\dsos$ decay from each vertex outside the intersection.
 % \end{remark}
  The above prompts to define the following hidden edge indicator for any permissible shape,
 \begin{definition}[Hidden edge indicator]
 	For a given permissible shape $\tau$, let $E_{h(\tau) }$ be the set of hidden edges that can be added between $V(\tau)\setminus (U_\tau\cup V_\tau) $ and $U_\tau\cap V_\tau$ that does not violate the permissibility of $\tau$ (due to ordering of edge-removing), and we define \[ 
 	q(\tau) \coloneqq \sum_{E'\subseteq E_{h(\tau)}} \prod_{e\in E'}\left(-\sqrt{\frac{p}{1-p}}  \chi_e  \right)
 	\]
 \end{definition}
 Furthermore, for our counting scheme to apply, it is also convenient for us to group the shapes according to which vertices in $U_\tau$ and in $V_\tau$ are incident to some vertices outside. Towards this end, we appeal to active-profile defined as the following,
\begin{definition}[Active-profile]
	For any (permissible) middle shape $\tau$, we call the set of indices $\calP = \{U_{active}, V_{active} \}\subseteq [\dsos]^2$ an active profile for $\tau$ if 
		any vertex $v\in U_{active}, V_{active}$ is incident to some vertex outside $U_\tau\cap V_\tau$.
\end{definition}

\begin{lemma}[Bounding active-profiles] 
	For each diagonal $U$, each active profile that gets charged to it can be identified at a cost of $\dsos^{|U_{active}| + |V_{active} }$.
\end{lemma}
\begin{definition}[Graph matrix with prescribed active-profile] 
	Given an active profile $\calP$, we define the corresponding graph matrix for the given profile as \[ 
	M_{\calP} = \sum_{\tau: \text{shapes with prescribed active-profile}} M_{\tau}
	\]
\end{definition}

\begin{proof}
	Since active profile is a collection of middle shapes, by middle shape property, we have $|U_{active}| = | V_{active}|$. It then suffices for us to identify a label in $\dsos$ for each index in active-profile.
\end{proof}

\subsection{Wrapping up middle shape bound}
We are now ready to apply our matrix norm bounds on each active-profile. 
% Notice that for a single active-profile $\calP$, apply our graph matrix bound on
% \begin{align*}
% 	\lambda_{\cal_P} \calM_{\calP} \coloneqq  \sum_{\substack{\tau: \text{middle shapes with active profile} \calP  \\ \text{non-trivial, permissible}  \\ E(U_\tau\cap V_\tau)=\emptyset }  }  \lambda_\tau \cdot    M_\tau \cdot q(\tau) 
% 	\end{align*}
% where we let $\beta$ be the coloring of any single shape $\tau$, and observe that \begin{align*}
% 	&\sum_{\substack{\tau: \text{middle shapes with active profile} \calP  \\ \text{non-trivial  }  }} \lambda_\tau \sum_{\beta:V(\tau)\rightarrow [k]} \cm_\tau \cval{E(\tau)}_\beta \cdot \\&  = \sum_{\substack{\tau: \text{middle shapes with active profile} \calP  \\ \text{non-trivial, permissible} \\  E(U_\tau\cap V_\tau)=\emptyset  }  } \sum_{\beta:V(\beta)\rightarrow [k]}  \lambda_\tau \cdot  \cval{E(\tau)}_\beta \cdot  \cm_\tau \cdot q(\tau) 
% \end{align*}
\begin{lemma}[Grouped middle shape bound]\label{lemma:grouped-middle-shape}
    \[ 
    \sum_{\calP:\text{active profiles}: E(U_\tau)\cap E(V_\tau) = \emptyset}\| \lambda_{\calP}M_{\calP} \| \leq \frac{1}{10}    \]
\end{lemma}
\begin{proof}

Applying our ''grouped'' norm bound on $\lambda_{\calP} \cm_{\calP}$ for active-profile $\calP$ gives us \begin{align*}
	\|\lambda_{\calP} \calM_{\calP}\|&\leq \max_{\tau\in \tau_P \text{non-trivial, permissible} } \|  \lambda_{\tau}  \cdot M_\tau \cdot q(\tau) \| \\
	&\leq \max_{\tau\in \tau_P \text{non-trivial, permissible} } \|  \lambda_{\tau}  \cdot M_\tau\| \cdot |q(\tau)|
\end{align*}
where we observe the following,\begin{enumerate}
%	\item The function $c(\calP)$ measures the identification cost for a fixed 
	\item The coefficient is given by \[\lambda_\tau =  (\frac{k}{n})^{|V(\tau)|-\frac{|U_\tau|+|V_\tau|}{2}} \cdot \left(-\sqrt{\frac{p}{1-p}}   \right)^{|E(\tau)|}  \,;\]
	\item The norm bound, via the block-value bound, is given by \[
 \max_{\tau, S} \cnorm^{|V(\tau)| } \cdot \sqrt{n}^{|V(\tau)\setminus S|}\left(\sqrt{\frac{n}{d}}\right)^{|E(S)|}\sqrt{n}^{|I(\al)|}
\cdot \dang(\al\setminus S ) \cdot \float(\al) \,;
	 \]
	 \item The hidden edge indicator has magnitude bounded by \[ 
	 \|q(\tau)\|\leq  (1+o_n(1))^{|V(\tau)\setminus U_\tau\cup V_\tau|}\,;
	 \]
\end{enumerate}
We now mimic our charging strategy for a single middle shape, and note that the factor other than $\dsos$ dependence follows from the charging for a single shape, and we make clear the $\dsos$ dependence here as,\begin{enumerate}
	\item Assuming $\frac{k}{n} = O(\frac{1}{c_\eta \sqrt{d}\cdot \dsos^2 })$;
	\item Each vertex outside $U_\tau\cup V_\tau$ contributes a single factor of $\dsos$ via the blow-up from hidden-edge indicator, while each comes with a full factor of decay $\frac{1}{\dsos^2}$;
	\item Each vertex in $U_\tau \Delta V_\tau$ does not contribute any $\dsos$ factor from hidden-edge indicator, and each comes with a factor of $\frac{1}{\dsos}$ since each comes with half a vertex coefficient;
	\item By our sparsity bound, $|E(S)|\leq 5|V(S)|$, and moreover, notice since we look at permissible shape, there are no edges inside $E(U_\tau\cap V_\tau)$; and the edges between $U_\tau\cap V_\tau$ and $V(\tau)\setminus (U_\tau\cap V_\tau )$ are at most $2\cdot |V(\tau)\setminus (U_\tau\cap V_\tau )|$, and therefore, we have $|E(S)|\leq 7|V(\tau)\setminus U_\tau\cap V_\tau| \leq c^{|V(\tau)\setminus U_\tau\cap V_\tau|}$ for some constant $c$ subsumed by $c_\eta$ factor of vertex-decay;
	\item Combining the above gives us a factor of $\left(\frac{1}{\dsos}\right)^{|U_\tau\Delta V_\tau|} = \left(\frac{1}{\dsos}\right)^{|U_{active}(\calP) +|V_{active}(\calP)|   }$;
	\end{enumerate}
	Summing over all active-profiles gives us \begin{align*}
		\sum_{\calP} \|\lambda_{\calP}\calM_{\calP} \| \leq \max_{\calP}{ c(\calP)\cdot  \left(\frac{1}{\dsos}\right)^{|U_{active}(\calP) +|V_{active}(\calP)|   } }  \ll \frac{1}{10}
	\end{align*}
	where we recall that the identification cost of $\calP$ is bounded by $\dsos^{|U_{active}(\calP)| + |V_{active}(\calP)| }$.
\end{proof}

 \begin{remark}\label{remark:extra-gap-from-middle-shape}
     Notice the above argument carries through for $\frac{k}{n} = O(\frac{1}{\sqrt{d}\cdot \dsos^2 } )$ while we in fact have a larger decay of $\frac{1}{\sqrt{d\cdot \dsos^4}}$. This is reserved for the extra $\dsos$-factor needed for each vertex in intersection shape.
 \end{remark}
 This completes our proof to Lemma~\ref{lemma:midmain} by noting that putting edges from $E(U_\tau \cap V_\tau)$ back into anuy shape for any active-profile groups out the independent set indicator $\Pi$ again.
\subsection{Bounding intersection terms}
\label{subsec:intersection-term}
\paragraph{Preliminaries for intersection term}
\begin{definition}[Ribbon composition]
	Given ribbons $R_1, R_2$, we call them \emph{composable} if $V_{R_1}=U_{R_2}$. For ribbons with boundary set indexed by subgraphs, we additionally constrain $E(V_{R_1}) = E(U_{R_2})$. We use $R_1\circ R_2$ to represent the ribbon from the composition of $R_1$ and $R_2$.
\end{definition}
\begin{definition}[Proper composition]
	Given $R_1,R_2$ composable ribbons, we call them a proper composition if there is no surprise intersection beyond the boundary, i.e., $(V(R_1)\setminus V_{R_1})\cap (V(R_2) \setminus U_{R_2})=\emptyset$. 
\end{definition}
\begin{remark}
	For properly composable ribbons $R_1,R_2$, we have $M_{R_1}M_{R_2} = M_{R_1\circ R_2}$.
\end{remark}
\begin{definition}[Improper shape and phantom graph]
	We call a shape improper if it contains multi-vertices (repeated vertices) or multi-edges. The underlying multigraph of an improper shape is also called a \emph{phantom graph}, and we usually use $\tau_P$ to refer to both the improper shape and its underlying multigraph.
\end{definition}
\begin{definition}[Shape composition]
	Given shapes $\al, \beta$, we call them \emph{composable} if $V_\al = U_\beta$. For composable shapes, we write $\zeta = \al\circ \beta$ for $\zeta$ any (possibly improper) shape whose multigraph can be obtained by composing ribbons of shape $\al$ and $\beta$. 
\end{definition}
Unfortunately, it is no longer true that $M_\al M_\beta = M_{\al\circ \beta}$ as ribbon injectivity dose not hold beyond the anticipated boundary, and collision outside the boundary would inevitably produce improper shapes that are the main subject of this section.

\paragraph{Factorizing intersection terms}
We now make concrete our charging strategy for the intersection terms. Recall from our big picture of $LQL^T$ decomposition, intersection can happen whenever shapes from $L, Q, L^T$ intersect with one and another. For a left, middle, right shape $\sigma, \tau, \sigma'$ that intersect to some improper shape $\tau_P$, we follow the recursive factorization framework from planted clique where we factorize the left/right shape from intersection term $\lambda_{\tau_P}M_{\tau_P}$ and attempt to charge to the corresponding diagonal of the new left/right shape.  
%	 With this big picture in mind, we also want to note that in previous works, $\gamma, \gamma'$ are just defined to be any left/right shape; however, this needs to be modified for SMVS case since left shape has one side indexed by vertices, and another indexed by subgraph, while $\gamma$ needs to have both sides indexed by subgraph.

%where $U_\gamma$ is the MVS separating $U_\sigma$ from $V_\sigma\cup \int(\tau_P)$, and $V_{\gamma'}$ is the MVS separating $V_{\sigma'}$ from $U_{\sigma'}\cup \int(\tau_P)$ where $\int(\tau_P)$ is the set of intersected vertices. 
With this big picture in mind, we are ready for more technical definitions.
\begin{definition}[Intersected vertices $\int(\tau_P)$]
	Given a composable pair of $\sigma, \tau,\sigma'$ that intersect to some improper shape $\tau_P$, let $\int(\tau_P)$ be the set of vertices that gets intersected in $\tau_P$, i.e., the vertices that appear with multiplicity greater than $1$ in the associated multigraph of $\tau_P$.
\end{definition}
%we will upgrade $U_\gamma$ from the (dense) MVS to $(U_\gamma, E(U_\gamma)$ being the corresponding SMVS (and ultimately PMVS when combined with the coloring); furthermore, we notice that the dependence of $\gamma$ on $\tau_P$ can be broken once we specify the set of vertices in $U_\gamma$ that get intersected in $\tau_P$, and this motivates our following definition for left/right intersection color-shape. 
		 \begin{definition}[Left-intersection shape $\gamma$]
	 A shape $\gam$ is a left-intersection color-shape if \begin{enumerate}
	 	\item $\gam$ comes with an an additional label for each vertex in $V(\gam)$ specifying whether this vertex is in \emph{to be intersected}, i.e., whether the vertex is in $\int(\gamma)$
%	 	\item $(U_\gam, \cp_{U_\gam))$ is a positive glyph;
	 	\item $V_\gam.$ is the unique MVS for separating $V_\gam$ and $U_\gam$;
	 	\item $U_\gam $ is a MVS that separates $V_\gam\cup \int(\gam)$ and $U_\gam$;
	 	\item We follow the convention where we put edges in $E(V_\gam)$ to the middle shape $E(U_\tau)$, hence, $E(V_\gamma) = \emptyset$.
	 		 \end{enumerate}
	 		 Consequently, we define the corresponding shape coefficient for $\gam$ to be \[ 
	 		 \lambda_\gamma \coloneqq  \left(\frac{k}{n} \right)^{|V(\gam)|-\frac{|U_\gam|+|V_\gam|}{2} }	 \cdot \edgeco^{|E(\gam)| } \]
	 \end{definition}
	 \begin{remark}
	 	For our application, we will further factorize out the edges in $E(U_\gam\cap V_{\gam} ) $ as they form the independent-set indicator, hence we actually have $E(U_\gamma \cap V_\gam) = \emptyset$ (and similarly for $\gp$) in our analysis. 
	 \end{remark}
	 Analogously, we define the right-intersection color shape as the following. \begin{definition}[Right-intersection shape $\gamma'$]
	 A shape	$\gamma'$ is a left-intersection color-shape if \begin{enumerate}
	 	\item $\gam'$ comes with an an additional label for each vertex in $V(\gam')$ specifying whether this vertex is in \emph{to be intersected}, i.e., whether the vertex is in $\int(\gamma')$
%	 	\item $V_\gam.$ is the unique MVS for separating $V_\gam$ and $U_\gam$;
	 	\item $V_\gp $ is a MVS that separates $V_\gp\cup \int(\gp)$ and $U_\gp$;
	 	\item We follow the convention where we put edges in $E(U_\gp)$ to the middle shape $E(V_\tau)$, hence, $E(U_\gp) = \emptyset$.	 		 \end{enumerate}
	 		 Consequently, we define the corresponding shape coefficient for $\gp$ to be \[ 
	 		 \lambda_{\gp} =  \left( \frac{k}{n} \right)^{|V(\gp)|-\frac{|U_\gp|+|V_\gp|}{2} }\cdot \edgeco^{|E(\gp)|}
	 		 \]
	 \end{definition}
	\begin{definition}[Intersectable tuple $(\gam,\tau,\gam')$]
		We call a composable tuple $(\gam,\tau,\gam')$ intersectable if there exists some intersection pattern $\tau_P$ s.t. any to-be-intersected vertex in $\gamma$ and $\gamma'$ gets intersected in $\tau_P$, i.e.,
	\[ 
	\int(\gamma) \cup \int(\gamma')= \int(\tau_P)
	\] 
	For an intersectable tuple, we define its coefficient \[ 
		\lambda_{\gam\circ\tau\circ\gp} \coloneqq \lambda_\gam\cdot \lambda_\tau\cdot\lambda_\gp =\left(\frac{k}{n}  \right)^{|V(\gam\circ\tau\circ\gp)|-\frac{|U_\gam|+|V_\gp|}{2} }\cdot   \edgeco^{|E(\gam\circ\tau\circ\gp)}
		\]
	\end{definition}
\paragraph{Overview of intersection-term charging}
	The charging strategy is largely similar to the 2-step charging strategy for middle-shape charging, despite each step may contain more technical twists.
	 \begin{enumerate}
	 	\item For each intersection term $\gam\circ\tau\circ\gam'$ it can be charged to the appropriate diagonal $U_\gam$ or $V_\gp$ (via \cref{lem:intersection-main}) ;
	 	\item For each diagonal, "not too many" intersection terms are being charged to it so that we can afford a union bound.
	 \end{enumerate}
	 \subsection{Charging a particular intersection shape}
		 \begin{lemma}\label{lem:intersection-main}
	 	For intersection-tuple $(\gam, \tau,\gam')$ that intersect to $\tau_P$, 
	 \[ \lambda_{\gam\circ\tau\circ\gp}\cdot \|M_{\tau_P} \|\leq o(1)   \]
	 \end{lemma}

%	 		 Unpacking the equation gives us \begin{align*}
%	 	LHS&\leq \edgeco^{|E(\tau_P)|}\cdot \|M_{\tau_P}\|\cdot \|H_{\gam\circ\tau\circ\gam',(\cp_{U_\gam},\cp_{V_\gp} )}\|\\
%	 \end{align*}
%	 We will again split into the factors of $n$ and $k$.
%	 
	Since $\tau_P$ is a multigraph that arises from intersection, before we unpack the above equation, we need to first understand what our norm bounds $\|M_{\tau_P}\|$ yield for improper graphs coming from intersections.
	% For starters, let us first recall that our norm bounds for graph matrix for color-shape have essentially three components: 1) the factors of $n$ from usual graph matrix norm bounds; 2) the factors of color-automorphism, or simply shape automorphism when we restrict to the single-color case, and 3) the factor of $k$ as we combine the color choices for a vertex, and the potential edge factor of $k-1$ from color-value of a pure-color edge into graph matrix. 
 % Fortunately, despite our usual norm bounds target (proper) shape that does not have multi-edges (while multi-vertices makes no buzz), it is not too hard to extend the norm bounds for improper shape once we consider linearization of improper shapes. 
	
	\subsubsection{Bounding $\|\cm_{\tau_P}\|$}
\begin{fact}
	For $\chi_e$ a $p$-biased Fourier character, we have \[ 
	\chi_e^k= \E[\chi_e^k] + \E[\chi_e^{k+1}]\cdot \chi_e
	\]
\end{fact}
\begin{definition}[Linearization of $\tau_P$]
	Given a multigraph $\tau_P$, we call $\psi$ a linearization of $\tau_P$ if \begin{enumerate}
		\item Each edge in $\psi$ is a (multi-)edge;
		\item Each multiedge in $\tau_P$ has multiplicity $0$ or $1$ in $\psi$.
	\end{enumerate} 
	\end{definition}
\begin{definition}[Vapor coefficient $\vap$ for mul-$1$ edge]
	Fix $\tau_P$ a multigraph, and $\psi(\tau_P)$ its linearization, for any $e\in \psi(E(\tau_P))$, we define its vapor coefficient to be \[ 
	\vap(e) = \mul_{\tau_P}(e) - 1
	\]  
	For intuition, any edge of multiplicity $2$ in $\tau_P$ that linearizes to a mul-$1$ edge in $\psi$ receives a vapor coefficient of $1$. 
\end{definition}
\begin{definition}[Vapor and phantom coefficient $\theta$ for $\psi(\tau_P)$]
	Fix $\tau_P$ a multigraph, and $\psi(\tau_P)$ its linearization, we define the corresponding vapor coefficient of $\psi$ to be \[ 
	\theta_{\tau_P}(\psi) =\sum_{e\in E(\psi)} \vap(e) 
	\]
	Furthermore, let $\phantom_{\tau_P}(\psi)$ be the number of edges that disappear from $E(\tau_P)$ in the linearization $\psi$. The dependence on $\tau_P$ is dropped whenever it is clear.
\end{definition}
\begin{definition}[Floating stick]
	Given an improper graph $\tau_P$ and its linearization $\psi$, we call a component $C\in \psi$ a \emph{floating stick} if \begin{enumerate}
	\item $C$ is floating;
	\item $V(C) = E_\psi(C) - 1$, i.e., this is path.
	\end{enumerate}
\end{definition}
\begin{definition}[Floating cycle]
	Given an improper graph $\tau_P$ and its linearization $\psi$, we call a component $C\in \psi$ a \emph{floating stick} if \begin{enumerate}
	\item $C$ is floating;
	\item $V(C) \geq  E_\psi(C) $, i.e., this is at least a cycle.
	\end{enumerate}

\end{definition}
We remark that the above floating components are further distinguished as they have different contribution to matrix norm bounds: a floating stick $C$ contributes norm $\sqrt{n}^{V(C)}\exp(-d\cdot (|V(C)|-1) )$ while a floating cycle would contribute norm $\max \left(\sqrt{n}^{V(C)}, (\sqrt{\frac{n}{d}})^{E_\psi(C)} \right) \cdot \left(2|V(\tau_P)|q_{\tau_P} \right) $ depending on whether the floating cycle is part of the separator.  

%\begin{proposition} For an intersectable tuple $(\gamma,\tau,\gam')$ that intersects to $\tau_P$,
%	\[ 
%	\|M_{\tau_P}\| \leq 2^{\vap+ \phantom} \left(\sqrt{\frac{n}{d}} \right)^{\theta(\psi)} \max_{S:\text{separator in } \psi}\left(\sqrt{\frac{n}{d}}\right)^{E_{\psi}(S)}\sqrt{n}^{V(\tau_P)\setminus V(S) }\sqrt{n}^{I_{\psi}}\sqrt{d}^{\dang(\psi)} 
%\]
%\end{proposition}
%\begin{remark}
%	Each phantom/vapor edge comes with an extra $\eta$ factor uncharged that can be used to handle the extra factor $2$ blow-up each edge.
%\end{remark}
%

\begin{proposition}[Norm bounds for improper shape from $\gam\circ\tau\circ\gp$ intersection]
	Let $\gam\circ\tau\circ\gp$ be an intersection-tuple that intersect to $\tau_P$, \begin{align*}
		\|M_{\tau_P}\| &\leq  \cnorm^{|V(\tau_P) \setminus U_{\tau_P}\cap V_{\tau_P}|} 2^{|E(\tau_P)|} \max_{\substack{(\psi, S)\\ \psi: \text{linearization of $\tau_P$}\\ S:\text{a separator for }\psi }}  \left(\sqrt{\frac{1-p}{p}} \right)^{\theta(\psi)} \left(\sqrt{\frac{1-p}{p}}\right)^{E_{\psi}(S)}\sqrt{n}^{V(\tau_P)\setminus V(S) }\sqrt{n}^{I_{\psi}}\\& \cdot  \dang(\psi\setminus S)\cdot \float(\psi) 
	\end{align*}
\end{proposition}

\begin{proof}
	We highlight the difference from graph matrix norm bounds for usual color-shape,\\
	 \textbf{Identifying the linearization:} The linearization $\psi$ is fixed by identifying for each multiedge whether it becomes mul-$0$ or $1$ in $\psi$, hence $ 2^{|E(\tau_P)|}$ is sufficient where $|E(\tau_P)|$ is the number of distinct edges in $\tau_P$;
	
	% \textbf{Factors of $k$:} It is not hard to see that the tight case comes from each edge from being a pure-edge for the connected part, as we can either assign a factor of $k-1$ to an edge if it leads to new vertices or it is a pure-edge, giving us a factor of $(k-1)^{|E(\gam\circ\tau\circ\gp)|}$;
	
	\textbf{Factors of $n$:} To bound the factor of $n$, we note that its factor comes from the vapor coefficient via linearization and the norm bound factor for the linearization of $\psi$. Via our linearizartion, each multiplicity of edge that vaporizes gives a factor of $ \sqrt{\frac{1-p}{p}} $, captured by the factor of $\left(\sqrt{\frac{1-p}{p}}\right)^{\theta(\psi)}$, combining with the norm bound factor for $\psi$ gives us \begin{align*}
		&\sqrt{\frac{1-p}{p}}^{\theta(\psi)} \cdot \max_{S: \text{separator for }\psi }   \left(\sqrt{\frac{1-p}{p}}\right)^{E_{\psi}(S)}\sqrt{n}^{V(\tau_P)\setminus V(S) }\sqrt{n}^{I_{\psi}} \\& \cdot  \float(\psi) \cdot \dang(\psi\setminus S)
	\end{align*} 
		
% \jnote{serious changes needed below}		
	
% 	\textbf{Automorphism factor} Recall that the graph matrix by definition also has built-in normalizing factor from automorphism (or color-automorphism as we are working with color-shape), and applying our matrix norm bounds prompts us that we need a normalization of $\frac{1}{\aut(\psi)}$ to off-set the automorphism factor in the norm bounds, while we notice we pick up normalizing constant $\frac{1}{|\aut(\gam)|\cdot |\aut(\tau)|\cdot |\aut(\gp)| }$ from the normalizing factor from graph matrix corresponding to $\gam,\tau,\gp$ respectively. 
% 	\begin{lemma}
% 		For any shape $\psi$ that comes from linearization of intersection-tuple $\gam\circ\tau\circ\gp$, there exists a constant $c_\aut>0$ s.t.  \[ 
% 		\frac{\aut(\tau_P)}{ \aut(\gam) \cdot \aut(\tau)\cdot \cdot \aut(\gp)} \leq c^{|E(\gam\circ\tau\circ\gp)|}
% 		\]	\end{lemma}
% 		\begin{proof} We first observe that given a $\tau_P$ coming from a single level of $\gam\circ\tau\circ\gp$ intersection, each shape $\gam,\tau,\gp$ can be identified using edge-decay; and furthermore, there cannot be too many non-isomorphic shapes on the same number of edges, and this completes the proof.	
% \end{proof}
%	\end{enumerate}
Combining the above gives us the desired norm bound.
\end{proof}
Now that we have identified the quantity of interest for $\|\cm_{\tau_P}\|$, it is time to move on and see why the norm bound factor can be offset by $\lambda_{\gam\circ\tau\circ\gp}$, and hopefully, be left with some gap (at least a constant per edge) for counting. Let's proceed by unpacking the desired equation, \begin{align*}
	\lambda_{\gam\circ\tau\circ\gp} \|\cm_{\tau_P}\| &\leq \left(\frac{k}{n}\right)^{|V(\gam\circ\tau\circ\gp)|-\frac{|U_\gamma|+|V_\gp|}{2} } \cdot  \edgeco^{|E(\gam\circ\tau\circ\gp)|}\cdot 2^{|E(\tau_P)|} \\&\cdot \max_{\substack{(\psi, S)\\ \psi: \text{linearization of $\tau_P$}\\ S:\text{a separator for }\psi }}  \left(\sqrt{\frac{n}{d}} \right)^{\theta(\psi)} \left(\sqrt{\frac{n}{d}}\right)^{E_{\psi}(S)}\sqrt{n}^{V(\tau_P)\setminus V(S) }\sqrt{n}^{I_{\psi}}  \cdot  \float(\psi)\cdot \dang(\psi\setminus S)\\
	&\leq \left(\frac{1}{\ceta\cdot \sqrt{d}}\right)^{|V(\gam\circ\tau\circ\gp)|- \frac{|U_\gam|+|V_\gp|}{2} }\cdot  \left(\sqrt{\frac{d}{n}} \right)^{|E(\gam\circ\tau\circ\gp)|}\cdot 2^{|E(\tau_P)|} \\&\cdot \max_{\substack{(\psi, S)\\ \psi: \text{linearization of $\tau_P$}\\ S:\text{a separator for }\psi }}  \left(\sqrt{\frac{n}{d}} \right)^{\theta(\psi)} \left(\sqrt{\frac{n}{d}}\right)^{E_{\psi}(S)}\sqrt{n}^{V(\tau_P)\setminus V(S) }\sqrt{n}^{I_{\psi}}\\&  \cdot \float(\psi)\cdot \dang(\psi\setminus S)
\end{align*}
where we transfer the $\frac{1}{k}$ vertex factor to $\frac{1}{\ceta\sqrt{d}}$ using our hardness assumption.
 
\subsubsection{Controlling the factor of $d$ and $n$ for intersection term}
In this subsection, the quantity of interest would be the factor of $d$ and $n$. Our goal of this section is to show the following lemma,
\begin{lemma}\label{lemma:charging-single-intersection}
For $\tau_P$ that arises from intersection of $\gam\circ\tau\circ\gp$, 
	\begin{align*}
		&\left(\frac{1}{\ceta\cdot \sqrt{d}}\right)^{|V(\gam\circ\tau\circ\gp)|- \frac{|U_\gam|+|V_\gp|}{2} }\cdot  \left(\sqrt{\frac{d}{n}} \right)^{|E(\gam\circ\tau\circ\gp)|}\cdot 2^{|E(\tau_P)|} \cdot \max_{\substack{(\psi, S)\\ \psi: \text{linearization of $\tau_P$}\\ S:\text{a separator for }\psi }}  \left(\sqrt{\frac{n}{d}} \right)^{\theta(\psi)} \ \left(\sqrt{\frac{n}{d}}\right)^{E_{\psi}(S)}\sqrt{n}^{V(\tau_P)\setminus V(S) }\\&\cdot  \sqrt{n}^{I_{\psi}}\cdot  \float(\psi)\cdot \dang(\psi) 
\leq o(1)
	\end{align*} 
\end{lemma}
\paragraph{Proof overview}
The above equation lies at the core of our argument for bounding intersection terms. It would be helpful to unpack and classify the terms from the equation into a couple of components (ordered by their relative magnitude).

To start with, each vertex (outside the separator) contributes a factor of $\sqrt{n}$ to the above equation (ignoring its potential dangling/floating factor captured in a separate argument), it suffices for us to assign a factor of $\sqrt{d}$ for each such vertex, and suppose the vertex is connected to the separator $S$ in $\tau_P$, we can offset its contribution via \[
	\frac{1}{\sqrt{d}}\cdot \sqrt{n} \cdot \sqrt{\frac{d}{n}} \leq 1 
	 \]
	 Similarly, for isolated vertex, we need to assign an extra factor of $\sqrt{d}$ and an extra multiplicity of edges in $E(\gam\circ\tau\circ\gp)$. Besides connectivity in $\tau_P$ that requires some work, the major ingredient for this component is the vertex-intersection-tradeoff that says we have enough $\frac{1}{\sqrt{d}}$ factors for vertices outside the new separator after the intersection, captured by the following lemma,
	\begin{lemma}[Vertex factor for intersection (Intersection trade-off lemma from \cite{JPRTX, potechin2023machinery} )] For $\gam\circ\tau\circ\gp$ that intersects to $\tau_P$,
		\[ 
	 |V(\gam\circ\tau\circ\gp)| - \frac{|U_\gam|+|V_\gp|}{2} \geq |V(\tau_P)\setminus V(S)| + |I_\psi|
	 \]
	 for any linearization $\psi$ and separator $S$ of $\psi$;

	\end{lemma}
	Moreover, we need to assign at least one edge in $E(\gam\circ\tau\circ\gp)$ for each vertex in $V(\tau_P)\setminus V(S)$, and potentially more if they are isolated outlined by the above argument. Additionally, we need to identify extra $\sqrt{\frac{d}{n}}$ factor to handle dangling and floating components. The extra gap here relies upon the $k$-wise symmetry of our coefficients, that we start with $\gam$,$\tau$,$\gp$ that do not have dangling branches in the very beginning; hence, any component that becomes floating/dangling is a result of over-linearization (i.e., multi-edges become mul-$0$). That being said, multi-edges turning mul-$0$ can in general be tricky as they are the reason we pick up extra factors for isolated vertices (and dang/floating factors). Towards this end, we first observe the following,
\begin{proposition}[One phantom edge-mul for one dangling/floating factor]
	It suffices for us to assign one edge-multiplicity to each dangling/floating factor.
\end{proposition}
\begin{proof}
	Observe that we pick up a coefficient $\sqrt{\frac{d}{n}}$ for each edge's phantom multiplicity, and each dangling/floating factor is either $2|V(\psi)|q_\psi$ or $\sqrt{n}\exp(-d)$ (in which case the component is outside the separator). \end{proof}

	At this point, we would like an edge-factor assignment scheme that allows us to 	assign edges (specifically their multiplicities) in $E(\gam\circ\tau\circ\gp)$ such that
	
	\textbf{Edge-factor assignment target: necessary condition}
	\begin{enumerate}
		\item Assign one multiplicity to each multiplicity that vaporizes in $\theta(\psi)$ and separator edges in $E_\psi(S)$;
		\item Assign at least one edge's multiplicity to each non-isolated vertex in $V(\tau_P)\setminus V(S)$;
		\item Assign at least two multiplicities to each isolated vertex;
		\item Assign at at least multiplicity $c>0$ (where $c$ may be fractional if needed) to each floating component or dangling branch to offset the dangling/floating factor.
\end{enumerate}

We note that the following target would also be sufficient, as it strengthens the last condition to each multiplicity being charged to at most one dangling/floating factor.

\begin{mdframed}[frametitle = {Edge-factor assignment target: sufficient condition}]
\begin{enumerate}
	\item Assign one multiplicity to each multiplicity that vaporizes in $\theta(\psi)$ and separator edges in $E_\psi(S)$;
		\item Assign at least one edge's multiplicity to each non-isolated vertex in $V(\tau_P)\setminus V(S)$;
		\item Assign at least two multiplicities to each isolated vertex;
		\item Assign at at least one multiplicity to each floating component or dangling branches to offset the dangling/floating factor.\end{enumerate}	
\end{mdframed}

Before we delve into the argument, we remind the reader that $\phantom(\psi))$ is the multiplicity of edges that vanish from $E(\tau_P)$, i.e., of multiedges that become multiplicity $0$ in $\psi$. And instead of working with $E(\gam\circ\tau\circ\gp)$ that counts multiplicity of edges, it is easier for us to consider the (non multi-)graph  $\psi$ and $\phantom(\psi)$ that allows us to avoid worrying about vapor edges and their multiplicities as well. 
%As each edge needs to be at least mul-$2$ to become phantom (i.e., completely vanish), a straightforward bound we would be working is simply 
\begin{claim}
	Recall $\theta(\psi)$ is the multiplicity of edges that vaporize in the linearization to $\psi$ from $\tau_P$, \[ 
	|E_\psi| +\phantom(\psi) =  |E(\gam\circ\tau\circ\gp)| - |\theta(\psi)| 
	\]
\end{claim}
\begin{lemma}[Edge-assignment] \label{lemma:edge-assignment } For $\gam\circ\tau\circ\gp$ that intersects to $\tau_P$ and let $\psi$ be a linearization of $\tau_P$, s.t. \[ 
|E_\psi|+ \phantom(\psi) \geq |V(\tau_P)\setminus V(S)|+|E_\psi(S)|+|I_\psi| + |\dang_\psi|+|\float_\psi|  \numberthis
\]
where we define \[ 
|\float_\psi|  \coloneqq |\fs_\psi| + |\fc_\psi|  
\]
and remind the reader that $\phantom(\psi)$ counts the multiplicity of edges that completely vanish in $\psi$ from linearization $\tau_P$.
\end{lemma}
	\begin{proof}[Proof to \cref{lemma:edge-assignment }]
	We start by considering $\psi$ the linearized graph (recall $\psi$ may have phantom edges removed) and put phantom edges back along the way. Take $S$ to be the SMVS for $\psi$, we then consider the following recursive process to traverse $\psi$ via edges in $E(\tau_P)$.
	
	 Throughout the process, vertices in the graph $V(\psi)$ can be partitioned as the following,
	\begin{enumerate}
		\item $V_W\subseteq V(\tau_P)$: those reachable from $W$ via edges in $E_\psi$ or already in $W$;
		\item For a connected component $C\subseteq V(\psi)\setminus V_W$ in $E_\psi$ while not yet reachable from $W$, it is either a non-floating component, a floating-stick, or a floating-cyc;
		\item For isolated vertices in $V(\psi)\setminus V_W$, we can group them according to the phantom edges into a collection of iso-connected components.
	\end{enumerate} 
	\begin{mdframed}[frametitle = {$W$-exploration procedure }]
		\begin{enumerate}
		\item Let $W$ be the current set of vertices visited (initialized to be $S$ an arbitrary SMVS of $\psi$);
		\item Explore the vertices (not yet in $W$) while connected to $W$ via edges in $E(\psi)$, i.e., assign the edge to each vertex it leads to;
		\item Process the dangling/floating factor for vertices in $W$, and this may potentially explore component connected to $W$ via phantom edge;
		\item Explore a component connected to $W$ via some phantom edge;
		\item For vertices outside $W$ and not reachable via phantom from $W$, there must be a phantom-edge connecting two different components, process that phantom edge and explore both components;
		\item Repeat this process until all vertices are pushed into $W$.
	\end{enumerate}

	\end{mdframed}

	 For starters, we first notice for the component $C_S$ that is connected to the separator of $\psi$, we have \[
	|E_\psi(C_S)| \geq  |V(C_S)\setminus V(S)| + |E_\psi(S)|
	 \]
	 as each vertex is connected to the separator, and we can charge it to the edge leading to the vertex when we BFS from $S$, and each separator edge in $E_\psi(S)$ contributes an $1$ to both sides. That being said, we may have dangling/floating factors not yet accounted in $C_S$, not to mention there may be components not reachable from $S$ in $\psi$.

	 For the subsequent steps involving dangling/floating factor , we need to repeatedly apply the following proposition which allows us to identify either there is a new phantom edge incident or there is some excess edge not used in the \emph{exploration} process in the component, \begin{proposition}[Degree-at-least-two] \label{prop:at-least-deg-two}
	 In $\tau_P$ where phantom edges are yet to vanish, each vertex in $V(\tau_P)\setminus (U_{\tau_P}\cup V_{\tau_P}) $ have degree at least $2$ (not counting multiplicities).
	 \end{proposition}
	 \begin{proof}
	 Each vertex in $V(\tau_P)\setminus (U_{\tau_P}\cup V_{\tau_P})$ comes from at least one of the following two categories, $V(\gamma)\setminus U_\gam, V(\gp)\setminus V_\gp$ and $V(\tau)\setminus (U_\tau\cup V_\tau)$. Any vertex in the first category has degree at least $2$ in $\gamma$ from $\gamma$ shape property that $V_\gamma$ is the MVS for $\gamma$ (analogously for $\gp)$, and any vertex in $V(\tau)\setminus U_\tau\cup V_\tau$ has degree at least $2$ in $\tau$ by $k$-wise symmetry.
	 \end{proof}

%	 \begin{corollary}
%	 	For any component reached via some phantom edge in the process while currently not in $W$, it is incident to some uncharged phantom edge. 
%	 \end{corollary}
	Let $W$ be the current set, and let $E_\psi(W), \phantom_W(\psi)$ be the multiplicity of edges traversed so far, and $I_\psi, \dang_\psi, \float_\psi$ be the factors processed so far. We maintain the invariant, \[ 
	 |E_\psi(W)|+ \phantom_W(\psi) \geq |V(W)\setminus V(S)|+|E_\psi(S)|+|I_\psi(W)| + |\dang_\psi(W) |+|\float_\psi(W) | \]
	 In the base case, we have $W=S$, and the inequality holds immediately.	As noted above, it is also immediate for $V_W = V(C_S)$. We now start processing potential dangling/floating factors on vertices in $V_W$. Notice in the base case, the only factors involved are dangling factors, and we can restrict our attention to its end-point, and notice it is either incident to an uncharged edge in $E_\psi$, or a phantom edge via Proposition~\ref{prop:at-least-deg-two}. For a phantom edge, we discuss by cases depending on where it leads to, as it may incur new floating factors for which we need to either identify there is either some excess edge in the new component, or a new phantom edge incident to the component.
	 	 \begin{enumerate}
	 	\item The phantom edge leads to another vertex in $W$, and the vertex it leads to may be assigned another dangling/floating factor as well, in which case we pick up at least $2$  (from phantom multiplicities) and $2$ from RHS for $2$ dangling/floating factors;
	 	\item The phantom edge leads to a component $C$ not yet explored, i.e., not in $W$, in which case we pick up at least $2$ phantom multiplicities from the LHS. We assign one multiplicity to the current dangling/floating factor being processed with one remaining phantom multiplicity yet-to-assign . It now remains to assign this multiplicity and identify an excess edge for the floating factor of the new component,
	 	\begin{itemize}
	 		\item If $C$ is a floating-stick, assign the above multiplicity to the vertex in $C$ incident to the current phantom edge, which also puts $C$ into $W$. Moreover, we claim that $V(C)$ must have at least one more degree-$1$ vertex in $\psi$ that is incident to some not yet processed phantom edge by Proposition~\ref{prop:at-least-deg-two} (notice we get extra gap if $C$ has more than $2$ deg-1 vertices);
	 		\item If $C$ is a floating-cycle, it must be outside $S$, while $E_\psi(C)\geq V_\psi(C)$, hence we can use the remaining multiplicity from earlier to assign to the floating factor of $C$, and use the edges in $E_\psi(C)$ to \emph{explore} vertices in $V_\psi(C)$;
	 		\item If $C$ is non-floating, (eg. $V(C)$ intersects with $U_\psi\cup V_\psi$), there is no dangling/floating factor on $C$, and we assign the above multiplicity to the vertex in $C$ incident to the current phantom edge;
	 		\item If $C$ is an iso-connected component, i.e., a connected component in $\tau_P$ that become isolated under $\psi$, since each vertex in $V(C)$ is of degree at least $2$ in $\tau_P$, assign the remaining multiplicity from above to the first vertex in $V(C)$. Note that since it is isolated, it remains for us to find assign this \emph{source} vertex another edge-multiplicity. Putting that aside, we can traverse the iso-connected component, and charge each phantom edge (and its $2$ multiplicities) to the new vertex it explores. Since each vertex has degree at least $2$, there must be an incidental phantom edge, and we can assign one of its multiplicities to the \emph{source} vertex of this component, filling the gap from earlier. The other multiplicity is then reserved to explore the potentially new component.
 	 		 	\end{itemize} 
	 	\end{enumerate}
	 	
	 	For components not yet explored, and they are not connected to the current $W$ via any phantom edge, we claim that there must be a phantom edge that connects two unexplored components. This follows via our intersection-tuple property, that in $\tau_P$, each vertex has a path to $U_\psi$ and a path to $V_\psi$. Hence it is either reachable from the separator in $\psi$, or reachable from some phantom edge. We can then use this two multiplicities to \emph{explore} these two components into $W$, and notice each of them may be floating, while we can follow the above argument and identify a new phantom edge or an excess edge for each floating component. This completes our proof of Lemma~\ref{lemma:edge-assignment }.	
 \end{proof}

 \subsection{Bounding intersection count}
 Analogous to our charging in the base level for middle shapes, as we have shown each single intersection term can be bounded, we may now complete the proof by bounding the count. That said, we first notice that each intersection term can be specified as the following.

%  \begin{lemma}
%  	Let $\exc(\tau_P)$ be the number of excess edges in $\tau_P$ where we recall $\tau_P$ is the graph potentially with multiedges, we can bound the number of  intersection patterns  by \[ 
%  	(c_\int)^{E(\tau_P)} \cdot (2|V(\tau_P)|q_{\tau_P})^{\exc(\tau_P)} 
%  	\]
%  \end{lemma}
%  %Idea: it is not sufficient to simply consider non-equivalent intersection patterns as we may have non-equivalent patterns that correspond to the same intersected shape $\tau_P$ with the same coefficient. The confusion comes in specifying $\dsos$ levels some ''dormant'' vertex if we start from $\tau$ and go outwards to $\gam_k, \gp_k$, but the choice of $\dsos$ levels comes with $-1$ so they can be grouped to get a count of $1$ as opposed to $\dsos$ from the number of choices of picking the levels.

%  Overall idea:
%  \begin{enumerate}	\item The idea is to specify for each edge where which $\dsos$ level it comes from, while only using a constant per edge;
%  	\item The primitive we can use is to keep a counter indicating which $\dsos$ level are we on currently at the vertex, and then for the next edge, we can simply go $+1/0/-1$ indicating whether we change the $\dsos$ level;
%  	\item Unless truly necessary, i.e., we are seeing cycle in $\tau_P$, a label in $2|V(\al)|q$ is not needed in the above argument;
% % %	\item There is a way to traverse non-squeezable intersection tuple s.t. each edge is traversed in adjacent $\dsos$ level.
%  	\end{enumerate}
 \paragraph{Identifying what to bound}
 \begin{definition}[Intersection-pattern]
 	We call a tuple $\left(\gam_t,\dots, \gam_1, \tau, \gp_1, \dots,\gp, P_\int \right)$ an intersection pattern such that  
 	\begin{enumerate}
 		\item $P_\int$ specifies which vertices get intersected and how they are intersected, and let $P_{\int,L,R}$ be $P_\int$ restricted to the first $(L,R)$-levels $\gam_L,\dots,\gam_1,\tau,\gp_1,\dots,\gp_R$.
 		\item For any $\ell,r  \in \N$, $\bigg(\gam_\ell, \left(\gam_{\ell-1},\dots,\gam_1,\tau, \gp_1, \dots, \gp_{r-1}\right), \gam_r\bigg)$ is a intersection-tuple with $P_{\ell, r}$;
 			\item Additionally, call $\depth(P_\int)= l$ the intersection-depth of the corresponding pattern;
		\item For technical reasons, we allow $\gam_i = \gam_{i-1}$ or $\gp_i=\gp_{i-1}$ for $i>1$ to represent the case when there is no \emph{new} left/right intersection. However, both cannot hold at the same time.
  	\end{enumerate}
 	\end{definition}

 \begin{definition}[Intersection shape $\tau_P$]
     We call a shape an intersection shape if it it a shape that can be obtained by ribbon intersections of at most $\dsos$ levels.
 \end{definition}
To bound the count of intersection terms, we first observe that $\tau_P$ when viewed as shape can be counted in the identical way as for middle shape, with the only difference being for a fixed shape, there may be various ways of ribbon interesctions (intersection patterns) that obtain the same underlying multi-graph as a shape. That said, to the number of intersection patterns for a fixed shape, it suffices for us to identify a $\dsos$ label for each edge specifying which level the edge comes from (i.e. the depth of $\gamma_i$ where the edge is from).

As a result, given the final shape $\tau_P$ viewed as a multi-graph with information of where each edge comes from, the intersection pattern is immediately fixed as one may \begin{enumerate}
    \item Draw out each $\gamma_i, \gamma_i'^T$for each $i\in [\dsos]$ and $\tau$ (the base-level) according to the depth-information from each edge;
    \item The vertex intersection, i.e. which vertex colliding with each other, is also fixed as the final shape of $\tau_P$ is also given.
\end{enumerate}

\begin{proposition}\label{prop:int-count}
    Given a shape $\tau_P$, the number of intersection levels can be identified at a cost of \[ 
    c_{int}^{|V(\tau_P)\setminus U_{\tau_P}\cap V_{\tau_P}|}  \cdot \dsos^{|V(\tau_P)|}
    \]
    for some absolute constant $c_{int}$ where we remind the reader that $|V(\tau_P)|$ counts multiplicity of intersected vertices.
\end{proposition}
\begin{proof}
    We apply this bound alongside our norm bound for grouped matrix. In addition to the block-value bound for middle shape, it suffices for us to go through the shape in an extra pass and identify where each edge comes from (i.e., identifying the particular intersection pattern).

   First, identify a level of $[\dsos]$ for each vertex-multiplicity identifying which $\dsos$ level it appears, and this is subsumed by the extra decay factor of $ \frac{1}{\dsos}$ from our coefficient; once all the $\dsos$ levels are identified, we appeal to the following structural left/right-intersection shape property: each $\gamma$-shape is $U$-wedge.
   
   That said, once $\tau_P$ is identified as an edge-set in the trace-method calculation alongside the $[\dsos]$ label for each vertex-multiplicity we are provided now, it suffices for us to apply wedge-bundling idea for each $\gamma_i, \gamma_i^T$ level recursively from outermost level to the base level of $\tau$, and applying ~\pref{lem: bundling_cost}, this is a single constant per vertex and this completes our bound.
\end{proof}

\subsection{Wrapping up the intersection shape bound}
\begin{proof}[Proof to Lemma~\ref{lemma:intmain}]
    This section is in parallel with the summary for middle shape bound. Notice that for a single active-profile $\calP$, apply our graph matrix bound on
\begin{align*}
	\lambda_{\cal_P} \calM_{\calP} \coloneqq  \sum_{\substack{\tau_P: \text{intersection shapes with active profile} \calP  \\ \text{non-trivial, permissible}  \\ E(U_\tau\cap V_\tau)=\emptyset }  }   \lambda_{\tau_P} \cdot   M_{\tau_P} \cdot q(\tau) 
	\end{align*}
% where we let $\beta$ be the coloring of any single shape $\tau$, and observe that \begin{align*}
% 	&\sum_{\substack{\tau: \text{middle shapes with active profile} \calP  \\ \text{non-trivial  }  }} \lambda_\tau \sum_{\beta:V(\tau)\rightarrow [k]} \cm_\tau \cval{E(\tau)}_\beta \cdot \\&  = \sum_{\substack{\tau: \text{middle shapes with active profile} \calP  \\ \text{non-trivial, permissible} \\  E(U_\tau\cap V_\tau)=\emptyset  }  } \sum_{\beta:V(\beta)\rightarrow [k]}  \lambda_\tau \cdot  \cval{E(\tau)}_\beta \cdot  \cm_\tau \cdot q(\tau) 
% \end{align*}
Applying our ''grouped'' norm bound on $\lambda_{\calP} \cm_{\calP}$ for active-profile $\calP$  while now $\calP$ is the collection of intersection shapes $\tau_P$ gives us \begin{align*}
	\|\lambda_{\calP} \calM_{\calP}\|&\leq \max_{\tau\in \tau_P(\calP) \text{non-trivial, permissible} } \|  \lambda_{\tau_P}  \cdot M_{\tau_P} \cdot q(\tau_P)   \| \\
	&\leq \max_{\tau_P\in \tau_P(\calP) \text{non-trivial, permissible} } \|  \lambda_{\tau_P}  \cdot M_{\tau_P}\| \cdot |q(\tau_P)|
\end{align*}

We now mimic our charging strategy for middle shapes , and we note that the factor other than $\dsos$ dependence follows from the charging for a single intersection shape in Lemma~\cref{lemma:charging-single-intersection}, and we make clear the $\dsos$ dependence here as,\begin{enumerate}
    
	\item Assuming $\frac{k}{n} = O(\frac{1}{c_\eta \sqrt{d}\cdot \dsos^4 })$, we first consider the following $\dsos$ dependence that inheir from that of middle shapes.;
	\item Each vertex outside $U_\tau\cup V_\tau$ contributes a single factor of $\dsos$ via the blow-up from hidden-edge indicator, while each comes with a full factor of decay $\frac{1}{\dsos^2}$;
	\item Each vertex in $U_\tau \Delta V_\tau$ does not contribute any $\dsos$ factor from hidden-edge indicator, and each comes with a factor of $\frac{1}{\dsos}$ since each comes with half a vertex coefficient;
	\item By our sparsity bound, $|E(S)|\leq 5|V(S)|$, and moreover, notice since we look at permissible shape, there are no edges inside $E(U_\tau\cap V_\tau)$; and the edges between $U_\tau\cap V_\tau$ and $V(\tau)\setminus (U_\tau\cap V_\tau )$ are at most $2\cdot |V(\tau)\setminus (U_\tau\cap V_\tau )|$, and therefore, we have $|E(S)|\leq 7|V(\tau)\setminus U_\tau\cap V_\tau| \leq c^{|V(\tau)\setminus U_\tau\cap V_\tau|}$ for some constant $c$ subsumed by $c_\eta$ factor of vertex-decay;
    \item The difference in comparing to the middle shape charging is that we incur an overhead to identify intersection patterns, and this is a factor of an extra $\dsos$ for each vertex from Proposition~\ref{prop:int-count};
    \item Recall that we have an $\frac{1}{\dsos}$ excess from charging middle shapes, which can now be used to offset the extra $\dsos$ overhead in Rmemark~\ref{remark:extra-gap-from-middle-shape};
	\item Combining the above gives us a factor of $\left(\frac{1}{\dsos}\right)^{|U_\tau\Delta V_\tau|} = \left(\frac{1}{\dsos}\right)^{|U_{active}(\calP) +|V_{active}(\calP)|   }$. 
\end{enumerate}
	Summing over all active-profiles gives us \begin{align*}
		\|\sum_{\calP} \lambda_{\calP}M_{\calP} \| \leq \max_{\calP}{ c(\calP)\cdot  \left(\frac{1}{\dsos}\right)^{|U_{active}(\calP) +|V_{active}(\calP)|   } }  \ll \frac{1}{10}
	\end{align*}
	where we recall that the identification cost of $\calP$ is bounded by $\dsos^{|U_{active}(\calP)| + |V_{active}(\calP)| }$. Putting back the edges in $E(U_{\tau_P} \cap V_{\tau_P})$ gives us back the indepednent set indicator on the left and right again, and this completes our bound for the intersection terms.

\end{proof}

\section{Acknowledgement} We would like to thank anonymous reviewers for various suggestions on improving on our writing. J.X. would like to thank Xinyu Wu, Tim Hsieh, and Sidhanth Mohanty for various helpful discussions.
 \clearpage
 \newpage
 \bibliographystyle{alpha}
\bibliography{madhur, bib}
\clearpage
\newpage
	\appendix
	\section{Deferred details for trace-method calculation}

\subsection{Handling singleton edges} \label{sec:singleton-decay}
		Given a constraint graph associated with a walk $P$, we let $G_P$ be the graph induced by edges of the constraint graph (where the dependence of $P$ is usually dropped). Throughout this section, let $\cycle$ be the indicator function for nearby-$2$-cycle, and $\bdd$ for bounded degree. Let $S$ be the set of singleton edges, and $T$ the rest of edges; for our convenience, we will consider $J\subseteq S$ to be set of singleton edges that appear in the random graph sample, and $\widehat{J} \coloneqq S\setminus J$ those that do not; similarly, let $W\subseteq T$ be the set of edges that appear in the random graph sample, and $\widehat{W}\coloneqq T\setminus W$ be the rest. When it is clear, we will use $S$ to represent $|S|$.
		The key quantity we want to bound is the following,
	 		 		\begin{align*}
   &\left|\E\left[\val(P) \cdot\cycle\cdot \bdd \right ]\right|\\ &= \left|\E\left[\prod_{e\in S\cup T} \left((\chi_e -\frac{d}{n}) \sqrt{\frac{n}{d}} \right)^{\mul(e)} \cdot \cycle\cdot\bdd  \right] \right|
    \\
    &= \sqrt{\frac{n}{d}}^{S+\mul(T)} \left|\E\left[\sum_{W\subseteq T, J\subseteq S} (1-\frac{d}{n})^{\mul(W)+\wbar}(\frac{d}{n})^{W+\mul(\wbar)}(-1)^{\mul(\wbar)+\jbar } (1-\frac{d}{n})^{J+\jbar}(\frac{d}{n})^{J+\jbar}\cdot \cycle \cdot \bdd |G_{W\cup J}=1, G_{\wbar\cup\jbar}=0 \right]  \right|\\
    &\leq \sqrt{\frac{n}{d}}^{S+\mul(T)} \sum_{W\subseteq T} (\frac{d}{n})^{W+\mul(\wbar)} \left|\E\left[ \sum_{J\subseteq S}(1-\frac{d}{n})^{J+\jbar}(\frac{d}{n})^{J+\jbar}(-1)^{\jbar}\cdot\cycle\cdot \bdd |G_{W\cup J}=1, G_{\wbar\cup\jbar}=0  \right] \right|  \\
    &\leq   2^{T}\sqrt{\frac{n}{d}}^{S+\mul(T)}\max_{W\subseteq T}(\frac{d}{n})^{W+\mul(\wbar)} \left|\E\left[ \sum_{J\subseteq S}(1-\frac{d}{n})^{|S|}(\frac{d}{n})^{|S|} (-1)^{\jbar}\cdot \cycle \cdot \bdd|G_{W\cup J}=1, G_{\wbar\cup\jbar}=0  \right] \right| \\
    &= 2^{T}\sqrt{\frac{n}{d}}^{S+\mul(T)}\max_{W\subseteq T}  (\frac{d}{n})^{W+\mul(\wbar)}(\frac{d}{n})^{|S|} \left| \sum_{J\subseteq S}(-1)^{\jbar}\Pr_G[\calE(G)\land \calF(G) |G_{W\cup J}=1, G_{\wbar\cup\jbar}=0] \right|
    % &\leq  \left|\sum_{H\subseteq T, J\subseteq S}\prod_{e\in H\cup J} \left((1-\frac{d}{n})\sqrt{\frac{n}{d}}\right)^{\mul(e)} \prod_{e\in \wbar\cup \jbar}\left(-\frac{d}{n} \sqrt{\frac{n}{d}}\right)^{\mul(e)} \Pr[X_{H\cup J} = 1, X_{\wbar\cup \jbar}=0, \calE] \right|\\
    % &\leq \sqrt{n/d}^{S+\mul(T)} \left| \sum_{J\subseteq S} (1-\frac{d}{n})^{J}(-\frac{d}{n})^{S-J} \Pr[X_{T\cup J}=1, X_{S/J}=0, \calE]    \right|\\
    % &\leq \sqrt{n/d}^{S+\mul(T)} \left| \sum_{J\subseteq S} (1-\frac{d}{n})^{J}(-\frac{d}{n})^{S-J} (\frac{d}{n})^{J\cup T} (1-\frac{d}{n})^{S-J} \Pr[ \calE | X_{T\cup J}=1, X_{S/J}=0]  \right|\\
    % &\leq \sqrt{n/d}^{S+\mul(T)} (\frac{d}{n})^{S\cup T} \left|\E \sum_{J\subseteq S} (-1)^{S-J} \Pr[\calE| X_{T\cup J}=1, X_{S/J}=0]  \right|\\
    % &=\gamma(S,T) \left|\sum_{J\subseteq S} (-1)^{S-J} \Pr[\calE| X_{T\cup J}=1, X_{S/J}=0]   \right|
\end{align*}

It now remains to bound the quantity $(\frac{d}{n})^{|S|} \left| \sum_{J\subseteq S}(-1)^{\bar{J}}\Pr[\calE(G)\land \calF(G) |G_{W\cup J}=1, G_{\wbar\cup\jbar}=0] \right|$ for a fixed $W\subseteq T$. To work with the distribution of $G_{n,d/n}$ conditioned on $G_{W\cup J}=1, G_{\wbar\cup \jbar}=0$, we appeal to the observation in Fan-Montanari, and consider an alternative distribution on  $G'$ a random graph sample where we sample edges outside $S\cup T$ w.p. $\frac{d}{n}$.
\begin{observation}
Consider $G'$ a random graph distribution where we sample edges outside $S\cup T$ w.p. $\frac{d}{n}$, then for any $J\subseteq S, W\subseteq T$, the following two distributions equal to each other, \[ \mu\{G' \cup J \cup W\} = \mu\{G|G_{J\cup W}=1, G_{\jbar\cup \wbar}=0\}\]
\end{observation}
\begin{observation}
Defining $f(g')=\sum_{J\subseteq S} (-1)^{\jbar} \cycle\cdot \bdd(g'\cup J\cup H) $, then \[
\sum_{J\subseteq S} (-1)^{\jbar} \Pr_{G\sim G_{n,d/n}}[\calE(G)| G_{J\cup W}=1, G_{ \bar{J} \cup \wbar} =0] = E_{g'\sim G'}[f(g')]\]
\end{observation}

We now shift our attention to the $f$ function. Consider a fixed $g'$ sampled from $G'$, notice that if there is an edge $e\in S$ that is \textbf{unforced}, i.e., for any $J\subseteq S\setminus\{e\}$, $\cycle\cdot \bdd(g'\cup \{e\}\cup J\cup W) =\cycle\cdot \bdd(g'\cup J\cup W)$, then $f(g')=0$. 

\begin{definition}
 Given a subgraph $g'$ and edge set $S,W$, an edge $e\in S$ is forced w.r.t. $(S,W)$ if there exists $J\subseteq S\setminus \{e\}$,
 \[ 
 \cycle(g'\cup \{e\}\cup J\cup W)\neq \cycle(g'\cup J\cup W)
 \]
 or \[
  \bdd(g'\cup \{e\}\cup J\cup W)\neq \bdd(g'\cup J\cup W)
  \]
\end{definition}
	\subsection{Singleton decay from absence of nearby 2-cycle }
	 		 	In this subsection, we bound the value of the walk under the conditioning on the absence of nearby 2-cycle.	Recall $\cycle$ is the indicator function for the absence of nearby $2$-cycle, 	 		 	
%	 		 	\begin{lemma} Let $\cl(E)$ be the number of cycles induced by edges in $E$,
%	 		 		\[
% \left|\E\left[\val(P)\cycle\right]\right|\leq 2^{S+T}\sqrt{\frac{n}{d}}^{S+\mul(T)}\max_{W\subseteq T}  (\frac{d}{n})^{W+\mul(\wbar )+S} C(\log n)^2n^{-0.7(\frac{|S|}{\ell}-\cl(S\cup W))}
% \]
%	 		 	\end{lemma}

\begin{lemma}\label{lem: bounding-new-cycles}
\[\Pr_{g'\sim G'}[\text{ every edge in S is forced by $\calE$ w.r.t. } (S,W) ] \leq  Cn^{\delta}n^{-0.7(\frac{|S|}{\ell}-\cl(S\cup W) )}\]
\end{lemma}
\begin{remark}
	We will simply bound $\cl(S\cup W)\leq m$ the number of surprise visits when combined with our counting.
\end{remark} 
\begin{remark}
	The important quantity is that we each singleton edge in $S$ gets assigned a factor of at most $n^{\frac{1}{c\log n}}$ and the other factors are either subsumed as auxiliary data that can be offset by long path in the trace calculation (after we take the $\frac{1}{q}$-th root, or it comes as a $n^{o(1)}$ blow-up due to surprise visit.
\end{remark}
% \jnote{the log factor might not be exact}
\begin{proof}
We observe that this is the probability that each edge in $S$ is inside a cycle of length at most $\ell$ in $g'\cup S\cup W$ (that does not exist in $S\cup W$); and since there are $|S|$ edges in total, we have at least $\frac{|S|}{\ell}$ new cycles in $g'\cup S\cup W$. We now complete the proof by appealing to Lemma A.3 from \cite{FM17, BMR19}.
\end{proof}

\begin{lemma}[Restatement of \cite{FM17} ] Fix $d>1$ and $\delta>0$ some absolute constant. Let $S\subseteq \binom{n}{2}$ of size at most $n^{\delta}$. and let $\cl(S)$ denote the number of cycles induced by the edge set $S$, and let $\ell\leq 0.4\log_d n$. Let $G_0\cup S$ denote the random subgraph of the complete graph in which each edge outside of $S$ is present independently w.p. $\frac{d}{n}$ and each edge inside $S$ present w.p. $1$. Let $V\subseteq [n]$ be the set of vertices incident to at least one edge in $S$, then for some $C=C(d)>0, N_0=N_0(d)>0, n\geq N_0$ and any $0\leq t\leq n^{O(\delta)}$,
\[ 
\Pr[\cl(B_\ell(V; G_0\cup S)\geq \cl(S)+t ]\leq Cn^\delta\cdot  n^{-0.7t}
\]
\end{lemma}
To wrap up the proof for the claim, we observe that if $\frac{S}{\ell}-\cl(S\cup T)<0$, then the bound is immediate as we are bounding the above probability by $1$;
on the other hand, if $\frac{S}{\ell}-\cl(S\cup T)\geq 0$, it follows from Lemma \ref{lem: bounding-new-cycles}, 

%\begin{align*}
%    \left|\E\left[\val(P) \cycle \right ]\right| &\leq 2^{T}\sqrt{\frac{n}{d}}^{S+\mul(T)}\max_{W\subseteq T}  (\frac{d}{n})^{W+\mul(\bar{W})}(\frac{d}{n})^{|S|}  \left|\sum_{J\subseteq S} (-1)^{S-J} \Pr_{G\sim G_{n,d/n}}[\calE| G_{ J}=1, G_{\cup }\}=0]   \right|\\ 
%    &=  2^{T}\sqrt{\frac{n}{d}}^{S+\mul(T)}\max_{W\subseteq T}  (\frac{d}{n})^{W+\mul(\wbar)}(\frac{d}{n})^{|S|}  |\E_{g'}[f(g')]|\\
%    &\leq 2^{T}\sqrt{\frac{n}{d}}^{S+\mul(T)}\max_{H\subseteq T}  (\frac{d}{n})^{W+\mul(\wbar)}(\frac{d}{n})^{|S|} 2^{S}\Pr[f(g')\neq 0]\\
%    &\leq 2^{S+T}\sqrt{\frac{n}{d}}^{S+\mul(T)}\max_{H\subseteq T}  (\frac{d}{n})^{W+\mul(\wbar)+S} C(\log n)^2n^{-0.7(\frac{|S|}{\ell}-\cl(S\cup W))}
%\end{align*}
	 		 	
\subsection{Singleton decay from bounded degree}
\begin{claim}
	Given $S$ a set of singleton edges, and let $m$ be the number of surprise visits in the walk, there is a set of edges $\tilde{S}\subseteq S$ s.t. $|\tilde{S}|\geq \frac{|S|-m}{2}$ and edges in $\tilde{S}$ are vertex-disjoint. 
\end{claim}	 	
\begin{proof}
	Remove all the surprise edges, the remaining subgraph is cycle-free, we then DFS on the graph and pick edges into $\tilde{S}$ in the alternating manner, we can ensure the edges picked are vertex-disjoint, and there is an $\tilde{S}$ that picks at least half of the remaining edges.
\end{proof}
We will further split $\tilde{S}$ into $\tilde{S}=S_1\cup S_2$ where $S_1$ is the set of edges whose both endpoints have degree less than $4d$ in $S\cup H$, and $S_2$ rest of the edges in $\tilde{S}$.
\begin{claim} There is an absolute constant $c_0>0$ s.t. for a vertex $i$, for any integer $t\in [6d, 10d]$,
\[ \Pr[\deg_i(g')=t]  \leq d^{-c_0d/\log d} \]
\end{claim}
\begin{proof}
Noticing the degree of a vertex in $g'$ is simply a sum of at least $n-O(q|E(\al)|)$ independent Bernoulli random variables with expectation $\frac{d}{n}$, for $|E(\al)|\ll n$, we have \[ 
\Pr[\deg_i(g')=t]\leq \Pr[\deg_i(g') = 6d] \leq (\frac{ne}{6d})^{6d}(\frac{d}{n})^{6d}\leq  e^{-c_0\cdot d }\leq d^{-c_0d/\log d}
\]
\end{proof}
Since the edges in $\tilde{S}$ are vertex-disjoint, the probability factorizes over edges, and we have \begin{align*}
    \Pr_{g'\sim G'}[\text{every edge in } \tilde{S} \text{ is forced by } \calF]&\leq \prod_{\{ i,j\} \in \tilde{S}}\Pr_{g'} [\deg_i(g')=\tau-1-\deg_i(H) \lor \deg_j(g')=\tau-1-\deg_j(H)  ]\\&\leq \prod_{\{i,j\}\in S_1} \Pr_{g'} [\deg_i(g')=\tau-1-\deg_i(H) \lor \deg_j(g')=\tau-1-\deg_j(H)  ]
    \\&\leq d^{-c_0d |S_1|/\log d}
\end{align*}
\begin{claim}
Recall $m$ is the number of surprise visits, and $2q|E(\al)|$ the length of the walk on the constraint graph,
\[
|S_2|\leq \frac{2q|E(\al)|-m}{2d}
\]
\end{claim}
\begin{proof}
We first start by removing the surprise visits, and then we observe for each edge in $S_2$, we can charge $4d$ edges to it, and since each edge can be charged to both its endpoints, we have \[
|S_2|\leq \frac{2(2q|E(\al)|-m)}{4d}=\frac{2q|E(\al)|-m}{2d}
\]
\end{proof}
Combining the above claims gives us the following lemma,
\begin{lemma}\label{lem: bdd-singleton-decay}
There exists a constant $c_0>0$, \[\Pr_{g'\sim G'}[\text{every edge in } S \text{ is forced by } \calF ] \leq d^{-c_0d (|S|-\frac{2q|E(\al)|-m}{2d})/\log d} =d^{-c_0d|S|/\log d}\cdot d^{c_0(2q|E(\al)-m)/d}
 \]
\end{lemma}

\subsection{Edge-value assignment} \label{sec:edge-val-assignment}

We now proceed to unpack the value of a given walk, and let $S$ be the set of singleton edges (that only appears once in $\calP$), and $T$ the rest of the edges. For the following section, let $\chi_e$ be the $p$-biased Fourier character, and let $g_e$ be the corresponding $ \{0,1\}$ indicator variable for whether an edge appears in the random graph sample $G$, and notice up to $(1+\frac{1}{n}) $ factor, \[ 
	\chi_e \approx (g_e-\frac{d}{n})\cdot \sqrt{\frac{n}{d}}
	\]
	For our convenience, we will consider $J\subseteq S$ to be the set of singleton edges that appear in the random graph sample, $\jbar\coloneqq S\setminus J$ those that do not; similarly,  let $W\subseteq T$ be the set of edges that appear in the random graph sample, and define $\wbar$ analogously. 
	\begin{proposition}\label{prop: walk-value-factor} 
	For some constant $c_0>0$, let $m$ be number of surprise visits in the walk that lead to visited vertices in the walk $\calP$, and 
	\begin{align*}
		&\left|\E_{G\sim G_{n,d/n}} \left[\val(P)\cdot \cycle \cdot \bdd  \right]\right|\\ &\leq 2^T \sqrt{\frac{n}{d}}^{S+\mul(T)}\max_{W\subseteq T}  (\frac{d}{n})^{|S|+|T|+(\mul(\wbar)-|\wbar|) }\max\left(C(\log n)^2n^{-0.7(\frac{|S|}{\ell}-\cl(S\cup W) ) }, 
    d^{-c_0\cdot d|S|/\log d + c_0(2q\cdot E-m) }    \right)
	\end{align*}	
	\end{proposition}
	\begin{proof}
		This follows by direct expansion,  \begin{align*}
	&	\left|\E_{G\sim G_{n,d/n}} \left[\val(P)\cdot \cycle \cdot \bdd  \right]\right|\\&=
		\left| \E_G\left[\prod_{e\in S\cup T} \chi_e^{\mul(e)}\cdot \cycle \cdot\bdd \right]  \right|
\\ &= \left|\E_G\left[\left((\chi_e-\frac{d}{n})\cdot \sqrt{\frac{n}{d}}\right)^{\mul(e)}\cdot \cycle\cdot \bdd  \right]  \right| \\
&=\left(\sqrt{\frac{n}{d}}\right)^{|S|+\mul(T)}  \left|\E\left[\sum_{W\subseteq T, J\subseteq S} (1-\frac{d}{n})^{\mul(W)+\wbar}(\frac{d}{n})^{W+\mul(\wbar)}(-1)^{\mul(\wbar)+\jbar } (1-\frac{d}{n})^{J+\jbar}(\frac{d}{n})^{J+\jbar}\cdot \cycle \cdot \bdd |G_{W\cup J}=1, G_{\wbar\cup\jbar}=0 \right]  \right|\\
 &\leq \sqrt{\frac{n}{d}}^{S+\mul(T)} \sum_{W\subseteq T} (\frac{d}{n})^{W+\mul(\wbar)} \left|\E\left[ \sum_{J\subseteq S}(1-\frac{d}{n})^{J+\jbar}(\frac{d}{n})^{J+\jbar}(-1)^{\jbar}\cdot\cycle\cdot \bdd |G_{W\cup J}=1, G_{\wbar\cup\jbar}=0  \right] \right|  \\
    &\leq   2^{T}\sqrt{\frac{n}{d}}^{S+\mul(T)}\max_{W\subseteq T}(\frac{d}{n})^{W+\mul(\wbar)} \left|\E\left[ \sum_{J\subseteq S}(1-\frac{d}{n})^{|S|}(\frac{d}{n})^{|S|} (-1)^{\jbar}\cdot \cycle|G_{W\cup J}=1, G_{\wbar\cup\jbar}=0  \right] \right| \\
    &= 2^{T}\sqrt{\frac{n}{d}}^{S+\mul(T)}\max_{W\subseteq T}  (\frac{d}{n})^{W+\mul(\wbar)}(\frac{d}{n})^{|S|} \left| \sum_{J\subseteq S}(-1)^{\jbar}\Pr_G[\calE(G)\land \calF(G) |G_{W\cup J}=1, G_{\wbar\cup\jbar}=0] \right|\\
    &\leq 2^T \sqrt{\frac{n}{d}}^{S+\mul(T)}\max_{W\subseteq T}  (\frac{d}{n})^{|S|+|T|+(\mul(\wbar)-|\wbar|) }\max\left(Cn^\delta n^{-0.7(\frac{|S|}{\ell}-\cl(S\cup W) ) }, 
    d^{-c_0\cdot d|S|/\log d + c_0(2q\cdot E-m)/d }    \right)
			\end{align*}
				where we use singleton-decay from the appendix.
	\end{proof}
	Observing the guarantee from the proposition, we notice each edge has the following contribution, \begin{enumerate}
	\item if it appears exactly twice in the walk, it contributes a value $1$ as it contributes $2$ factors to $\mul(T)$ and $1$ factor to $|T|$, altogether giving us $(\sqrt{\frac{n}{d}})^2\cdot \frac{d}{n}=1$  ;
	\item if it appears only once in the walk, it contributes a factor of $1$ to $S$, picking up a value of $\sqrt{\frac{n}{d}}\cdot \frac{d}{n}\cdot\singdecay$;
	\item it it appears $t>2$ times in the walk, it contributes a factor of $\sqrt{\frac{n}{d}}^{t}\cdot \frac{d}{n}\leq  \left( \sqrt{\frac{n}{d}} \right)^{t-2} $; 
\end{enumerate} 
\subsection{From trace to concentration}
\begin{proof}[Proof to \cref{thm:norm-theorem}]
We now show our the concentration component of main theorem in \cref{thm:norm-theorem} from the corresponding block-value bounds in \cref{cor:block-value-single-shape} and Lemma~\cref{lem:block-value-group-main-lemma}, which then complete the proof to our main theorem of norm bounds by conditioning the probability that the graph sample has no nearby $2$-cycle which happens with probability $1-o_n(1)$.
\end{proof}

\begin{lemma}
For any $\eps>0$, with probability at least $n^{-100}$, for $G\sim G_{n,d/n}$ with high-degree vertices removed and conditioning on no nearby 2-cycle,  graphical matrix of any shape $\al$ has its norm bounded by $B_q(\al)$, i.e., \[ 
\Pr[ \text{exists some shape s.t. } \|M_{\al}\| > (1+\eps) B_q(\al) | G \text{ has bounded degree and no nearby 2-cycle } ] \leq n^{-100}
\]
\end{lemma}
\begin{proof}
For any $t >1$ and $q\in \N$, \begin{align*}
	&\Pr[\|M_\al\|> t B_q(\al)| G \text{ has bounded degree and no nearby 2-cycle } ]\\&\leq \Pr[\Tr \left(M_\al M_\al^T \right)^q> (t \cdot B_q(\al))^{2q}| G \text{ has bounded degree and no nearby 2-cycle} ] \\
	&\leq \frac{\E\left[\Tr\left(M_\al M_\al^T)^q \right) |\bdd\cdot \cycle \right] }{ (t\cdot B_q(\al))^{2q}}\\
	&\leq \frac{n^{|U_\alpha|}B_q(\al)^{2q} }{(t \cdot B_q(\al))^{2q} }  \quad\text{(by block-value bound in Corollary~\ref{cor:block-value-single-shape})}\\
	&\leq \left(\frac{1}{t}\right)^{2q}
\end{align*}
In order to obtain a tail bound that scales with $O(n^{-100})$, it suffices to take $t \leq O( n^{-100})^{1/2q } = 1+\eps $ as $q=\Omega(\dsos\log^2 n).$ 
\end{proof}

    \section{Deferred details for PSD analysis}
% \begin{claim}[Large independent-set]\label{claim:indset-value}  Fix a color $r\in [k]$, 
% \[
% \pE[\sum_{v\in [n]} x_{v,r}] = \Omega(\frac{n}{k} ).
% \]
% \end{claim}
% \begin{proof}
% We show that \[ 
%    \sum_{v\in [n]} \pE [x_{v,i }] - \frac{n}{k} = o(1)
%    \]
%    as for each vertex $\pE [x_{v,i }] $, the trivial shape comes with a coefficient of $\frac{1}{k}$. 
   
%    Following $k$-wise symmetries of the coloring coefficient and connected truncation,  the only contributing term is cycles (with potentially more excess edges), however, by our block-value argument, this contributes at most $O(1) \cdot (\frac{1}{k})^{|V(\al)| }$ (since at each block, such shape is fixed) for each $\al$ with $|U_\al|=|V_\al|=|U_\al\cap V_\al|=1$.
%    \end{proof}

%   To apply the above concentration for our PSDness analysis, observe that we are applying matrix norm bound once for each active-profile, and there are at most $2^{O(\dsos)}$ many active profiles. Since we take $\dsos= O(1)$ in the interested regime, our probabilistic norm bounds can readily afford a union-bound over active-profiles, and this is also the sole union bound needed.

\subsection{Truncation error}
\begin{proof}[Proof to Lemma~\ref{lemma:truncation-main}]
    We first show that a single truncation shape can be bounded, and then extend it to a count in an analogous way to our our charging for middle/intersection shape. Again, we start by identifying what we aim to bound here.

\begin{definition}[Truncation shape] We call a shape $\al$ an intersection shape if $\al= \sigma\circ\tau\circ\sigma'^{T}$ for $\sigma, \sigma'^T$ left/right shapes, and $\tau$ a middle shape, and in addition, it satisfies the following constraints,\begin{enumerate}
    \item $|V(\sigma)|, |V(\tau)|, |V(\sigma')| \leq D_V = \Theta(\dsos\cdot \log n)$;
    \item $|V(\al) | >D_V$
\end{enumerate}
In  other words, it can be obtained from the multiplication of left/middle/right matrices in the approximate decomposition as each piece is within the size constraint locally, however, it does not appear in the final construction $\Lambda$ due to the added sum exceeds the constraint.
\end{definition}
Following our counting from middle/intersection shape that each shape can be specified at a cost of $c^{|V(\al)|}$ for some constant $c>0$, it suffices for us to bound \[ 
\lambda_\al \cdot B_q(\al)\ll 1
\]
for our block-value function $B_q(\al)$ of any $\al$ that is a truncation shape. However, this follows by observing \begin{align*}
B_q(\al)&\leq \max_{S:\text{separator} } \left(\frac{1}{k}\right)^{|V(\al)| -\dsos|}\cdot  \sqrt{n}^{|V(\al)\setminus V(S)|} \cdot (\sqrt{\frac{d}{n}})^{|E(\tau)|}\cdot \left(\sqrt{\frac{n}{d}}\right)^{|E(S)|}
\end{align*}
where we observe that $\frac{|U_\al|+|V_\al|}{2} \leq \dsos$. Moreover, consider the BFS argument from the separator again as in our charging for middle shape, charging each vertex outside the separator by the edge that leads to it gives a factor of (loosely as we ignore the extra $\frac{1}{\dsos^4}$ decay) \[ 
\frac{1}{k}\cdot \sqrt{n} \leq \frac{1}{c_\eta \cdot \sqrt{d}} \cdot \sqrt{\frac{d}{n}}\cdot \sqrt{n} \leq \frac{1}{\ceta}
\]
However, since the separator is at mot size $\dsos$, there are $\Omega(\dsos\log n)$ vertices outside the separator, hence, we pick up a factor of in total \[
\left(\frac{1}{\ceta}\right)^{\Omega(\dsos\log n)} \leq n^{-\Omega(\dsos)}
\]
Combining with the counting from middle shape gives us the desired bound.
\end{proof}
\subsection{Sum of left shapes is well-conditioned}
\begin{proof}[Proof to Lemma~\ref{lemma:singval-main}]
Observe that \[ 
(\sum_{\sigma}\lambda_\sigma  M_\sigma) \cdot (\sum_{\sigma}\lambda_\sigma  M_\sigma) ^T = \sum_{v=0}^{\dsos} (\sum_{\sigma:|V_\sigma|=v}\lambda_\sigma  M_\sigma)(\sum_{\sigma':|V_{\sigma'^T}|=v}\lambda_{\sigma'}  M_{\sigma'})^T 
\]
Thus it suffices for us to find weights $w_j$  such that \[ 
\sum_{v=0}^{\dsos} w_j(\sum_{\sigma:|V_\sigma|=v}\lambda_\sigma  M_\sigma)(\sum_{\sigma':|V_{\sigma'^T}|=v}\lambda_{\sigma'}  M_{\sigma'})^T \succeq \Pi
\,.\]
Since each term is individually PSD, the left-hand side is PSD-dominated by \[ 
\sum_{v=0}^{\dsos} (\sum_{\sigma:|V_\sigma|=v}\lambda_\sigma  M_\sigma)(\sum_{\sigma':|V_{\sigma'^T}|=v}\lambda_{\sigma'}  M_{\sigma'})^T  \preceq (\max_v w_v)(\sum_{\sigma:|V_\sigma|=v}\lambda_\sigma  \cm_\sigma)(\sum_{\sigma':|V_{\sigma'^T}|=v}\lambda_{\sigma'}  M_{\sigma'})^T
\]
which would then yield a lower bound of $\frac{1}{w_j}$. It suffices for us to pick weight $w_v = n^{|V_\sigma|} $ via an identical argument of Corollary E.62 in \cite{jones2023sumofsquares} and verify that $\max_v w_v \leq n^{\dsos}$. This concludes our proof for singular-value lower bounds of the left shapes.
\end{proof}

% \input{kronecker_factorization}
% \bibliographystyle{alpha}
% \bibliography{madhur}

% \bibliographystyle{alphaurl}
% \bibliography{madhur}
\end{document}